\PassOptionsToPackage{unicode}{hyperref}
\PassOptionsToPackage{hyphens}{url}
\PassOptionsToPackage{dvipsnames,svgnames,x11names}{xcolor}
\documentclass[
  12pt]{article}

\usepackage{amsmath, amsthm, amsfonts, amssymb}
\usepackage{graphicx,psfrag,epsf}
\usepackage{enumerate}
\usepackage{enumitem}
\usepackage{natbib}
\usepackage{bbm}
\usepackage{hyperref}
\usepackage{cleveref}
\usepackage{bm}
\usepackage{xcolor}
\usepackage{multirow}
\usepackage{booktabs}
\usepackage{rotating}
\usepackage{url} % not crucial - just used below for the URL 
\usepackage{setspace}
\usepackage{ragged2e}
\usepackage{amsbsy}
\usepackage{mathtools}
\newcommand{\blind}{1}

\def\bX{\bm{X}}

\def\bV{\bm{V}}

\def\bZ{\bm{Z}}

\renewcommand{\vec}[1]{\mathbf{1}}

\newcommand{\boldeta}{\boldsymbol{\eta}}

\providecommand{\rev}[1]{#1}
\providecommand{\revb}[1]{#1}

\providecommand{\bx}{\boldsymbol{x}}

\providecommand{\bc}{\boldsymbol{c}}
\providecommand{\bd}{\boldsymbol{d}}
\providecommand{\bb}{\boldsymbol{b}}
\providecommand{\bnu}{\boldsymbol{\nu}}
\providecommand{\bdeta}{\boldsymbol{\eta}}
\providecommand{\btheta}{\boldsymbol{\theta}}
\providecommand{\bbeta}{\boldsymbol{\beta}}

\providecommand{\bzero}{\boldsymbol{0}}
\providecommand{\bD}{\boldsymbol{D}}
\providecommand{\bO}{\boldsymbol{O}}
\providecommand{\bB}{\boldsymbol{B}}
\providecommand{\bSigma}{\boldsymbol{\Sigma}}
\providecommand{\bbG}{\mathbb{G}}
\providecommand{\bbE}{\mathbb{E}}
\providecommand{\bbP}{\mathbb{P}}

\providecommand{\calH}{\mathcal{H}}
\providecommand{\trans}{\mathsf{T}}

\providecommand{\argmin}{\operatorname*{arg\,min}}

\providecommand{\circled}[1]{\textcircled{\scriptsize #1}}

\providecommand{\bnuhat}{\widehat{\boldsymbol{\nu}}}
\providecommand{\bdetahat}{\widehat{\boldsymbol{\eta}}}

\providecommand{\bthetatilde}{\widetilde{\boldsymbol{\theta}}}
\providecommand{\bSigmatilde}{\widetilde{\boldsymbol{\Sigma}}}
\providecommand{\bM}{\boldsymbol{M}}
\providecommand{\calT}{\mathcal{T}}

\newtheorem{theorem}{Theorem}
\newtheorem{proposition}{Proposition}
\newtheorem{lemma}{Lemma}
\newtheorem{corollary}{Corollary}
\newtheorem{remark}{Remark}

\newtheorem{assumption}{Assumption}[]
\begin{document}

\def\spacingset#1{\renewcommand{\baselinestretch}%
{#1}\small\normalsize} \spacingset{1}

%%%%%%%%%%%%%%%%%%%%%%%%%%%%%%%%%%%%%%%%%%%%%%%%%%%%%%%%%%%%%%%%%%%%%%%%%%%%%%

\if1\blind
{
  \title{\bf Heterogeneous Effects of Continuous Treatments via \rev{Conditional Modified Treatment Policies}}
  \author{Samhita Pal\hspace{.2cm}\\
    Department of Biostatistics, Vanderbilt University Medical Center\\
    and \\
    Jared D. Huling\thanks{Corresponding Author: \href{mailto:huling@umn.edu}{huling@umn.edu}} \\
    Division of Biostatistics and Health Data Science, University of Minnesota}
    \date{}
  \maketitle
} \fi

\if0\blind
{
  \bigskip
  \bigskip
  \bigskip
  \begin{center}
    {\LARGE\bf Heterogeneous Effects of Continuous Treatments via Conditional Modified Treatment Policies}
\end{center}
  \medskip
} \fi

\bigskip
\begin{abstract}
\rev{For continuous treatments such as drug dose or ventilator intensity, a key clinically actionable question is \revb{whether a modest, patient-specific adjustment to the current dose would help or harm, rather than whether to treat at all}. Standard estimands such as average or conditional dose-response functions require positivity across a wide range of doses, an assumption that routinely fails in observational clinical data where protocols tie dosing to patient characteristics. We develop a framework for characterizing heterogeneity in the effects of small shifts (``nudges'') of a continuous treatment. We study two estimands: the conditional nudge effect, the expected outcome change if an individual at dose $a$ with covariates $\bx$ had their dose shifted by $\delta$, and the conditional modified treatment policy (CMTP) effect, which averages the nudge over the observed dose distribution given covariates. Both are identified under weak, shift-specific local positivity and exchangeability conditions. Recasting the $\delta$-shift as a two-arm comparison through a duplicated-data construction, we develop weighting and A-learning estimators based on squared-error and negative log-likelihood losses, introduce augmented versions that reduce variance without changing the target, and establish asymptotic normality of proposed estimators. Simulations support the theory, and an analysis of mechanical ventilation data from MIMIC-III identifies patient profiles predicted to benefit from a modest reduction in mechanical power.}
\end{abstract}

\noindent%
\vfill

\newpage
\spacingset{1.8} % DON'T change the spacing!
\section{Introduction}\label{sec:intro}
In many clinical and policy settings, the key questions are how much a specific amount of a continuous treatment, such as a drug dose or a ventilator setting, affects patient outcomes and whether the same amount benefits all patients equally, not just whether the treatment works at all. In patients with Acute Respiratory Distress Syndrome (ARDS), for example, the mechanical power (MP) of a ventilator has effects on mortality that vary along a continuum as it is changed. Further, the effect of MP on mortality is unlikely to be constant across patients. For example, a patient with severely reduced lung compliance or high vasopressor requirements may respond differently to the same ventilator settings compared with a patient with milder illness. Thus, the population-average dose-response function (ADRF) does not characterize key variation that clinicians need to make individualized decisions. Understanding which patients benefit from what degree of change to MP is thus helpful for making personalized treatment recommendations that improve outcomes without exposing low-risk patients to unnecessary ventilator changes.

Such continuous treatments introduce challenges for heterogeneous treatment effect estimation that go beyond those encountered with binary exposures.  With a binary treatment, the positivity assumption requires only that every patient has a positive probability of receiving either treatment; with a continuous treatment, positivity requires that every patient has positive density at every dose level in the support, a condition that is often implausible when clinical protocols constrain treatment values based on patient characteristics. In ARDS patients, ventilation guidelines explicitly restrict MP settings based on measures of the patient's condition. For example, the PaO$_2$/FiO$_2$ ratio, a pre-treatment covariate characterizing the degree of acute hypoxemia, induces strong confounding by determining both the severity of disease and the range of MP values clinicians will administer \citep{huling2024independence, papazian2019formal}. As a result, some combinations of disease severity and ventilator power are rarely or never observed in practice, and global positivity over the full treatment range fails. Heterogeneity exacerbates this problem, e.g., the same MP may be relatively safe in a compliant lung but concentrates tidal energy in a reduced volume of recruitable lung units in poorly compliant lungs, increasing the risk of ventilator-induced injury \citep{siegel2019ventilator}. %More broadly, for continuous interventions such as drug dosage, exposure intensity, or time-on-treatment, causal inference must characterize not just \rev{the} population ADRF but \rev{also} how effects vary across individuals and doses, while respecting the overlap structure actually present in the data.

Beyond average dose-response curves, individualized care benefits from an understanding of how dose-response relationships and marginal improvements from small shifts vary with patient characteristics $\bX$ and current dose $A$. {To date}, little work has been done at the intersection of heterogeneous treatment effect estimation and continuous treatments. \citet{zhu2024contrastive} study heterogeneous effects of multivariate continuous treatments with high-dimensional covariates through a representation-learning approach. Recently, \citet{shin2024treatment} define estimands such as the contrast surface $\tau_{a_1,a_0}(\bx)=\mathbb{E}\{Y(a_1)-Y(a_0)\mid \bX=\bx\}$ and a multivariate treatment-effect variable-importance measure (MTE-VIM), where $Y(a)$ is the potential outcome under treatment $A = a$. However, these procedures require strong positivity over high-dimensional covariates and a wide treatment range, which is often implausible in realistic settings. Due to failure of positivity in continuous treatments, these contrast surfaces are simply not identified over the full treatment range.

{In this article}, we focus on a novel local-shift estimand designed to {mitigate} the positivity limitations discussed above. Specifically, we ask: for patients with covariates \(\bX=\boldsymbol{x}\), what is the expected change in outcome if treatment is shifted by a clinically modest increment \(\delta\) from the value actually received? Because this policy changes treatment only around its values realized in practice, it requires substantially weaker positivity than the full ADRF or global contrast surfaces in that identification needs positive density only near the observed treatment level. The analysis therefore remains anchored to routinely used, clinically feasible doses. Consequently, recommendations based on our
local-shift estimands involve only modest departures from standard practice, potentially making uptake of recommendations more palatable for practitioners. Related work on incremental propensity-score interventions and other stochastic treatment distributions tied to the observed treatment process has developed causal effects that avoid global positivity, including recent nonparametric work on heterogeneous, positivity-robust contrasts \citep{wen2023intervention, mcclean2024nonparametric}{; our methods instead focus on continuous-treatments and shift the observed dose by a fixed amount $\delta$.} Our approach is also motivated by modified treatment policies for continuous exposures \citep{jiang2025modified, diaz2023nonparametric}, which interpret treatment shifts as targeting a hypothetical post-modification population.
 
This work focuses on two novel estimands: the conditional nudge effect, which characterizes how an individual at dose $a$ with covariates $\bx$ would change if nudged by $\delta$, and the conditional modified treatment policy (CMTP) effect, which averages the nudge effect over the observed treatment distribution given covariates. The resulting estimands support individualized recommendations of the form ``for patients with profile $\bX=\boldsymbol{x}$, a small decrease in MP is expected to improve outcome.'' The CMTP is identified under a weak, local positivity condition, enabling its valid estimation in data where treatment values are driven by clinical protocols. To estimate the conditional nudge effect and CMTP, we use the connection between MTPs and binary treatments shown in \citet{jiang25ventilator} that uses a duplicated-data construction comparing $A$ with $A+\delta$. With this reparameterization we develop loss function based estimators in both a weighting and A-learning setup, with both squared-error and negative-log-likelihood losses \citep{chen2017general} for different outcome types. We introduce augmented versions that leave the target unchanged but reduce variance. We provide large-sample theory for the nudge estimators and obtain CMTP inference either by direct estimation or by plug-in from the estimated nudge surface.

A further practical advantage of our approach is that it does not require estimating the generalized propensity score (GPS), a conditional density that is notoriously difficult to model reliably with high-dimensional covariates and whose inverse can produce extreme, unstable weights \citep{naimi2014constructing, colangelo2020double, huling2024independence}. Instead, the duplicated-data construction requires only a propensity model for the binary arm indicator $\Lambda$, {a standard binary regression problem that avoids the conditional density estimation the GPS requires}. In contrast, most existing work on continuous treatments is built around the GPS $r(a,\bx)=f_{A\mid \bX}(a\mid \bx)$, {which plays the balancing role of the binary propensity score in identifying and estimating the average dose-response function $\mu(a)=\mathbb{E}\{Y(a)\}$ under unconfoundedness and positivity} \citep{robins2000marginal, laan2003unified, diaz2013targeted, hirano2004propensity, imai2004causal, austin2018assessing, zhao2020propensity}{, with subsequent contributions developing doubly robust nonparametric and weighting estimators} \citep{kennedy2017non, ai2021unified}. Remedies for GPS instability include direct inverse-GPS estimation \citep{colangelo2020double}, covariate balancing weights \citep{yiu2018covariate, vegetabile2020nonparametric, fong2018covariate}, and kernel-based independence weights \citep{kallus2019kernel, martinet2020balancing, huling2024independence}, though none circumvent the more fundamental requirement that global positivity over the full treatment support must hold for the ADRF to be identified.

{The remainder of this article is organized as follows.} \Cref{sec:methodology} introduces the estimands and identification assumptions, formulates the nudge and CMTP effects through optimization problems, develops both averaged and direct CMTP estimators, extends the framework to exponential-family outcomes, and \Cref{sec:emp_est} defines the empirical estimators. \Cref{sec:theory} presents asymptotic normality results, \Cref{sec:simu} reports simulations, and \Cref{sec:real data} applies the method to MIMIC-III mechanical ventilation data. Finally, we present our concluding remarks in \Cref{sec:conclusion}.
\section{Methodology}\label{sec:methodology}

% \subsection{Notation}\label{subsec:notation}
% The observable quantities we consider consist of the random triplet $(\bX,A,Y)$, where $\bX\in\calX\subseteq\bbR^p$ is a vector of pre-treatment covariates, $A\in\calA\subseteq\bbR$ is a continuous-valued treatment variable indicating the assigned dose for a unit, and $Y\in\calY\subseteq\bbR$ is an outcome of interest. The variate $(\bX,A,Y)$ has a joint distribution $F_{\bX,A,Y}$. We denote the marginal density of the treatment and covariates as $f_A(a)$ and $f_\bX(\bx)$, respectively,  the conditional density of the treatment given $\bX$ as $f_{A|\bX}(a|\bx)$, and their joint density as $f_{\bX,A}(a,\bx)$. Similarly, corresponding distribution functions are denoted $F_A(a) = \bbP(A\leq a)$, $F_\bX(\bx) = \bbP(\bX\leq \bx)$, $F_{A|\bX}(a|\bx) = \bbP(A\leq a \:|\: \bX = \bx)$, and $F_{\bX,A}(\bx,a) = \bbP(\bX\leq\bx,A\leq a)$. Our observed data consists of $n$ i.i.d. samples $(\bX_i,A_i,Y_i)_{i=1}^n$ from $(\bX,A,Y)$. Note that we drop the subscripts on the density and cumulative distribution functions when there is no ambiguity. 

Let \((\bX,A,Y)\sim F_{\bX,A,Y}\), where
\(\bX\in\mathcal X\subseteq\mathbb R^p\) denotes pre-treatment covariates,
\(A\in\mathcal A\subseteq\mathbb R\) a continuous treatment or dose, and
\(Y\in\mathcal Y\subseteq\mathbb R\) the outcome. We write \(f_A\) and
\(f_{\bX}\) for the marginal densities of \(A\) and \(\bX\),
\(f_{A\mid\bX}(a\mid\bx)\) for the conditional density of \(A\) given \(\bX\),
and \(f_{\bX,A}(\bx,a)\) for their joint density. The corresponding distribution
functions are
\(F_A(a)=\mathbb P(A\le a)\),
\(F_{\bX}(\bx)=\mathbb P(\bX\le\bx)\),
\(F_{A\mid\bX}(a\mid\bx)=\mathbb P(A\le a\mid\bX=\bx)\), and
\(F_{\bX,A}(\bx,a)=\mathbb P(\bX\le\bx,A\le a)\).
The observed data are \(n\) i.i.d. copies
\(\{(\bX_i,A_i,Y_i)\}_{i=1}^n\). We suppress subscripts on densities and
distribution functions when no ambiguity arises.

\subsection{Estimands}\label{subsec:estimands}
We work in the potential-outcomes framework where for each dose level $a\in\mathcal A$, $Y(a)$ denotes the outcome that would be observed were the treatment set to $a$. The causal conditional average dose-response function (CADRF),
$\mu(a,\bx)\equiv \mathbb E\!\left[\,Y(a)\mid \bX=\bx\,\right]\text{ for } a\in\mathcal A,$ has been the main approach to characterizing heterogeneity of effect for continuous-valued treatments in the literature.
The CADRF describes how the conditional mean potential outcome varies with $a$ given $\bX = \bx$. Estimating $\mu(a,\bx)$ nonparametrically is challenging in practice because $\bX$ may be high-dimensional and because identification requires (i) no unmeasured confounding $Y(a)\perp A\mid \bX$ and (ii) local overlap/positivity, informally, that all doses $a$ of interest occur with nonzero probability at covariate value $\bx$. Motivated by these challenges and by policy questions about the effects of small changes to practice, instead of {targeting} the entire CADRF, we focus on two local shift estimands that {are identified under assumptions of differing strength and offer distinct, but related, interpretations}. We define the conditional nudge effect at dose $a$ and covariates $\bx$, which captures the expected change in outcome from nudging the dose $a$ at $\bx$ by a small increment $\delta$ (with $a,a+\delta\in\mathcal A$), as
$\tau_\delta(a,\bx)\;\equiv\;\mathbb E\!\left[\,Y(a+\delta)-Y(a)\mid \bX=\bx\,\right].$
Compared to the CADRF $\mu(a,\bx)$, the local nudge surface $\tau_\delta(a,\bx)$ is often a more tractable estimand. First, it typically has lower effective complexity because it is a {local contrast}{, in which terms depending on $\bx$ alone cancel in the outcome model}. Second, identification only requires overlap locally between $a$ and $a+\delta$ rather than across the entire range of $a$, a condition that is \textit{far} more plausible when $\delta$ is small. Third, we show that the contrast can be estimated directly, without specifying and fitting a full outcome regression model for $\mu(a,\bx)$. Finally, our nudge estimators can admit {doubly robust} strategies so that consistency still holds if either {of the two nuisance components (the arm propensity model or the outcome regression)} is correctly specified.

The policy analog of the conditional nudge effect, the {CMTP effect}, averages that nudge over the natural (observed) distribution of treatment at $\bx$:
{\small\begin{align}\label{CMTP}
    \tau_\delta(\bx)\;\equiv\;\mathbb E\!\left[\,Y(A+\delta)-Y(A)\mid \bX=\bx\,\right]
\;=\;\mathbb E\!\left[\,\tau_\delta(A,\bx)\mid \bX=\bx\,\right].
\end{align}}
Thus, CMTP is the conditional expectation of the nudge surface with respect to the conditional law of $A\mid \bX=\bx$, providing a covariate-specific measure of the consequence of uniformly shifting everyone’s dose by $\delta$ in that stratum. {The first expression in \eqref{CMTP} is the definition; the second equality, which represents the CMTP as the average of the nudge surface, holds under the mean nudge exchangeability condition (Assumption~\ref{ass:exchange}(A) below) but is not needed for identification of $\tau_\delta(\bx)$ itself.} The CMTP effect $\tau_\delta(\bx)$ enjoys several advantages over the conditional nudge surface. First, it is causally identifiable under weaker, shift-specific assumptions, making it more robust to confounding and positivity violations. Second, it {can be obtained} by marginalizing the estimated nudge over \(A\mid\bX\), {even though} its identification requires strictly weaker conditions than the nudge itself. Finally, because $\tau_\delta(\bx)$ is a function of $\bx$ alone, its modeling burden is lower: estimation involves either a direct one-step CMTP loss in $\bx$ or a plug-in average over $A\mid \bX=\bx$. While marginal modified treatment policies allow for shifts that depend on $A$ and $\bX$, they still require the user to specify the form of the shift \textit{a priori}; the CMTP more directly assesses effect heterogeneity in response to a shift and thus {yields a clearer sense} of which covariates may drive differential response.

\subsection{Causal assumptions and identification}\label{subsec:assumptions}
{Having defined the conditional nudge and CMTP effects, we now state the causal assumptions under which they can be identified from the observed data.}

% Notation: ;
% f(a,\bx) is the conditional density/mass of A given \bX=\bx.

\begin{assumption}[Consistency]\label{ass:consistency}
For all $(a,\bx)$ with ${f_{A\mid\bX}(a\mid\bx)}>0$, if $A=a$ then $Y=Y(a)$.
\end{assumption}

\begin{assumption}[Positivity]\label{ass:positivity}
{The shifted dose is feasible: for all $(a,\bx)$,
if $f_{A\mid\bX}(a\mid\bx)>0$ then $f_{A\mid\bX}\big(q(a,\bx)\mid\bx\big)>0$.}
\end{assumption}

\begin{assumption}[Exchangeability]\label{ass:exchange}
For all $(a,\bx)$ with ${f_{A\mid\bX}(a\mid\bx)}>0$ and $q(a,\bx)$ as the shift/policy map (e.g., $q(a,\bx)=a+\delta$), the following hold:
\begin{enumerate}[label=\textnormal{(\Alph*)}, leftmargin=2.25em]
\item \textit{Mean nudge effect exchangeability}:\label{ass:nudge_exchange}  \\$\mathbb{E}\!\left[\,Y\!\big(q(a,\bx)\big)-Y(a)\,\middle|\,\bX=\bx,\,A=a\right]
=
\mathbb{E}\!\left[\,Y\!\big(q(a,\bx)\big)-Y(a)\,\middle|\,\bX=\bx\right].$
% {\small\[

% \]}
\item[\textnormal{(A$'$)}] \textit{Conditional mean policy exchangeability}:\label{ass:cmtp_exchange}  \\$\mathbb{E}\!\left[\,Y\!\big(q(a,\bx)\big)\,\middle|\,\bX=\bx,\,A=a\right]
=
\mathbb{E}\!\left[\,Y\!\big(q(a,\bx)\big)\,\middle|\,\bX=\bx,\,A=q(a,\bx)\right].$
\end{enumerate}
\end{assumption}

Assumption \ref{ass:positivity} is substantially weaker than the positivity assumption required for estimating the causal average dose response function and can be made to hold by design, since $q(a, \bx)$ is a user-specified function. {For the fixed shift $q(a,\bx)=a+\delta$ the assumption cannot hold exactly at the boundary of a bounded dose support; throughout, we therefore restrict attention to the region $\{a : f_{A\mid\bX}(a\mid\bx)>0 \text{ and } f_{A\mid\bX}(a+\delta\mid\bx)>0\}$, although a boundary-respecting shift in the spirit of \citet{haneuse2013estimation} could be used instead.}
We note that Assumption {\ref{ass:exchange}} is substantially weaker than the mean exchangeability assumption required for estimating the causal average dose response function. The typical mean exchangeability assumption immediately implies {Assumption \ref{ass:exchange}} holds; further, if the conditional mean potential outcome function is smooth in $a$, then this assumption becomes more plausible {as} $q(a,\bx)$ becomes closer to $a$. 

\begin{proposition}\label{prop:nudge_outcome_identifiability}
    Under Assumptions \ref{ass:consistency}, \ref{ass:positivity}, {and \ref{ass:exchange}(A) and (A$'$), for every $(a,\bx)$ with $f_{A\mid\bX}(a\mid\bx)>0$,} the conditional nudge effect is identified in terms of the observed data as {\small$\bbE[Y \vert \bX = \bx, A = a + \delta] - \bbE[Y \vert \bX = \bx, A = a] = \tau_\delta(a,\bx)$}.
\end{proposition}

\begin{proposition}\label{prop:cmtp_outcome_identifiability}
	{Let $m_0(a,\bx) := \bbE[Y \mid \bX = \bx, A = a]$ denote the observed outcome regression.} Under Assumptions \ref{ass:consistency}, \ref{ass:positivity} and \ref{ass:exchange}(A$'$), {for $f_{\bX}$-almost every $\bx$,} the CMTP effect is identified in terms of the observed data as
	{\small$\bbE[{m_0}(A+\delta, \bx) \vert \bX = \bx] - \bbE\left[Y|\bX = \bx\right] = \tau_\delta(\bx)$}.
\end{proposition}
{We emphasize that the identification formula in \Cref{prop:cmtp_outcome_identifiability} involves the observed outcome regression $m_0(a,\bx)$ rather than the counterfactual CADRF, which is not identified under our assumptions. Notably, the CMTP effect requires only condition (A$'$) while the nudge surface requires both (A) and (A$'$), making precise the sense in which the CMTP is identified under strictly weaker conditions.} These propositions provide the identification results that justify our estimation procedures{, which we develop in the remainder of this section}. The first shows that the conditional nudge effect can be recovered from observed conditional mean differences, whereas the second establishes that the CMTP effect is identified by {contrasting the average of the outcome regression evaluated at the shifted dose with the average observed outcome}.

\subsection{Approach 1: Estimating the CMTP via the conditional nudge}\label{subsec:method}
In this section, we introduce a data duplication technique \citep{diaz2023nonparametric} that allows us to directly estimate both the conditional nudge effect and {the} CMTP without the need {to model} the entire outcome regression surface. We create a duplicated treatment $A_\Lambda$ that is, with marginal probability $1/2$, either the observed dose $A$ or an inverse-shifted dose $A-\delta$ (the inverse of the policy shift we ultimately care about). This requires the shift map $a\mapsto q(a,\bx)$ to be invertible (e.g., $q(a,\bx)=a+\delta$). We then define the arm label
{\small\[
\Lambda=\tfrac12 \ \text{if } A_\Lambda=A-\delta \ (\text{shifted arm}), 
\qquad
\Lambda=-\tfrac12 \ \text{if } A_\Lambda=A \ (\text{original arm}){,}
\]}
so that, conditional on $(\bX=\bx,A_\Lambda=a)$, the event $\Lambda=1/2$ corresponds to units whose original dose was $a+\delta$, while $\Lambda=-1/2$ corresponds to units with original dose $a$. Let $\pi_\Lambda(\bx,a):=\mathbb P\big(\Lambda=\tfrac12\mid \bX=\bx,A_\Lambda=a\big)$ denote the conditional probability {that an observation in the duplicated data belongs to the shifted arm, given covariates and the duplicated treatment value}. Then, a direct calculation using $\,\mathbb P(\Lambda=1/2)=\mathbb P(\Lambda=-1/2)=1/2\,$ gives $\pi_\Lambda(\bx,a)={g_{A\mid \bX}(a+\delta\mid \bx)}/\big\{{g_{A\mid \bX}(a+\delta\mid \bx)+g_{A\mid \bX}(a\mid \bx)}\big\},$
% {\small\begin{equation}\label{eq:piLambda}
% \end{equation}}
where $g_{A\mid \bX}$ is the conditional density (or mass) of $A$ given $\bX$. 

\subsubsection{Weighting estimator for the nudge effect}
We target $\tau_\delta(a,\bx)$ by minimizing a {balanced} squared-loss on the duplicated sample. For any candidate contrast $f(a,\bx)$ define the conditional population risk
{\small\begin{equation}\label{eq:Wloss}
\ell_W(f;a,\bx)=\mathbb E\!\left[w(a,\bx,\Lambda)\,\big\{Y-\Lambda\,f(a,\bx)\big\}^2\;\middle|\;\bX=\bx,\ A_\Lambda=a\right],
\end{equation}}
where expectations are with respect to the joint law of $(Y,A_\Lambda,\bX,\Lambda)$. Identification requires only that the weights satisfy the arm-balancing condition
{\small\begin{equation}\label{eqn:weight_function_condition}
w(a,\bx,\tfrac12)\,\mathbb P\!\left(\Lambda=\tfrac12\mid \bX=\bx,A_\Lambda=a\right)
=
w(a,\bx,-\tfrac12)\,\mathbb P\!\left(\Lambda=-\tfrac12\mid \bX=\bx,A_\Lambda=a\right),
\end{equation}}
in which case{, for strictly positive weights,} the unique minimizer is $f^*(a,\bx)=\tau_\delta(a,\bx)$ as shown later in \Cref{prop:nudge_weighting_loss_solution}. Several convenient choices satisfy \eqref{eqn:weight_function_condition}. First, \emph{IPW:} $w(a,\bx,1/2)=\pi_\Lambda(\bx,a)^{-1},$ and $\ w(a,\bx,-1/2)=\{1-\pi_\Lambda(\bx,a)\}^{-1}$. Second, \emph{Density-ratio:} either $w(a,\bx,1/2)=\{1-\pi_\Lambda(\bx,a)\}/\pi_\Lambda(\bx,a),\ w(a,\bx,-1/2)=1$, or the symmetric alternative with the roles reversed. These can be computed directly from $g_{A\mid \bX}$ using $({1-\pi_\Lambda(\bx,a)})/{\pi_\Lambda(\bx,a)}={g_{A\mid \bX}(a\mid \bx)}/{g_{A\mid \bX}(a+\delta\mid \bx)}.$ Third, \emph{Overlap:} $w(a,\bx,\tfrac12)=1-\pi_\Lambda(\bx,a),\ \ w(a,\bx,-\tfrac12)=\pi_\Lambda(\bx,a)$. A advantage of the IPW and overlap weights is that they {require modeling only} a conditional probability rather than a conditional density, which is often substantially more challenging to estimate reliably \citep{huling2024independence}. To reduce variance without changing the target, we use an augmented version that subtracts any function $m(a,\bx)$ of $(A_\Lambda,\bX)$:
{\small\begin{equation}\label{eq:Wloss-aug}
\ell_{W,\mathrm{aug}}(f;a,\bx)
=\mathbb E\!\left[w(a,\bx,\Lambda)\,\big\{Y-m(a,\bx)-\Lambda\,f(a,\bx)\big\}^2\;\middle|\;\bX=\bx,\ A_\Lambda=a\right].
\end{equation}}
The conditional losses above characterize the pointwise population target at a
fixed value of \((A_\Lambda,\bX)=(a,\bx)\); the empirical estimator minimizes the
corresponding marginal risk $\ell_{W,\mathrm{aug}}(f)
:=
\mathbb E[
\ell_{W,\mathrm{aug}}(f;A_\Lambda,\bX)]$
{(and analogously $\ell_W(f)$)} over a working class
\(f(a,\bx)=\btheta^\top \bb(a,\bx)\){, made precise in \Cref{sec:emp_est}, and thus estimates the projection of the pointwise nudge contrast onto the basis class}. This direct weighted-learning
formulation is motivated by the contrast-learning/subgroup-identification
framework of \citet{chen2017general}. 
% In practice, the empirical versions of \eqref{eq:Wloss}--\eqref{eq:Wloss-aug} are minimized over a basis $f(a,\bx)=\btheta^\top b(a,\bx)$, with $\pi_\Lambda$ (or $g_{A\mid \bX}$) and $m$ estimated flexibly; any weights satisfying \eqref{eqn:weight_function_condition} yield the same population minimizer. 
The following proposition formalizes the key identification step behind our weighting approach. Under the duplicated-arm construction \((A_\Lambda,\Lambda)\) and any choice of weights \(w(a,\bx,\Lambda)\) that satisfy the arm-balancing condition \eqref{eqn:weight_function_condition}, minimizing the conditional squared loss {\eqref{eq:Wloss} (equivalently, its marginal average $\ell_W(f)$)} recovers the unique contrast that equates the two arm-specific conditional means at \((\bx,a)\). Consequently, the population minimizer is the conditional nudge effect \(\tau_\delta(a,\bx)\).

\begin{proposition}\label{prop:nudge_weighting_loss_solution}
	{Suppose the weights are strictly positive and satisfy \eqref{eqn:weight_function_condition}, and that $f_{A\mid\bX}(a\mid\bx)>0$, so that $0<\pi_\Lambda(\bx,a)<1$ by Assumption~\ref{ass:positivity}. Then the minimizers $f_W^*$ of $\ell_W(f)$ and $f_{W,\mathrm{aug}}^*$ of $\ell_{W,\mathrm{aug}}(f)$ over measurable functions are almost everywhere unique and satisfy}
	{\small$f^*_W(a,\bx) = f_{W,{\mathrm{aug}}}^*(a,\bx) = \bbE[Y \vert \bX = \bx, A = a + \delta] - \bbE[Y \vert \bX = \bx, A = a]{,}$} {which equals $\tau_\delta(a,\bx)$ under the conditions of \Cref{prop:nudge_outcome_identifiability}.}
\end{proposition}
{\Cref{prop:nudge_weighting_loss_solution} demonstrates that} the conditional loss identifies the pointwise nudge effect. When the loss
is averaged over the marginal distribution of \((A_\Lambda,\bX)\) and minimized
over a finite-dimensional working model, the resulting estimator targets the
corresponding projection of this pointwise effect. Moreover, the augmentation function \(m(a,\bx)\) is included for precision rather than
identification. As \(m\) depends only on \((A_\Lambda,\bX)\), it is
subtracted equally from the two duplicated arms at a fixed value of
\((A_\Lambda,\bX)\){; since the balancing condition \eqref{eqn:weight_function_condition} implies $\mathbb E[w(a,\bx,\Lambda)\Lambda \mid \bX=\bx, A_\Lambda=a]=0$, the subtraction} does not alter the population contrast identified
by the loss. A natural choice is a working outcome regression,
$m(a,\bx)\approx \mathbb E[Y\mid A=a,\bX=\bx],$
or equivalently in the duplicated-arm notation, $m(a,\bx)\approx \mathbb E[Y\mid A_\Lambda=a,\bX=\bx].$
This removes outcome variation explained by the treatment level and covariates
before estimating the nudge/CMTP contrast, and can therefore improve efficiency. In
applications, \(m\) may be estimated using a parametric regression, splines, or
flexible machine-learning methods. Correct specification of \(m\) is not needed
for identification of the nudge/CMTP effects; rather, the quality of \(m\) affects the
variance of the resulting estimator. The additional first-order variability due
to estimating \(m\) is incorporated in the augmented A-learning asymptotic result
in Section~\ref{sec:theory}.

\subsubsection{A-learning estimator for the nudge effect}
{An alternative to the weighting construction, in the spirit of A-learning, removes main-effect variation by centering the arm indicator rather than by reweighting the loss. To this end, we} define the {recoded} arm indicator $\widetilde\Lambda:=\Lambda+\tfrac12\in\{0,1\}$ and let
{\small\begin{align}\label{A-learn_nudge}
    \ell_A(f;a,\bx)
=\mathbb E\!\left[\big\{Y-(\widetilde\Lambda-\pi_\Lambda(\bx,a))\,f(a,\bx)\big\}^2\;\middle|\;\bX=\bx,\ A_\Lambda=a\right].
\end{align}}
Minimizing $\ell_A$ over $f$ yields the same target $f^*(a,\bx)=\tau_\delta(a,\bx)$ (the centering removes main-effect contamination). This is formalized in \Cref{prop:nudge_alearning_loss_solution}. As with weighting, we introduce an {augmented} A-learning loss
{\small\begin{align}\label{aug-A-learn_nudge}
    \ell_{A,\mathrm{aug}}(f;a,\bx)
=\mathbb E\!\left[\big\{Y-m(a,\bx)-(\widetilde\Lambda-\pi_\Lambda(\bx,a))\,f(a,\bx)\big\}^2\;\middle|\;\bX=\bx,\ A_\Lambda=a\right],
\end{align}}
which preserves the target but improves efficiency by removing prognostic variation through $m(a,\bx)$. {The corresponding marginal losses are $\ell_A(f):=\mathbb E[\ell_A(f;A_\Lambda,\bX)]$ and $\ell_{A,\mathrm{aug}}(f):=\mathbb E[\ell_{A,\mathrm{aug}}(f;A_\Lambda,\bX)]$.}
As in Proposition~\ref{prop:nudge_weighting_loss_solution}, the A-learning objective regresses the outcome on the centered treatment indicator
\(\widetilde{\Lambda}-\pi_\Lambda(\bx,a)\), interacted with a basis in \((a,\bx)\), and identifies the same contrast. Centering removes main-effect variation, so the population first-order condition isolates the difference between the two treatment levels at \((a,\bx)\) so that the minimizer of \(\ell_A(f)\) is \(\tau_\delta(a,\bx)\). Subtracting \textit{any} outcome regression \(m(a,\bx)\) in the augmented loss \(\ell_{A,\mathrm{aug}}(f)\) has no effect on bias, as augmentation only partials out outcome variation as a function of covariates and thus reduces variance without affecting identification.

\begin{proposition}\label{prop:nudge_alearning_loss_solution}
	{Suppose $f_{A\mid\bX}(a\mid\bx)>0$, so that $0<\pi_\Lambda(\bx,a)<1$ by Assumption~\ref{ass:positivity}. Then the minimizers $f_A^*$ of $\ell_A(f)$ and $f_{A,\mathrm{aug}}^*$ of $\ell_{A,\mathrm{aug}}(f)$ over measurable functions are almost everywhere unique and satisfy}
	{\small$f^*_A(a,\bx) = f^*_{A,{\mathrm{aug}}}(a,\bx) = \bbE[Y \vert \bX = \bx, A = a + \delta] - \bbE[Y \vert \bX = \bx, A = a]{,}$} {which equals $\tau_\delta(a,\bx)$ under the conditions of \Cref{prop:nudge_outcome_identifiability}.}
\end{proposition}
Overall, this duplicated-arm formulation is simple to implement, accommodates a wide choice of balancing weights, and supports variance-reducing augmentation, while directly targeting the localized causal contrast $\tau_\delta(a,\bx)$ that underpins the CMTP via $\tau_\delta(\bx)$.

\subsubsection{CMTP estimation by averaging over the conditional nudge}
Both the weighting and A-learning estimators target the local contrast{, that is, the conditional nudge effect $\tau_\delta(a,\bx)$ at dose $a$ and covariates $\bx$ under a shift of size $\delta$}. The CMTP at covariate
level \(\bx\) averages this effect over the conditional distribution of the
observed dose given \(\bX=\bx\). If \(g_{A\mid\bX}(a\mid\bx)\) denotes the
corresponding conditional density or mass function, then
{\small
\(
\tau_\delta(\bx)
=
\mathbb{E}\!\left[Y(A+\delta)-Y(A)\mid\bX=\bx\right]
=
\int \tau_\delta(a,\bx)\,
g_{A\mid\bX}(a\mid\bx)\,\mathrm{d}a,
\)
}
with the integral replaced by a sum for discrete \(A\). Thus, once
\(\tau_\delta(a,\bx)\) is identified and estimated, the CMTP follows by
averaging the conditional nudge over the exposure distribution at \(\bx\). In practice, we represent the conditional nudge parametrically, for example
\(f(a,\bx)=\btheta^\top\bb(a,\bx)\), and obtain a plug-in CMTP by projecting the
fitted values \(\hat f(A_i,\bX_i)\) onto a lower-dimensional basis \(c(\bx)\),
writing \(\tau_\delta(\bx)\approx\bbeta^\top c(\bx)\){; the details are given in \Cref{sec:emp_est}}. {While this plug-in approach is straightforward, it requires an initial fit of the entire nudge surface; in the next subsection we develop estimators that target the CMTP directly.}

% \subsection{Exponential family setup and logistic example}

% \subsubsection*{Weighted loss with NLL (logistic case)}

% \subsubsection*{Logistic A-learning loss with NLL}

\subsection{Approach 2: Direct estimation of the CMTP}\label{sec:direct_nudge}
Here, we propose a direct estimation method for the CMTP that does not require
first estimating the local nudge effect \(\tau_\delta(a,\bx)\) and then
averaging over the conditional distribution of \(A\mid \bX=\bx\). Instead, the
method uses the duplicated-arm construction to form an A-learning loss whose
contrast depends only on \(\bx\), so that the resulting estimator targets
\(\tau_\delta(\bx)\) directly. {A disadvantage of obtaining the CMTP from an initial nudge fit is the need to model the conditional nudge effect itself, a statistical task of higher complexity than modeling the CMTP alone.} {The construction exploits only \Cref{ass:exchange}(A$'$), mirroring the weaker identification requirements of the CMTP established in \Cref{prop:cmtp_outcome_identifiability}.}

We work with a squared-error A-learning loss for directly estimating the CMTP in the main {text} and defer a weighting counterpart to Supplement A.1. The direct A-learning loss is similar to \eqref{A-learn_nudge}, except that we no longer model a local contrast at a fixed dose \(a\), and instead let the contrast depend only on the covariates \(\bx\), integrating over the mixture distribution of \(A_\Lambda\) given \(\bX=\bx\). Concretely, recall that
\eqref{A-learn_nudge} targets the local nudge \(\tau_\delta(a,\bx)\). For the CMTP, we instead consider functions \(f:\mathcal X\to\mathbb R\) and define the direct A-learning CMTP loss $\ell_A^{\mathrm{CMTP}}(f;\bX)
=
\mathbb E\!\left[
  \big\{Y-g(\bX,A_\Lambda)\,f(\bX)\big\}^2
  \;\middle|\;{\bX}
\right],$
where \(g(\bX,A_\Lambda):=\widetilde\Lambda-\pi_\Lambda(\bX,A_\Lambda)\) has conditional mean zero given \((\bX,A_\Lambda)\). The corresponding marginal loss is simply $\ell_A^{\mathrm{CMTP}}(f) = \mathbb E[\ell_A^{\mathrm{CMTP}}(f;\bX)].$ This construction preserves the same centering idea as in \eqref{A-learn_nudge}, but now targets the aggregate effect \(\tau_\delta(\bx)=\mathbb E\{Y(A+\delta)-Y(A)\mid \bX=\bx\}\) rather than the local contrast at a fixed \(a\). We define $h(\bx,a) := \pi_\Lambda(\bx,a)\big\{1-\pi_\Lambda(\bx,a)\big\},$
so that $g(\bx,A_\Lambda)$ has conditional mean zero and conditional variance $h(\bx,A_\Lambda)$ given $(\bX,A_\Lambda)$. Under this unstabilized direct CMTP loss for measurable $f:\mathcal X\to\mathbb R$, the {first-order condition (FOC)} is $\mathbb{E}\Big[g(\bX,A_\Lambda)\,\big\{Y - g(\bX,A_\Lambda) f(\bX)\big\} \,\Big|\, {\bX}\Big] = 0,$
% {\small\[
% \]}
which yields the pointwise minimizer $f^*(\bx)
=
{\mathbb{E}\big[g(\bx,A_\Lambda)Y \mid \bX = \bx\big]}/
     {\mathbb{E}\big[g^2(\bx,A_\Lambda) \mid \bX = \bx\big]}
\equiv
{N(\bx)}/{D(\bx)}$.

\begin{lemma}\label{lem:NxDx}
{Let $m_0(a,\bx)$ be the outcome regression in \Cref{prop:cmtp_outcome_identifiability}.} For every \(\bx\),  $N(\bx)
=
\mathbb{E}\big[
h(\bx,A_\Lambda)\,\{{m_0}(A_\Lambda+\delta,\bx)-{m_0}(A_\Lambda,\bx)\}
\mid \bX=\bx
\big]$, and $D(\bx)
=
\mathbb{E}\big[h(\bx,A_\Lambda)\mid \bX=\bx\big].$
\end{lemma}
{\begin{proof}
See the first part of the proof of Proposition~\ref{direct_CMTP_AL} in the Supplementary Material%, where $N(\bx)$ and $D(\bx)$ are evaluated by conditioning on $(A_\Lambda,\bX)$.
\end{proof}}
\noindent From \Cref{lem:NxDx}, we see that the basic direct A-learning loss recovers a \emph{weighted} CMTP, with weights $h(\bx,A_\Lambda)$, rather than the unweighted CMTP $\tau_\delta(\bx)$. {It is important to recognize that the expectations in \Cref{lem:NxDx} are taken with respect to the duplicated-arm mixture law of $A_\Lambda$ given $\bX=\bx$, whose conditional density is $f_{A_\Lambda\mid\bX}(a\mid\bx)=\tfrac12\{g_{A\mid\bX}(a\mid\bx)+g_{A\mid\bX}(a+\delta\mid\bx)\}$, whereas the CMTP averages over the natural law of $A$ given $\bX=\bx$. Expressed against the natural law (the measure defining the CMTP), the implicit weight is $\pi_\Lambda(\bx,a)$.
%, owing to the identity $h(\bx,a)\,f_{A_\Lambda\mid\bX}(a\mid\bx)=\tfrac12\,\pi_\Lambda(\bx,a)\,g_{A\mid\bX}(a\mid\bx)$. 
To eliminate the weighting entirely and directly target $\tau_\delta(\bx)$, we therefore stabilize the conditional loss by dividing by $\pi_\Lambda(\bX,A_\Lambda)$:}
{\small\begin{align}\label{dir-A-learning}
    \ell_A^{\mathrm{CMTP}}(f;\bX)
=
\mathbb E\Bigg[
\frac{\big\{Y - g(\bX,A_\Lambda)\,f(\bX)\big\}^2}
     {{\pi_\Lambda(\bX,A_\Lambda)}}
\,\Bigg|\, \bX\Bigg],
\end{align}}
and, analogously, define the augmented version
{\small\begin{align}\label{dir-A-learning_aug}
    \ell_{A,\mathrm{aug}}^{\mathrm{CMTP}}(f;\bX)
=
\mathbb E\Bigg[
\frac{\big\{Y - m(A_\Lambda,\bX) - g(\bX,A_\Lambda)\,f(\bX)\big\}^2}
     {{\pi_\Lambda(\bX,A_\Lambda)}}
\,\Bigg|\, \bX\Bigg],
\end{align}}
for an arbitrary outcome regression $m(A_\Lambda,\bX)$. The following result shows that both stabilized losses identify the CMTP $\tau_\delta(\bx)$, and that augmentation by $m(A_\Lambda,\bX)$ does not change the target. In particular, the unique pointwise minimizers of both stabilized A-learning CMTP losses recover the CMTP $\tau_\delta(\bx)$. Because the marginal risk is the expectation of the conditional risk over \(\bX\), namely
\(\ell(f)=\mathbb E\{\ell(f;\bX)\}\), any measurable function that minimizes the conditional loss pointwise also minimizes the marginal loss; conversely, any marginal minimizer agrees with the pointwise minimizer almost surely.

\begin{proposition}\label{direct_CMTP_AL}
{Suppose Assumptions~\ref{ass:consistency}, \ref{ass:positivity}, and \ref{ass:exchange}(A$'$) hold. Under the stabilized direct A-learning CMTP loss \eqref{dir-A-learning} and its augmented counterpart \eqref{dir-A-learning_aug} for an arbitrary function $m(A_\Lambda,\bX)$, let $f_A^*$ and $f_{A,\mathrm{aug}}^*$ denote any minimizers of $\ell_{A}^{\mathrm{CMTP}}(f)$ and $\ell_{A,\mathrm{aug}}^{\mathrm{CMTP}}(f)$, respectively. Then } $f_A^*(\bx) = f_{A,\mathrm{aug}}^*(\bx)
=
{\bbE\big[
m_0(A+\delta,\bx)-m_0(A,\bx)
\mid \bX=\bx
\big]
=
\tau_\delta(\bx),}$
{where the expectation is over the natural conditional law of $A$ given $\bX=\bx$ and the final equality follows from \Cref{prop:cmtp_outcome_identifiability}.}
\end{proposition}
\noindent \Cref{direct_CMTP_AL} shows that stabilizing the direct A-learning loss {by $\pi_\Lambda$} removes the implicit {mixture-law} weighting from \Cref{lem:NxDx}, so that the resulting population minimizer targets the unweighted CMTP itself, with or without outcome-regression augmentation. At the population level, if \(f^*(\bx)=\bbeta^\top c(\bx)\), the stabilized
direct A-learning coefficient targets the {unweighted} \(L_2(P_{\bX})\) projection of
the CMTP onto the span of \(c\): $\bbeta^{\mathrm{dir},0}_{A}
    =
    \arg\min_{\bbeta}
    E\left[
    \{\tau_\delta(\bX)-\bbeta^\top \bc(\bX)\}^2
    \right],$
{since the stabilized score has constant conditional second moment, $\bbE[g^2(\bX,A_\Lambda)/\pi_\Lambda(\bX,A_\Lambda)\mid \bX]=1/2$.}
If the working model is correctly specified, so that
\(\tau_\delta(\bx)=\bbeta_0^\top c(\bx)\), then
\(\bbeta^{\mathrm{dir},0}_{A}=\bbeta_0\). {Our development thus far has relied on squared-error losses. The duplicated-arm construction extends naturally to canonical exponential-family losses for binary and other non-continuous outcomes. We next develop the link-scale A-learning nudge loss and its identification result, and use the resulting logistic-loss estimators in Sections~\ref{sec:simu} and \ref{sec:real data}.}

\subsection{Exponential family setup and logistic example}\label{subsec:expfam-logistic}

We extend the duplicated-arm construction and the weighting/A-learning losses to canonical exponential-family models, focusing on logistic loss for binary outcomes.
 Throughout this subsection we employ the \textit{working model} {in which}, conditional on $(A,\bX)$, the outcome $Y$ follows a one-parameter exponential family with natural parameter ${\eta}$ and negative log-likelihood (NLL) $M(y,{\eta}) = -y\,{\eta} + b({\eta}),$
so that
\(
\mu(a,\bx) := \mathbb{E}[Y \mid A=a,\bX=\bx],
\;
{\eta} = {k}\{\mu(a,\bx)\},
\;
b'({\eta}) = \mu({\eta}),
\)
where ${k}$ is the canonical link and $\mu({\eta})$ denotes the mean as a function of the natural parameter. {Note that in this subsection $\mu(a,\bx)$ denotes the observed outcome regression, written $m_0(a,\bx)$ elsewhere in the paper, rather than the counterfactual CADRF of Section~\ref{subsec:estimands}; the two coincide under the conditions of \Cref{prop:nudge_outcome_identifiability}.} In the logistic example with $Y\in\{0,1\}$, we have ${\eta} = \log\left({\mu}/({1-\mu})\right),$ $b({\eta}) = \log\{1+\exp({\eta})\},$ $\mu({\eta}) = {e^{\eta}}/({1+e^{\eta}}).$
Here ${k}$ is the logit link and $\mu(\cdot)$ is the expit function. In what follows we work directly with the NLL $M$ instead of squared loss, and interpret the resulting targets on the \emph{link scale}, mapping back to the mean scale via $\mu(\cdot)$ when needed. In this setting the {mean-scale nudge effect} at a given dose-covariate pair \((a,\bx)\) remains {$\tau_\delta(a,\bx)
=
\mu(a+\delta,\bx)-\mu(a,\bx)$ under the conditions of \Cref{prop:nudge_outcome_identifiability},}
so that, under a canonical link, the nudge is fully determined by the change in the natural parameter \({\eta}\) induced by shifting the dose from \(a\) to \(a+\delta\). {For} example, in the logistic case this evaluates to $\tau_\delta(a,\bx)
= \mu\bigl(\eta(a+\delta,\bx)\bigr)-\mu\bigl(\eta(a,\bx)\bigr)
= \operatorname{expit}\!\bigl(\eta(a+\delta,\bx)\bigr)
-\operatorname{expit}\!\bigl(\eta(a,\bx)\bigr),
\text{ where }
\operatorname{expit}(u)\equiv {e^u}/{1+e^u},$
 a nonlinear function of both \(a\) and \(\bx\), even when the working model for \({\eta}(a,\bx)\) is linear in \((a,\bx)\). In parallel, the {mean-scale CMTP} at covariate level \(\bx\) {satisfies, by \Cref{prop:cmtp_outcome_identifiability},} $\tau_\delta(\bx)
=\mathbb{E}\bigl[\mu(A+\delta,\bx)-\mu(A,{\bx})\mid \bX=\bx\bigr],$
so that \(\tau_\delta(\bx)\) aggregates the local nudge effects \(\tau_\delta(a,\bx)\) over the conditional distribution of \(A\mid \bX=\bx\). In the logistic example this yields $\tau_\delta(\bx)
=
\mathbb{E}[
\operatorname{expit}\!\bigl(\eta(A+\delta,\bx)\bigr)
-
\operatorname{expit}\!\bigl(\eta(A,\bx)\bigr)
\,\mid\, \bX=\bx]$. 
Here, we construct an A-learning analog based on the NLL, and defer the weighting version to the {Supplementary Section~A.2}. Recall the {recoded} arm indicator $\widetilde\Lambda := \Lambda + \tfrac12 \in \{0,1\}$ and the arm propensity $\pi_\Lambda(\bx,a) := \mathbb{P}\bigl(\widetilde\Lambda = 1 \mid \bX=\bx,A_\Lambda=a\bigr).$
For a candidate link-scale contrast $f(a,\bx)$, we define the A-learning NLL loss
{\small\begin{equation}\label{eq:A-NLL-loss}
\ell_A^{\mathrm{NLL}}(f;a,\bx)
=
\mathbb{E}\Big[
  -Y\bigl(\widetilde\Lambda - \pi_\Lambda(\bx,a)\bigr)\, f(a,\bx)
  + b\bigl( (\widetilde\Lambda - \pi_\Lambda(\bx,a))\, f(a,\bx) \bigr)
  \,\Big|\,
 \bX = \bx,\ A_\Lambda = a
\Big],
\end{equation}}
where $b$ is the log-partition function. This loss regresses $Y$ on the centered arm indicator $\widetilde\Lambda - \pi_\Lambda(\bx,a)$ scaled by $f(a,\bx)$, mimicking the usual A-learning construction in a generalized-linear setting. The FOC with respect to $f$ leads to a contrast equation analogous to Supplementary equation (S3). Now, defining $\ell_A^{\mathrm{NLL}}(f)$ as the marginal loss {$\mathbb E[\ell_A^{\mathrm{NLL}}(f;A_\Lambda,\bX)]$}, we state the following proposition.

\begin{proposition}\label{prop:nudge_alearning_nll}
 {Suppose the conditions of \Cref{prop:nudge_outcome_identifiability} hold and $f_{A\mid\bX}(a\mid\bx)>0$, so that $\pi_\Lambda(\bx,a)\in(0,1)$.} Assume $\ell_A^{\mathrm{NLL}}(f)${, the marginal loss corresponding to} \eqref{eq:A-NLL-loss}{,} admits a unique minimizer $f_A^*(a,\bx)$. Then $f_A^*(a,\bx)$ satisfies $\mu\bigl( (1-\pi_\Lambda(\bx,a)) f_A^*(a,\bx) \bigr)
-
\mu\bigl( -\pi_\Lambda(\bx,a)\, f_A^*(a,\bx) \bigr)
=
\tau_\delta(a,\bx).$
\end{proposition}
\noindent In particular, in the logistic case this identity again links the minimizer $f_A^*(a,\bx)$ to the arm-specific difference in success probabilities at $(\bx,a)$, now expressed through the centered arm indicator $\widetilde\Lambda - \pi_\Lambda(\bx,a)$. As in the squared-loss A-learning construction, the centering removes main-effect variation and isolates the arm contrast at $(a,\bx)$ on the link scale.
As before, one may subtract an offset $m(a,\bx)$ from $Y$ inside the linear term in \eqref{eq:A-NLL-loss} to reduce variance without affecting the population target. In the simulations presented in the Supplement, we focus on the non-augmented form for simplicity. Moreover, a discussion {of} a direct NLL loss for CMTP estimation can be found in {Supplementary Section~A.3}.

\section{Empirical estimators}\label{sec:emp_est}

The preceding population losses define weighting and A-learning through conditional risks under the duplicated-arm construction. In finite samples, we restrict \(\tau_\delta(a,\bx)\) to a semiparametric class and estimate it by minimizing empirical analogs of these losses. We next give the corresponding objectives, show that linear-basis representations yield {(possibly weighted) least-squares estimators}, and characterize their population limits. These empirical losses are optimized in our
implementation and form the basis for the large-sample results and simulation studies in
{Sections~\ref{sec:theory} and~\ref{sec:simu}}. For compactness, we write $\bZ_i^\lambda := (A_i-\lambda\delta,\bX_i), 
    \bb_i^\lambda := \bb(\bZ_i^\lambda) \text{ and }
    Y_i^\lambda := Y_i \text{ for }
    \lambda\in\{0,1\}.$
Thus \(\bZ_i^0=(A_i,\bX_i)\) denotes the observed arm and
\(\bZ_i^1=(A_i-\delta,\bX_i)\) denotes the shifted arm. 
Define the duplicated empirical average $\bbP_n^\Lambda f_i^\lambda
    := n^{-1}\sum_{i=1}^n\sum_{\lambda=0}^1 f_i^\lambda .$
For A-learning, define the transformed basis $\widetilde \bb_i^\lambda
    :=
    \left\{
        \widetilde\Lambda_i^\lambda
        -
        \widehat\pi_\Lambda({\bZ}_i^\lambda)
    \right\}
    {\bb}_i^\lambda,$
where {$\widetilde\Lambda_i^\lambda=\lambda$ is the arm indicator of the duplicated row and $\widehat\pi_\Lambda(\bZ_i^\lambda)=\widehat\pi_\Lambda(A_i-\lambda\delta,\bX_i)$ is the estimated arm propensity evaluated at the row's dose value}. Then the A-learning empirical loss is
{\small$$\widehat\btheta_A=
    \argmin_{\btheta}
    \bbP_n^\Lambda
    \left\{
        Y_i-{\btheta}^\top \widetilde {\bb}_i^\lambda
    \right\}^2 = \left\{
        \bbP_n^\Lambda
        \widetilde {\bb}_i^\lambda\widetilde {\bb}_i^{\lambda\top}
    \right\}^{-1}
    \left\{
        \bbP_n^\Lambda
        \widetilde {\bb}_i^\lambda Y_i
    \right\}.$$}

\begin{lemma}[Population limits of linear-basis nudge estimators]
\label{lem:nudge-linear-projection}
Let \(\tau_\delta(a,\bx)\) denote the conditional nudge effect and suppose we approximate it
in a linear basis, $\tau_\delta(a,\bx)\approx \btheta^\top \bb(a,\bx),$
where \(\bb(a,\bx)\in\mathbb R^p\) is a fixed vector of basis functions and
\(\btheta\in\mathbb R^p\) is unknown. {Suppose $\widehat\pi_\Lambda$ is uniformly consistent for $\pi_\Lambda$, the matrix $\mathbb E^\Lambda[h_A(A_\Lambda,\bX)\bb(A_\Lambda,\bX)\bb(A_\Lambda,\bX)^\top]$ is nonsingular, and the relevant second moments are finite.} Then
\(\widehat\btheta_A\to\btheta_A^*\) in probability, where
$\btheta_A^*
    =
    \argmin_{\btheta\in\mathbb R^p}
    \mathbb E{^\Lambda}
    \left[
        h_A({A_\Lambda},\bX)
        \left\{
            \tau_\delta({A_\Lambda},\bX)-\btheta^\top \bb({A_\Lambda},\bX)
        \right\}^2
    \right],$
with $h_A(a,\bx)
    :=
    \pi_\Lambda(a,\bx)\{1-\pi_\Lambda(a,\bx)\}{,}$
{where $\mathbb E^\Lambda$ denotes expectation under the duplicated-arm law of $(A_\Lambda,\bX)$, whose conditional density given $\bX=\bx$ is the mixture $\tfrac12\{g_{A\mid\bX}(a\mid\bx)+g_{A\mid\bX}(a+\delta\mid\bx)\}$; each observation contributes to the projection through both of its duplicated rows.}
\end{lemma}
% {\begin{proof}
% See Section~\ref{sec:supp_proofs} of the Supplementary Material.
% \end{proof}}

{In other words, under correct specification the A-learning estimator recovers the conditional nudge effect, while under misspecification it converges to the best $h_A$-weighted approximation of $\tau_\delta(a,\bx)$ in the span of the basis. Because $h_A=\pi_\Lambda(1-\pi_\Lambda)$, the projection emphasizes regions where the two duplicated arms overlap most, precisely where the nudge contrast is best identified.}
{A similar result for the weighting loss is presented in Section~A.4 of the Supplementary Material.}
For the direct CMTP estimators, we model the policy effect directly as
$\tau_\delta(\bx) \approx f_{\bbeta}(\bx) = {\bbeta}^\top \bc(\bx),$
where \(\bc(\bx)\) is a lower-dimensional basis in \(\bx\). Let
\(c_i=c(\bX_i)\), \(\hat g_i^\lambda=\lambda-\hat\pi_\Lambda(\bX_i,A_i-\lambda\delta)\),
and {\(\hat\pi_i^\lambda=\hat\pi_\Lambda(\bX_i,A_i-\lambda\delta)\)}, with
\(\lambda\in\{0,1\}\) denoting the duplicated arms. The stabilized direct
A-learning CMTP estimator is
{\small\[
    \hat\bbeta^{\mathrm{dir}}_{A}
    =
    \arg\min_{\bbeta}
    \bbP_n^\Lambda
    \frac{
    \{Y_i-\hat g_i^\lambda \bbeta^\top c_i\}^2
    }{{\hat\pi_i^\lambda}} =
    \left[
    \bbP_n^\Lambda
    \frac{(\hat g_i^\lambda)^2}{{\hat\pi_i^\lambda}}
    c_i c_i^\top
    \right]^{-1}
    \left[
    \bbP_n^\Lambda
    \frac{\hat g_i^\lambda}{{\hat\pi_i^\lambda}}
    c_i Y_i
    \right].
\]}
With augmentation by an outcome regression \(\hat m(A_i-\lambda\delta,\bX_i)\),
{\small\[
    \hat{\bbeta}^{\mathrm{dir}}_{A,\mathrm{aug}}
    =
    \left[
    \bbP_n^\Lambda
    \frac{(\hat g_i^\lambda)^2}{{\hat\pi_i^\lambda}}
    c_i c_i^\top
    \right]^{-1}
    \left[
    \bbP_n^\Lambda
    \frac{\hat g_i^\lambda}{{\hat\pi_i^\lambda}}
    c_i \{Y_i-\hat m(A_i-\lambda\delta,\bX_i)\}
    \right].
\]}
The estimators developed in this section provide the practical implementation of the weighting and A-learning procedures introduced in {Section~\ref{sec:methodology}}. By restricting the conditional nudge effect and the direct CMTP effect to finite-dimensional working models, the corresponding empirical objectives can be formulated as standard (weighted) least-squares problems, making estimation computationally straightforward. Lemma~\ref{lem:nudge-linear-projection} further demonstrates that, under model misspecification, these estimators will converge to weighted $L_2$ projections of the underlying causal effects onto the chosen basis rather than the true effects themselves, a standard property for semiparametric working models. These empirical estimating equations form the foundation for our inferential results. In the next section, we establish the large-sample behavior of our proposed estimators and derive asymptotic normality while accounting for the estimation of nuisance functions such as the duplicated-arm propensity model and, for the augmented estimators, the outcome regression model.

\section{Asymptotic Theory}
\label{sec:theory}

In this section we study the large-sample properties of our estimators of Section~\ref{sec:emp_est}, providing asymptotic normality that can be used to derive inferential procedures. We begin by fixing notation and stating regularity conditions. Let $\bO_i := (Y_i,A_i,\bX_i),
    \bZ_i^0 := (A_i,\bX_i),
    \bZ_i^1 := (A_i-\delta,\bX_i),$
and write ${\bb}_i^0 := \bb(\bZ_i^0),
    {\bb}_i^1 := {\bb}(\bZ_i^1).$
For a generic parameter value \(\boldeta\), define
$\pi_i^0(\boldeta) := \pi_\Lambda(\bZ_i^0;\boldeta),
    \pi_i^1(\boldeta) := \pi_\Lambda(\bZ_i^1;\boldeta).$
When no ambiguity arises, we suppress the dependence on \(\boldeta^*\) and write
$\pi_i^0 := \pi_i^0(\boldeta^*),
    \pi_i^1 := \pi_i^1(\boldeta^*).$
Let \(\pi_\Lambda(a,\bx;\boldeta)\) be a correctly specified model for
\(\pi_\Lambda(a,\bx)\), with true value \(\boldeta^*\). {In this and subsequent sections we write the dose argument first, $\pi_\Lambda(a,\bx)=\pi_\Lambda(\bx,a)$, matching the duplicated-row tuples $\bZ_i^\lambda=(A_i-\lambda\delta,\bX_i)$; the two orderings denote the same function.} {Throughout this section we assume that the matrix $\bD$ defined below and the matrices $J_{\boldeta}$ and $J_{\bnu}$ of conditions (C1) and (C4) are nonsingular, and that
$\mathbb E\big[(1+Y^2)\{\|\bb(A,\bX)\|^2+\|\bb(A-\delta,\bX)\|^2\}\big]<\infty$ and
$\mathbb E\big[\|\bb(A,\bX)\|^4+\|\bb(A-\delta,\bX)\|^4\big]<\infty$.} {Let \(U_{\boldeta}(\bO_i;\boldeta)\) denote the estimating function for the propensity-model parameter \(\boldeta\), so that \(\widehat{\boldeta}_n\) is defined by $\mathbb P_n U_{\boldeta}(\bO;\widehat{\boldeta}_n)=0.$ We assume the following conditions.}

\begin{enumerate}[label={(\bfseries C\arabic*)}]
    \item The estimator admits the asymptotic linear expansion
    $\sqrt n(\widehat{\boldeta}_n-{\boldeta}^*)
        =
        {-J_{\boldeta}^{-1}
        \sqrt n\,\bbP_n U_{\boldeta}(\bO_i;{\boldeta}^*)}
        +
        o_p(1),$
    where $J_{\boldeta}
        :=
        \mathbb E\{\dot U_{\boldeta}(\bO_i;\boldeta^*)\},
        \dot U_{\boldeta}(\bO_i;{\boldeta}^*)
        :=
        \frac{\partial}{\partial{\boldeta}}
        U_{\boldeta}(\bO_i;\boldeta)|_{{\boldeta}={\boldeta}^*}.$
    The asymptotic variance matrix of \(\widehat{\boldeta}_n\) is denoted by
    \(V_{\boldeta}\) and has finite entries.

    \item {The classes $\{\pi_\Lambda(\cdot;\boldeta):\boldeta\in\mathcal H_{\boldeta}^*\}$ and $\{\dot\pi_\Lambda(\cdot;\boldeta):\boldeta\in\mathcal H_{\boldeta}^*\}$ are $P$-Donsker with square-integrable envelopes, where $\mathcal H_{\boldeta}^*$ is the neighborhood of $\boldeta^*$ in (C3).}

    \item The map \(\pi_\Lambda(a,\bx;{\boldeta})\) is twice differentiable with respect to
    \({\boldeta}\) in \(\mathcal H_{\boldeta}^*\times\mathcal A\times\mathcal X\), where
    \(\mathcal H_{\boldeta}^*\) is a neighborhood of \({\boldeta}^*\), and has bounded,
    continuous derivatives.
\end{enumerate}

Let \(m(a,\bx;\bnu)\) be {a working model for some function of \(a\) and \(\bx\) only,
with limiting value \(m(a,\bx;\bnu^*)\)}, not necessarily equal to
\(\mathbb E[Y\mid A=a,\bX=\bx]\). For compactness, define
$m_i^0(\bnu):=m(\bZ_i^0;\bnu),
    \text{ and }
    m_i^1({\bnu}):=m(\bZ_i^1;\bnu).$
Assume the following additional conditions.

\begin{enumerate}[label={(\bfseries C\arabic*)}]
    \setcounter{enumi}{3}
    \item The estimator \(\widehat\bnu_n\) satisfies $\bbP_n U_{\bnu}(\bO_i;\widehat\bnu_n)=0$
    and admits the asymptotic linear expansion
    $\sqrt n(\widehat\bnu_n-\bnu^*)
        =
        {-J_{\bnu}^{-1}
        \sqrt n\,\bbP_n U_{\bnu}(\bO_i;\bnu^*)}
        +
        o_p(1),$
    where $J_{\bnu}
        :=
        \mathbb E\{\dot U_{\bnu}(\bO_i;\bnu^*)\},$ and 
    $\dot U_{\bnu}(\bO_i;\bnu^*)
        :=
        \left.
        \frac{\partial}{\partial\bnu}
        U_{\bnu}(\bO_i;\bnu)
        \right|_{\bnu=\bnu^*}.$
    The asymptotic variance matrix of \(\widehat\bnu_n\) is denoted by
    \(V_{\bnu}\) and has finite entries.
    \item {The classes $\{m(\cdot;\bnu):\bnu\in\mathcal H_{\bnu}^*\}$ and $\{\dot m(\cdot;\bnu):\bnu\in\mathcal H_{\bnu}^*\}$ are $P$-Donsker with square-integrable envelopes, where $\mathcal H_{\bnu}^*$ is the neighborhood of $\bnu^*$ in (C6).}
    \item The map \(m(a,\bx;\bnu)\) is twice differentiable with respect to \({\bnu}\) in
    \(\mathcal H_{\bnu}^*\times\mathcal A\times\mathcal X\), where
    \(\mathcal H_{\bnu}^*\) is a neighborhood of \(\bnu^*\), and has bounded,
    continuous derivatives.
\end{enumerate}
Conditions (C1)-(C3) are {standard regularity assumptions that account for}
the estimation of the propensity component \(\pi_\Lambda(a,\bx)\): (C1) requires
\(\widehat{\boldeta}_n\) to be asymptotically linear, so that the first-order
effect of estimating \(\boldeta^*\) is captured by the influence function
\({-J_{\boldeta}^{-1}U_{\boldeta}(\bO_i;\boldeta^*)}\), while (C2) and (C3)
provide the empirical process (Donsker) control and smoothness needed for
uniform Taylor expansion arguments in a neighborhood of \(\boldeta^*\).
Conditions (C4)-(C6) play the analogous role for the augmentation model
\(m(a,\bx;\bnu)\), which is an arbitrary working outcome regression and need
not be correctly specified as \(\mathbb E[Y\mid A=a,\bX=\bx]\). Together, these
assumptions yield an asymptotically linear representation consisting of the
main A-learning estimating function plus correction terms for nuisance
estimation.
For a generic function \(m\), define $Y_i^{0,m}:=Y_i-m_i^0,
\text{ and }
Y_i^{1,m}:=Y_i-m_i^1,$
where \(m_i^j=m(\bZ_i^j)\), \(j=0,1\). The corresponding A-learning
estimating function is
{\small\[
\begin{aligned}
\psi_A^{(m)}(\bO_i;\btheta,\boldeta)
:=
{\{1-\pi_i^1(\boldeta)\}{\bb}_i^1
\left[
Y_i^{1,m}
-\{1-\pi_i^1(\boldeta)\}{\bb}_i^{1\top}\btheta
\right]-
\pi_i^0(\boldeta){\bb}_i^0
\left[
Y_i^{0,m}
+\pi_i^0(\boldeta){\bb}_i^{0\top}\btheta
\right].}
\end{aligned}
\]}
The unaugmented estimator corresponds to \(m\equiv0\), whereas the augmented
estimator uses \(m(\cdot)=m(\cdot;\bnu)\). Accordingly, write
$\psi_A(\bO_i;\btheta,\boldeta)
=
\psi_A^{(0)}(\bO_i;\btheta,\boldeta),
\text{ and }
\psi_{A,\mathrm{aug}}(\bO_i;\btheta,\boldeta,\bnu)
=
\psi_A^{(m(\cdot;\bnu))}(\bO_i;\btheta,\boldeta).$
Let $\bD
:=
-
\tfrac{\partial}{\partial\btheta^\top}
\mathbb E\{\psi_A(\bO_i;\btheta,\boldeta^*)\}
\mid_{\btheta=\btheta_A^*},$
and let \(B_{\boldeta}\) denote the derivative of
\(\mathbb E\{\psi_A(\bO_i;\btheta_A^*,\boldeta)\}\) with respect to
\(\boldeta^\top\), evaluated at \(\boldeta=\boldeta^*\). Similarly, let
\(B_{\boldeta,\mathrm{aug}}\) and \(B_{\bnu,\mathrm{aug}}\) denote the
derivatives of
\(\mathbb E\{\psi_{A,\mathrm{aug}}
(\bO_i;\btheta_A^*,\boldeta,\bnu)\}\)
with respect to \(\boldeta^\top\) and \(\bnu^\top\), respectively, evaluated at
\((\boldeta^*,\bnu^*)\).

\begin{theorem}[{Asymptotic normality and variance reduction of augmented A-learning}]
\label{thm:a_learning_joint_asymptotic}
Let \(\widehat\btheta_A\) and
\(\widehat\btheta_{A,\mathrm{aug}}\) solve the empirical estimating equations
based on \(\psi_A\) and \(\psi_{A,\mathrm{aug}}\), respectively, and define $\phi_A(\bO) =
\psi_A(\bO;\btheta_A^*,\boldeta^*) -
B_{\boldeta}J_{\boldeta}^{-1}
U_{\boldeta}(\bO;\boldeta^*)$ and $\phi_{A,\mathrm{aug}}(\bO)=
\psi_{A,\mathrm{aug}}
(\bO;\btheta_A^*,\boldeta^*,\bnu^*)-
B_{\boldeta,\mathrm{aug}}J_{\boldeta}^{-1}
U_{\boldeta}(\bO;\boldeta^*)-
B_{\bnu,\mathrm{aug}}J_{\bnu}^{-1}
U_{\bnu}(\bO;\bnu^*).$
Under conditions {\normalfont (C1)-(C3)} for the unaugmented estimator and
{\normalfont (C1)-(C6)} for the augmented estimator,
{\small\[
\sqrt n(\widehat\btheta_A-\btheta_A^*)
\rightsquigarrow
N(\bzero,\bD^{-1}\bSigma_A\bD^{-1}),
\qquad
\sqrt n(\widehat\btheta_{A,\mathrm{aug}}-\btheta_A^*)
\rightsquigarrow
N(\bzero,\bD^{-1}\bSigma_{A,\mathrm{aug}}\bD^{-1}),
\]}
{where}
\(\bSigma_A=\mathbb E\{\phi_A(\bO)^{\otimes2}\}\) and
\(\bSigma_{A,\mathrm{aug}}
=\mathbb E\{\phi_{A,\mathrm{aug}}(\bO)^{\otimes2}\}\).
If, in addition,
\(m(a,\bx;\bnu^*)=\mathbb E(Y\mid A=a,\bX=\bx)\), the nuisance-score
coefficients are their population \(L_2\)-projection coefficients, {and $\bbE[\phi_{A,\mathrm{aug}}(\bO)\,D(A,\bX)^\top]=\bzero$, where $D(A,\bX):=\{1-\pi_\Lambda(A-\delta,\bX)\}\bb(A-\delta,\bX)m(A-\delta,\bX;\bnu^*)-\pi_\Lambda(A,\bX)\bb(A,\bX)m(A,\bX;\bnu^*)$ is the difference between the two scores,} then $\bSigma_{A,\mathrm{aug}}\preceq\bSigma_A,
$
and hence $\bD^{-1}\bSigma_{A,\mathrm{aug}}\bD^{-1}
\preceq
\bD^{-1}\bSigma_A\bD^{-1}.$
Equality holds if and only if
\(\phi_{A,\mathrm{aug}}(\bO)=\phi_A(\bO)\) almost surely.
\end{theorem}
{Theorem~\ref{thm:a_learning_joint_asymptotic} establishes that both A-learning nudge estimators are root-$n$ consistent and asymptotically normal, so that Wald-type inference can proceed from plug-in estimates of $\bD$ and $\bSigma_A$ or $\bSigma_{A,\mathrm{aug}}$. Its second part indicates that, when the augmentation function is correctly specified and the stated orthogonality condition holds, augmentation can only reduce the asymptotic variance, supporting its routine use.}
{A remark in the Supplementary Material records the usual projection simplification of the sandwich variance when $\widehat{\boldeta}_n$ is an efficient estimator of $\boldeta^*$ \citep{pierce1982asymptotic}, for both the unaugmented and augmented estimators.}
{The nudge asymptotics also yield inference for the plug-in CMTP estimator obtained by averaging the fitted nudge over the conditional dose distribution; the formal statement, whose guarantees depend on how well that distribution can be estimated, is given as Corollary~S1 in the Supplementary Material together with a remark on the same-sample projection implementation used in our numerical work. We therefore focus here on the direct CMTP estimator, which requires no such auxiliary estimate.}

Unlike the plug-in approach above,
{the direct CMTP estimator of Section~\ref{sec:emp_est}} targets the projection of \(\tau_\delta(\bX)\) directly through a
stabilized A-learning score, avoiding the intermediate estimation of the full
conditional nudge surface \(\tau_\delta(a,\bX)\). For the direct CMTP working model,
let \({\bb}_i:=b(\bX_i)\) denote the vector of basis functions in \(\bX_i\), and, as in
Section~\ref{sec:emp_est}, let $g_i^\lambda(\boldeta):=\lambda-\pi_i^\lambda(\boldeta)$
denote the centered arm score of the duplicated row $\lambda\in\{0,1\}$, where
$\pi_i^\lambda(\boldeta):=\pi_\Lambda(\bZ_i^\lambda;\boldeta)$. The
$\pi_\Lambda$-stabilized direct CMTP estimating function, $\psi_{C,\mathrm{aug}}(\bO_i;\btheta,\boldeta,\bnu) = \sum_{\lambda\in\{0,1\}}
({g_i^\lambda(\boldeta)}/{\pi_i^\lambda(\boldeta)})\,
{\bb}_i
\left[
Y_i-m_i^\lambda(\bnu)-g_i^\lambda(\boldeta)\,{\bb}_i^\top\btheta
\right]$, evaluates to
{\small \[
\frac{1-\pi_i^1(\boldeta)}{\pi_i^1(\boldeta)}\,
{\bb}_i
\left[
Y_i-m_i^1(\bnu)-\{1-\pi_i^1(\boldeta)\}{\bb}_i^\top\btheta
\right]
-
{\bb}_i
\left[
Y_i-m_i^0(\bnu)+\pi_i^0(\boldeta)\,{\bb}_i^\top\btheta
\right],
\]}
where the observed-arm row enters with the constant stabilized weight
$g_i^0(\boldeta)/\pi_i^0(\boldeta)=-1$.
Let \(\widehat\btheta_{C,\mathrm{aug}}\) solve
$\bbP_n\psi_{C,\mathrm{aug}}(\bO_i;\btheta,\widehat{\boldeta}_n,\widehat{\bnu}_n)=\bzero$
(equivalently, $\widehat\btheta_{C,\mathrm{aug}}=\hat\bbeta^{\mathrm{dir}}_{A,\mathrm{aug}}$
of Section~\ref{sec:emp_est} with $c=b$), and let \(\btheta_C^*\) solve
$\mathbb E\{\psi_{C,\mathrm{aug}}(\bO_i;\btheta_C^*,{\boldeta}^*,\bnu^*)\}=\bzero$.
Because the $\pi_\Lambda$-stabilization equalizes the conditional second moment of the
centered arm scores, the population Jacobian and moment vector take the forms
{\small\[
D_C
:=
-\left.\frac{\partial}{\partial\btheta}
\mathbb E\{\psi_{C,\mathrm{aug}}(\bO_i;\btheta,{\boldeta}^*,{\bnu}^*)\}\right|_{\btheta=\btheta_C^*}
=
\mathbb E\!\bigg[\sum_{\lambda\in\{0,1\}}
\frac{\{g_i^\lambda(\boldeta^*)\}^2}{\pi_i^\lambda(\boldeta^*)}\,\bb_i\bb_i^\top\bigg]
=
\mathbb E\{b(\bX)b(\bX)^\top\},
\]}
{\small\[
\mathbb E\!\bigg[\sum_{\lambda\in\{0,1\}}
\frac{g_i^\lambda(\boldeta^*)}{\pi_i^\lambda(\boldeta^*)}\,\bb_i\{Y_i-m_i^\lambda(\bnu)\}\bigg]
=
\mathbb E\{b(\bX)\tau_\delta(\bX)\}
\quad\text{for every }\bnu,
\]}
so that, under the identification assumptions for the CMTP, \(\btheta_C^*\) is 
the unweighted $L_2(P_{\bX})$ projection coefficient
$\btheta_C^*
=
\arg\min_{\btheta}
\mathbb E\left[
\{\tau_\delta(\bX)-b(\bX)^\top\btheta\}^2
\right],$
regardless of the working outcome model $m$
(Lemma~S3 of the Supplementary Material). Let
\(B_{C,{\boldeta}}\) and \(B_{C,{\bnu}}\) denote the population
derivatives of
\(\mathbb E\{\psi_{C,\mathrm{aug}}(\bO_i;\btheta_C^*,{\boldeta},{\bnu})\}\)
with respect to \({\boldeta}^\top\) and \({\bnu}^\top\), evaluated at \(({\boldeta}^*,{\bnu}^*)\).
Since $\mathbb E\{\psi_{C,\mathrm{aug}}(\bO_i;\btheta,\boldeta^*,\bnu)\}$ is constant in
$\bnu$, we have $B_{C,\bnu}=\bzero$: the direct CMTP estimating function is orthogonal
to the outcome-regression nuisance whenever the propensity model is correctly specified.

\begin{theorem}[Asymptotic normality of augmented direct CMTP A-learning]
\label{thm:direct_cmtp_aug_alearning_asymptotic}
Assume conditions {\normalfont (C1)-(C6)}, that $\mathbb E\{b(\bX)b(\bX)^\top\}$ is
nonsingular, that $\mathbb E\|b(\bX)\|^2<\infty$ and $\mathbb E(Y^2)<\infty$, and that
$\pi_\Lambda(a,\bx;\boldeta)\ge c_0>0$ for all $(a,\bx)$ and all $\boldeta$ in a
neighborhood of $\boldeta^*$. Then,
as \(n\to\infty\), $\sqrt n(\widehat\btheta_{C,\mathrm{aug}}-\btheta_C^*)
\xrightarrow{d}
N\left(
0,
D_C^{-1}
\Sigma_{C,\mathrm{aug}}(\btheta_C^*;{\boldeta}^*,{\bnu}^*)
D_C^{-1}
\right),$
where $D_C=\mathbb E\{b(\bX)b(\bX)^\top\}$ and
$\Sigma_{C,\mathrm{aug}}(\btheta_C^*;{\boldeta}^*,{\bnu}^*) = $ {\small$\mathbb E
\left[
\left\{
\psi_{C,\mathrm{aug}}(\bO_i;\btheta_C^*,{\boldeta}^*,{\bnu}^*)
-
B_{C,{\boldeta}}J_{\boldeta}^{-1}U_{\boldeta}(\bO_i;{\boldeta}^*)
\right\}^{\otimes 2}
\right];$}
no correction for the estimation of $\bnu$ appears because $B_{C,\bnu}=\bzero$.
If, in addition, $\widehat\boldeta_n$ is an efficient estimator of $\boldeta^*$, then
$\Sigma_{C,\mathrm{aug}}
=\widetilde\Sigma_{C,\mathrm{aug}}-B_{C,\boldeta}V_{\boldeta}B_{C,\boldeta}^\top$,
where $\widetilde\Sigma_{C,\mathrm{aug}}
=\mathbb E\{\psi_{C,\mathrm{aug}}(\bO_i;\btheta_C^*,\boldeta^*,\bnu^*)^{\otimes 2}\}$.
\end{theorem}
{Theorem~\ref{thm:direct_cmtp_aug_alearning_asymptotic} demonstrates that the direct CMTP estimator is root-$n$ consistent and asymptotically normal for the $L_2(P_{\bX})$ projection of the CMTP effect, with a sandwich variance accounting for the estimated propensity model. Because $B_{C,\bnu}=\bzero$, the outcome-regression augmentation contributes no first-order correction, so valid inference requires only a well-behaved propensity model while the working outcome model can be chosen freely to improve precision.}

\noindent
{Finally, the same M-estimation strategy yields an asymptotic normality result for the \emph{logistic} direct CMTP estimator, which uses the $\pi_\Lambda$-stabilized A-learning score with the logistic mean function and targets the direct CMTP projection on the link scale. The estimating function, the regularity conditions, and the formal statement are presented in Supplementary Theorem~S2 in the Supplementary Material, together with its proof. 

Together, these results establish asymptotically valid inference for the nudge and CMTP estimators under squared and logistic losses, while also showing that appropriate outcome-regression augmentation can improve efficiency without altering the target estimand. {In Section~\ref{sec:simu} we examine how well these large-sample guarantees describe finite-sample behavior across a range of designs.} {A parallel asymptotic normality result for the logistic A-learning \emph{nudge} estimator, which underpins the logistic nudge and average-CMTP analyses of Supplementary Section B.1 and \Cref{sec:real data}, is stated and proved as Theorem~S1 in the Supplementary Material.}

\section{Simulation Experiments}\label{sec:simu}
% \subsection{Setup 1: Linear}\label{subsec:setup1}
We evaluate finite-sample performance under a linear data-generating mechanism, where effects are defined on the mean scale{; the experiment assesses bias under correct specification, compares the weighting schemes, and quantifies the precision gains from augmentation}. We consider a continuous treatment \(A\) and a low-dimensional covariate vector \(\bX = (X_1, X_2)^\top\). {We generate the covariates i.i.d.\ with} $X_1\sim\mathrm{Unif}(-1,1)$ and $X_2\sim\mathrm{Bernoulli}(0.7)$. Conditional on {$\bX$}, the continuous dose $A$ follows an exponential distribution with mean $\exp\{2+0.5X_1+X_2\}$ (equivalently, rate $1/\exp\{2+0.5X_1+X_2\}$). Outcomes are drawn as $Y\mid (A,\bX)\sim\mathcal N\{\mu(A,\bX),1\}$ with $\mu(A,\bX)\;=\;1+X_1+X_2+\beta_AA+\alpha A^2+X_1\{\beta_{AX}A+\beta A^2\},$
where we fix $(\alpha,\beta,\beta_{AX},\beta_A)=(-0.15,\,0.2,\,0.25,\,-0.5)$ and take a shift size $\delta=0.5$. This specification induces a nonlinear and covariate-modulated relationship between the treatment and the outcome, while retaining a simple closed-form expression for the true causal effects under a shift intervention.
 % The two causal targets are the conditional nudge effect $\tau_\delta(a,x)=\mu(a+\delta,\bx)-\mu(a,x)$ and the CMTP effect, which  averages this conditional nudge effect over the distribution of the observed treatment given covariates, $\tau_\delta(\bx)=\mathbb E\{\mu(A+\delta,\bx)-\mu(A,x)\mid \bX=\bx\}$.
Under the {outcome model}, our first causal target is the nudge { effect}, which is given by:
$\tau_\delta(a,\bx) = \mu(a+\delta,\bx)-\mu(a,\bx) 
= \beta_A \delta + 2\alpha a \delta + \alpha \delta^2
+ x_1\bigl(\beta_{AX} \delta + 2\beta a \delta + \beta \delta^2\bigr).$
This is a quadratic function of \(a\) with coefficients depending on \(x_1\). Also, using \(\mathbb{E}[A \mid \bX = \bx] = \exp(2 + 0.5 x_1 + x_2)\), we obtain {an} explicit expression for our second causal target, the CMTP effect $\tau_\delta(\bx) = \mathbb E\{\mu(A+\delta,\bx)-\mu(A,{\bx})\mid \bX=\bx\}$, which averages {the} conditional nudge effect over the distribution of the observed treatment given covariates { and} in becomes $\beta_A \delta + \alpha \delta^2 + 2\alpha \delta\,\mathbb{E}[A\mid \bX=\bx] + x_1\bigl(\beta_{AX} \delta + \beta \delta^2 + 2\beta \delta\,\mathbb{E}[A\mid \bX=\bx]\bigr).$
  This serves as the ground truth when evaluating CMTP estimators. To evaluate estimators we project each target onto low-dimensional working bases. For the nudge effect we use the linear span of $\{1,\ A,\ X_1,\ X_2,\ A\!:\!X_1,\ A\!:\!X_2\}$; the corresponding population coefficients have closed forms
$(\text{Intercept},A,X_1,X_2,A\!:\!X_1,A\!:\!X_2)
=(\beta_A\delta+\alpha\delta^2,\ 2\alpha\delta,\ \beta_{AX}\delta+\beta\delta^2,\ 0,\ 2\beta\delta,\ 0).$
Similarly, for the CMTP effect we project onto $\{1,X_1,X_2\}$. {Here,} the true function \(\tau_\delta(\bx)\) is nonlinear in \(\bx\) through
\(\mathbb E[A\mid \bX=\bx]=\exp(2+0.5x_1+x_2)\). Therefore, when we fit a linear working
model in the basis \(c(\bx)=(1,x_1,x_2)^\top\), the target coefficient is the population
least-squares projection
${\bbeta}_{\mathrm{CMTP}}^*
=
\arg\min_{{\bbeta}\in\mathbb R^3}
\mathbb E\left[
\left\{
\tau_\delta(\bX)-{\bbeta}^\top \bc(\bX)
\right\}^2
\right]
=
\left\{\mathbb E[\bc(\bX)\bc(\bX)^\top]\right\}^{-1}
\mathbb E[\bc(\bX)\tau_\delta(\bX)].$
This projected coefficient vector serves as the ground truth for evaluating the CMTP
coefficient estimators.
{We compare three base methods (Density, Overlap, A-learning) and their augmented versions under squared-error loss, all with a common linear working model; for the CMTP we consider both the average and direct constructions, giving twelve CMTP estimators. For each $n\in\{500,1000,5000\}$ we generate independent datasets, fit all estimators, and record empirical root mean squared error (RMSE) for each coefficient, measuring how well they recover the true linear projections of $\tau_\delta(a,\bx)$ and $\tau_\delta(\bx)$.}

\begin{figure}[htbp]
    \centering
    \includegraphics[width=0.8\linewidth]{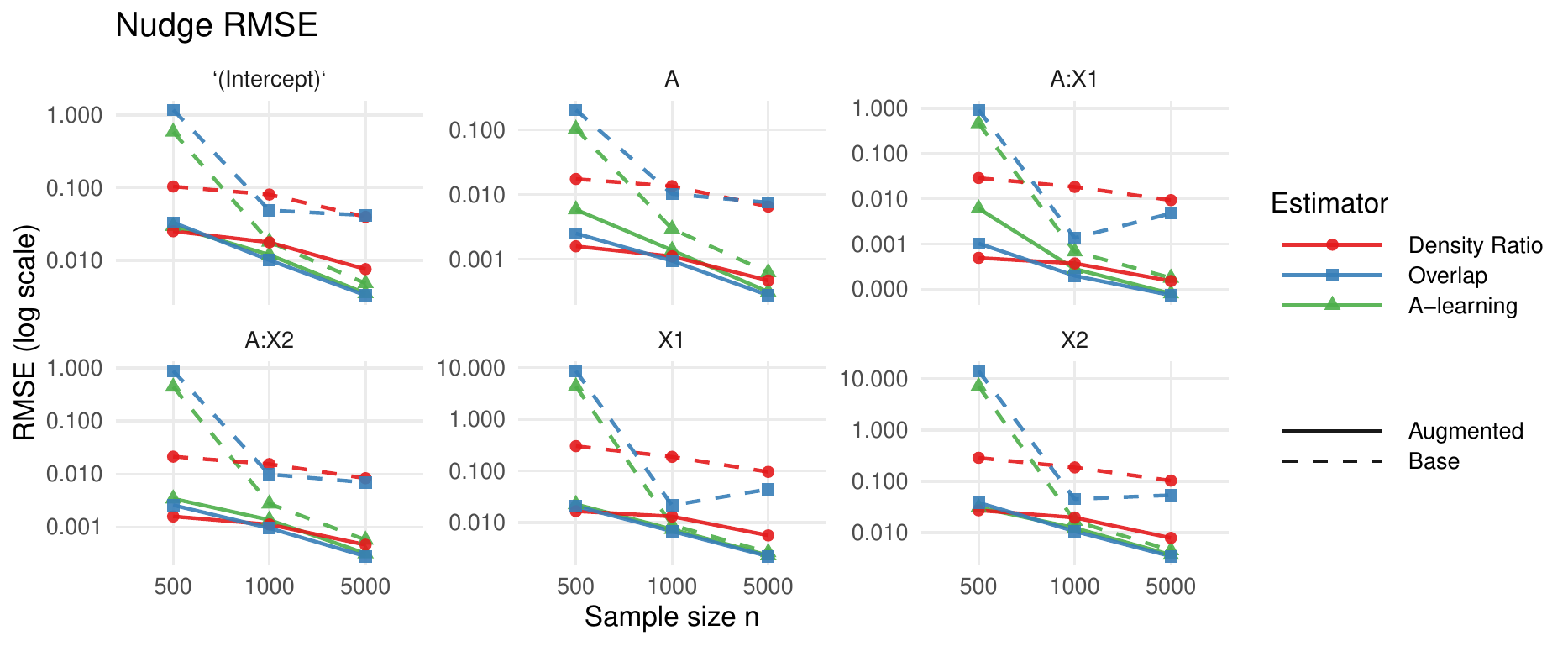}
    \caption{{Coefficient-wise RMSE (log scale) of nudge coefficients for Setup 1.}}
    \label{fig:Nudge_CMTP_rmse}
\end{figure}

\begin{table}[!htbp]
\centering
\footnotesize
\caption{Linear outcome, CMTP coefficients: empirical bias, RMSE, and 95\% Wald coverage. Bold = minimum $|\text{bias}|$ or minimum RMSE within each coefficient and $n$.}
\label{tab:linear_cmtp}
\resizebox{\textwidth}{!}{%
\begin{tabular}{lllllrrrrrrrrr}
\toprule
 & & & & & \multicolumn{3}{c}{$n=500$} & \multicolumn{3}{c}{$n=1{,}000$} & \multicolumn{3}{c}{$n=5{,}000$} \\
\cmidrule(lr){6-8} \cmidrule(lr){9-11} \cmidrule(lr){12-14}
Coef. & Truth & Constr. & Weighting & Aug. & Bias & RMSE & Cov. & Bias & RMSE & Cov. & Bias & RMSE & Cov. \\
\midrule
\multirow{12}{*}{$X_1$} & \multirow{12}{*}{0.9594} & \multirow{6}{*}{Average} & \multirow{2}{*}{A-learning} & No  & -0.0013 & 0.1348 & 0.9700 & -0.0026 & 0.1004 & 0.9800 & \textbf{0.0003} & \textbf{0.0285} & 0.9850 \\
 & & & & Yes & 0.0069 & \textbf{0.0851} & 0.9550 & -0.0024 & 0.1004 & 0.9500 & -0.0013 & \textbf{0.0285} & 0.9500 \\
\cmidrule(lr){4-14}
 & & & \multirow{2}{*}{Density Ratio} & No  & -0.0014 & 0.1349 & 0.9500 & -0.0027 & 0.1006 & 0.9250 & \textbf{0.0002} & \textbf{0.0285} & 0.9400 \\
 & & & & Yes & 0.0069 & \textbf{0.0851} & 0.9550 & -0.0024 & 0.1006 & 0.9400 & -0.0013 & \textbf{0.0285} & 0.9500 \\
\cmidrule(lr){4-14}
 & & & \multirow{2}{*}{Overlap} & No  & -0.0013 & 0.1348 & 0.9500 & -0.0026 & 0.1004 & 0.9300 & \textbf{0.0003} & \textbf{0.0285} & 0.9400 \\
 & & & & Yes & 0.0069 & \textbf{0.0851} & 0.9500 & -0.0024 & 0.1004 & 0.9400 & -0.0013 & \textbf{0.0285} & 0.9500 \\
\cmidrule(lr){3-14}
 & & \multirow{6}{*}{Direct} & \multirow{2}{*}{A-learning} & No  & -0.0016 & 0.1361 & 0.9700 & -0.0028 & \textbf{0.0624} & 0.9850 & \textbf{0.0002} & \textbf{0.0285} & 0.9850 \\
 & & & & Yes & 0.0067 & 0.0857 & 0.9500 & -0.0025 & \textbf{0.0624} & 0.9450 & -0.0014 & \textbf{0.0285} & 0.9500 \\
\cmidrule(lr){4-14}
 & & & \multirow{2}{*}{Density Ratio} & No  & -0.0016 & 0.1360 & 0.9700 & -0.0028 & \textbf{0.0624} & 0.9850 & \textbf{0.0002} & \textbf{0.0285} & 0.9850 \\
 & & & & Yes & 0.0067 & 0.0857 & 0.9500 & -0.0025 & \textbf{0.0624} & 0.9450 & -0.0014 & \textbf{0.0285} & 0.9500 \\
\cmidrule(lr){4-14}
 & & & \multirow{2}{*}{Overlap} & No  & \textbf{0.0002} & 0.1362 & 0.9700 & {-0.0010} & \textbf{0.0624} & 0.9850 & 0.0020 & \textbf{0.0285} & 0.9850 \\
 & & & & Yes & 0.0085 & 0.0860 & 0.9450 & \textbf{-0.0008} & \textbf{0.0624} & 0.9500 & 0.0005 & \textbf{0.0285} & 0.9500 \\
\midrule
\multirow{12}{*}{$X_2$} & \multirow{12}{*}{-0.6207} & \multirow{6}{*}{Average} & \multirow{2}{*}{A-learning} & No  & -0.0089 & 0.1236 & 1.0000 & -0.0018 & \textbf{0.0588} & 0.9900 & \textbf{0.0000} & 0.0269 & 1.0000 \\
 & & & & Yes & -0.0070 & 0.0821 & 0.9400 & 0.0037 & \textbf{0.0588} & 0.9500 & 0.0016 & 0.0269 & 0.9300 \\
\cmidrule(lr){4-14}
 & & & \multirow{2}{*}{Density Ratio} & No  & -0.0090 & 0.1236 & 0.9950 & -0.0019 & \textbf{0.0588} & 0.9650 & \textbf{0.0000} & 0.0269 & 0.9600 \\
 & & & & Yes & -0.0070 & 0.0821 & 0.9350 & 0.0038 & \textbf{0.0588} & 0.9500 & 0.0016 & 0.0269 & 0.9300 \\
\cmidrule(lr){4-14}
 & & & \multirow{2}{*}{Overlap} & No  & -0.0089 & 0.1236 & 0.9950 & -0.0018 & \textbf{0.0588} & 0.9650 & \textbf{0.0000} & 0.0269 & 0.9600 \\
 & & & & Yes & -0.0070 & 0.0821 & 0.9350 & 0.0037 & \textbf{0.0588} & 0.9500 & 0.0016 & 0.0269 & 0.9300 \\
\cmidrule(lr){3-14}
 & & \multirow{6}{*}{Direct} & \multirow{2}{*}{A-learning} & No  & -0.0090 & 0.1238 & 1.0000 & -0.0016 & 0.0589 & 0.9900 & 0.0001 & \textbf{0.0267} & 1.0000 \\
 & & & & Yes & -0.0071 & 0.0821 & 0.9400 & 0.0041 & 0.0589 & 0.9400 & 0.0017 & \textbf{0.0267} & 0.9300 \\
\cmidrule(lr){4-14}
 & & & \multirow{2}{*}{Density Ratio} & No  & -0.0090 & 0.1238 & 1.0000 & -0.0016 & 0.0589 & 0.9900 & 0.0001 & \textbf{0.0267} & 1.0000 \\
 & & & & Yes & -0.0071 & 0.0821 & 0.9400 & 0.0041 & 0.0589 & 0.9400 & 0.0017 & \textbf{0.0267} & 0.9300 \\
\cmidrule(lr){4-14}
 & & & \multirow{2}{*}{Overlap} & No  & -0.0077 & \textbf{0.1235} & 1.0000 & \textbf{-0.0003} & 0.0589 & 0.9900 & 0.0014 & \textbf{0.0267} & 1.0000 \\
 & & & & Yes & \textbf{-0.0058} & \textbf{0.0818} & 0.9200 & 0.0053 & 0.0589 & 0.9350 & 0.0030 & \textbf{0.0267} & 0.9250 \\
\midrule
\multirow{12}{*}{Int.} & \multirow{12}{*}{-0.4667} & \multirow{6}{*}{Average} & \multirow{2}{*}{A-learning} & No  & 0.0049 & 0.0645 & 0.9800 & 0.0025 & \textbf{0.0308} & 0.9750 & 0.0010 & \textbf{0.0144} & 0.9850 \\
 & & & & Yes & \textbf{0.0001} & \textbf{0.0459} & 0.9600 & \textbf{-0.0006} & \textbf{0.0308} & 0.9600 & \textbf{0.0001} & \textbf{0.0144} & 0.9200 \\
\cmidrule(lr){4-14}
 & & & \multirow{2}{*}{Density Ratio} & No  & 0.0049 & 0.0647 & 0.9500 & 0.0025 & 0.0309 & 0.9500 & 0.0011 & \textbf{0.0144} & 0.9700 \\
 & & & & Yes & \textbf{0.0001} & \textbf{0.0460} & 0.9600 & \textbf{-0.0006} & 0.0309 & 0.9600 & \textbf{0.0001} & \textbf{0.0144} & 0.9200 \\
\cmidrule(lr){4-14}
 & & & \multirow{2}{*}{Overlap} & No  & 0.0049 & 0.0645 & 0.9500 & 0.0025 & \textbf{0.0308} & 0.9500 & 0.0010 & \textbf{0.0144} & 0.9700 \\
 & & & & Yes & \textbf{0.0001} & \textbf{0.0459} & 0.9600 & \textbf{-0.0006} & \textbf{0.0308} & 0.9600 & \textbf{0.0001} & \textbf{0.0144} & 0.9200 \\
\cmidrule(lr){3-14}
 & & \multirow{6}{*}{Direct} & \multirow{2}{*}{A-learning} & No  & 0.0050 & 0.0646 & 0.9850 & 0.0022 & 0.0309 & 0.9750 & 0.0010 & \textbf{0.0144} & 0.9850 \\
 & & & & Yes & \textbf{0.0001} & \textbf{0.0459} & 0.9600 & -0.0009 & 0.0309 & 0.9600 & \textbf{0.0001} & \textbf{0.0144} & 0.9200 \\
\cmidrule(lr){4-14}
 & & & \multirow{2}{*}{Density Ratio} & No  & 0.0050 & 0.0647 & 0.9850 & 0.0022 & 0.0309 & 0.9750 & 0.0010 & \textbf{0.0144} & 0.9850 \\
 & & & & Yes & \textbf{0.0001} & \textbf{0.0459} & 0.9600 & -0.0008 & 0.0309 & 0.9600 & \textbf{0.0001} & \textbf{0.0144} & 0.9200 \\
\cmidrule(lr){4-14}
 & & & \multirow{2}{*}{Overlap} & No  & 0.0039 & \textbf{0.0644} & 0.9850 & 0.0011 & 0.0309 & 0.9750 & -0.0002 & \textbf{0.0144} & 0.9900 \\
 & & & & Yes & -0.0010 & \textbf{0.0459} & 0.9700 & -0.0020 & 0.0309 & 0.9650 & -0.0012 & \textbf{0.0144} & 0.9200 \\
\bottomrule
\end{tabular}}
\end{table}

Figure~\ref{fig:Nudge_CMTP_rmse} reports coefficient-wise log-RMSE for \(n\in{\{}500,1000,5000{\}}\){, for the nudge coefficients}. RMSE decreases with $n$ throughout{;} Density Ratio performs best at $n=500$, whereas Overlap and A-learning decline more steeply and outperform it by $n=5000$. {The} nearly identical {Overlap and A-learning} curves agree with the theoretical expectation that, under correct nuisance specification, both {target the same population projection}. Augmentation lowers RMSE across the board. In the CMTP results shown in \Cref{tab:linear_cmtp}, we display coefficient-level bias, RMSE, and 95\% Wald coverage for the A-learning, Density Ratio, and Overlap estimators under the squared error loss with and without augmentation. The three weighting approaches are nearly very similar in performance, with Overlap and A-learning typically matching or slightly outperforming Density Ratio; augmentation  improves precision consistently, and direct CMTP estimation is competitive with, and at larger sample sizes occasionally more accurate than, averaging the estimated nudge, particularly with Overlap weighting. Coverage is close to nominal throughout, though the augmented intercept dips to roughly 0.92 at 
$n=5,000$. Corresponding results for a logistic-outcome design are reported in Section B of the Supplement.

\section{Analysis of Mechanical Power Data}\label{sec:real data}
We analyze data from a cohort of ICU patients who received invasive mechanical ventilation for at least 48 hours, originally assembled and studied by \cite{serpa2018mechanical}. The dataset was extracted from the
\textsc{mimic-iii} clinical database. For each patient, we define the continuous treatment \(A\) as the mean of the minimum and maximum MP of ventilation (in Joules per minute) during the second 24-hour period in the ICU, and restrict attention to patients with mechanically plausible, moderate MP levels. The primary outcome \(Y\) is an indicator of in-hospital mortality. We include a rich set of baseline and early clinical variables as potential effect modifiers and confounders capturing patient severity and ventilator management, such as demographics, SAPS II score, smoking status, ARDS severity, vasopressor use, oxygenation and gas-exchange measures (e.g., PaO\(_2\), PaCO\(_2\), PaO\(_2\)/FiO\(_2\)), blood pressure, and ventilator settings from the first ICU day. After complete-case selection the analytical sample
comprised $n = 5{,}011$ patients. In this analysis, our goal is to characterize how a 1 J/min decrease or increase in MP would change the risk of in-hospital mortality as a function of \(\bX\), via the nudge effect \(\tau_\delta(a,\bx)\) and the corresponding CMTP effect \(\tau_\delta(\bx)\). This identifies clinically-relevant effect modifiers that drive heterogeneous benefits or harms of modest changes in MP. We present results for
shifts $\delta \in \{-1,\, +1\}$ J/min; the negative shift
corresponds to a modest \emph{reduction} and the positive shift to a
modest \emph{increase} in MP.

\begin{figure}[htbp]
    \centering
    \includegraphics[width=0.62\linewidth]{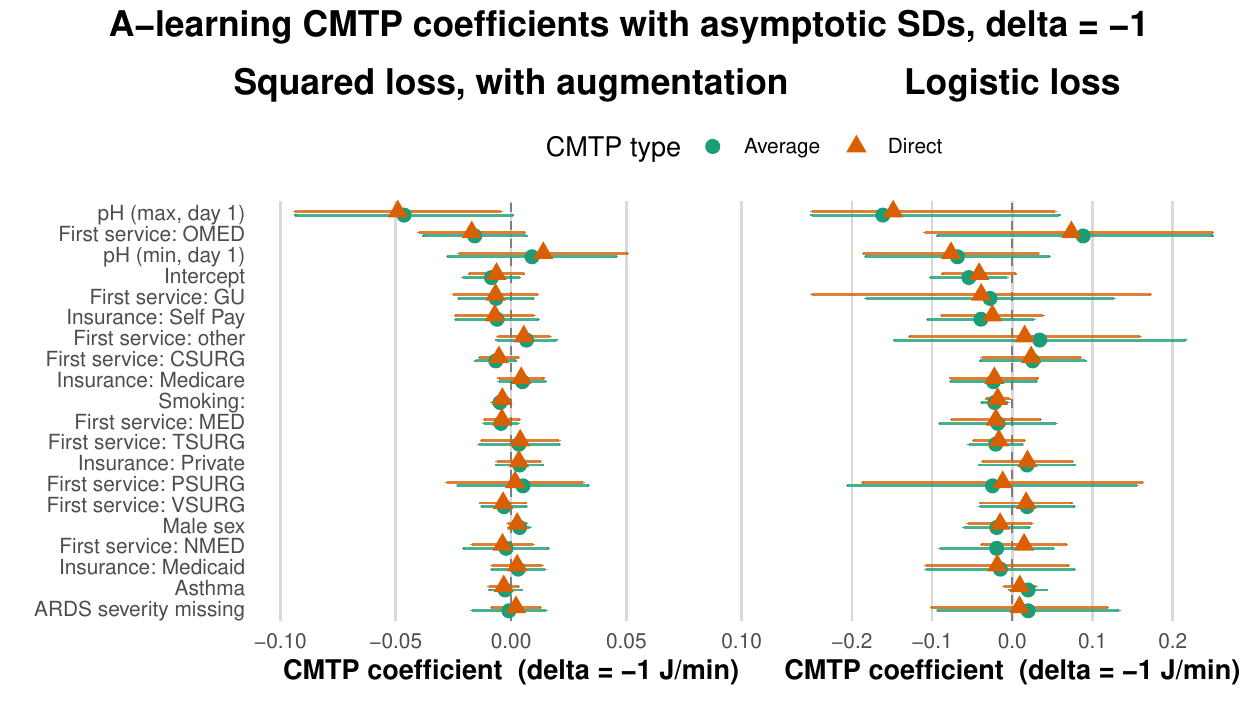}
    \includegraphics[width=0.62\linewidth]{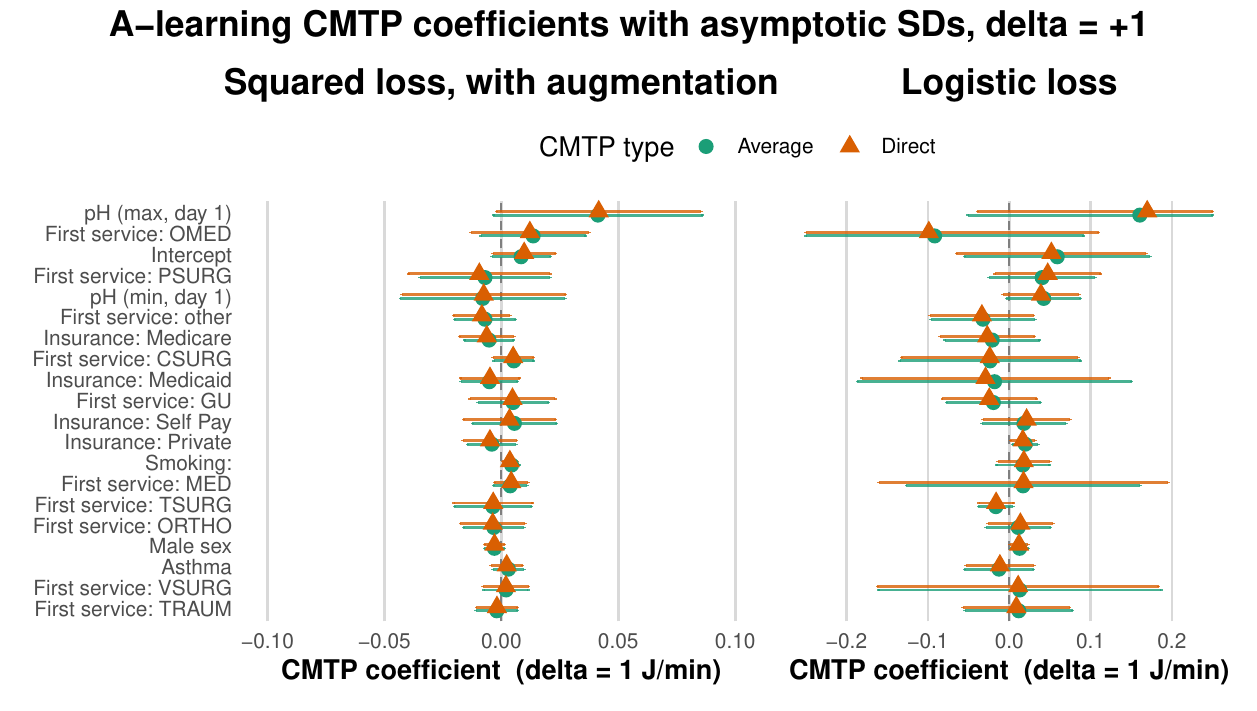}
    \caption[A-learning CMTP coefficient estimates]{{A-learning CMTP coefficient estimates with $\pm 1.96$ asymptotic standard deviation intervals, under $\delta=-1$ (top) and $\delta=+1$~J/min (bottom); squared-error loss with augmentation (left) and logistic loss (right). Colors distinguish the average and direct CMTP estimators.}}
    \label{fig:coef_est}
\end{figure}

Across all estimators and loss specifications, Figure~\ref{fig:coef_est} shows that the {intercept, which is interpretable as the predicted effect of the shift at the population-mean covariate profile, is} small in absolute magnitude and negative {under $\delta=-1$, with the sign reversing under $\delta=+1$}, indicating that a modest 1~J/min reduction in MP is associated with a slight decrease in the probability of in-hospital mortality at the population level. {The figure displays} the top 20 covariates ranked by their mean absolute CMTP coefficient across all six estimator combinations (three {loss specifications} $\times$ two {CMTP constructions}). The ranking and direction of the leading effect modifiers are broadly consistent across methods, {suggesting that} the substantive conclusions {are robust to the choice of estimator}. {Full coefficient estimates and standard errors for the A-learning CMTP estimators under all three loss specifications are reported in Supplementary Tables~S4 and~S5.} b{A-learning and overlap weighting (Supplentary Figure~S2) yield closely aligned estimates in direction and magnitude, consistent with targeting the same population projection under correct nuisance specification, whereas the density-ratio estimator (Supplementary Figure~S3) shows substantially wider spread, reflecting its greater sensitivity to density-model misspecification and the stabilization correction.}

The maximum arterial pH on ICU day 1 (\texttt{ph\_max\_day1})
emerges as the single strongest and most consistent effect modifier across all estimation strategies and both shift directions.  Under $\delta = -1$ J/min, its CMTP coefficient is negative across all methods: a higher maximum pH is associated with a \emph{greater} predicted reduction in mortality probability from lowering MP.  Under the symmetric positive shift $\delta = +1$ J/min, the coefficient reverses sign. {Clinically, a} high blood pH on day 1 often {indicates} respiratory alkalosis {from} the ventilator giving more support than needed{, so} slightly reducing the ventilator’s power lowers unnecessary strain without harming breathing. {In contrast,} patients with low pH usually need the extra ventilator support (for example, due to acid buildup or poor carbon dioxide removal){, and} reducing the power {does not} help and may even be harmful. Overall, reducing ventilator intensity may work best for patients who do not rely on it to correct serious breathing or metabolic problems.

{Supplementary} Figure~S4 displays the mean nudge effect
$\hat\tau_\delta(A_i, \bX_i) = \delta \cdot (\hat\beta_A +
\bX_i^\top \hat\beta_{AX})$ averaged within deciles of the observed MP distribution, for the A-learning squared-loss model with augmentation.  For $\delta = -1$ J/min, all deciles save the highest show negative mean nudge effects, indicating a predicted benefit from power reduction{, and} for $\delta = +1$ J/min the mirror image obtains. Because the nudge model is linear in $A$, variation across power deciles reflects systematic differences in the covariate profiles $\bX_i$ of patients who receive different power levels, rather than a {dose-response} gradient in the nudge itself{: the lowest deciles are enriched for higher pH and elective service types (most negative nudge), while the highest decile is enriched for severely ill patients with low pH and vasopressor dependence (near-zero nudge)}.

\section{Discussion}\label{sec:conclusion}

Continuous treatments arise in many clinical and policy settings, where the scientific question is often not simply whether treatment works, but how much treatment should be delivered{, and for whom}. Most existing methods for continuous treatments focus on ADRFs or CADRFs{, which} require positivity over the entire treatment range. In clinical observational studies{,} treatment values are often chosen based on patient severity, contraindications, protocols, and physician judgment, so many doses are rarely observed for certain patient profiles, and estimating the ADRF or CADRF surfaces will require extrapolation to covariate regions with little or no support.

Our framework targets local incremental effects through the conditional nudge effect and CMTP rather than global dose-response curves. Although related ideas appear in work on stochastic interventions and modified treatment policies, they have not been developed systematically for heterogeneous effects of continuous treatments. Because the estimands are anchored to observed treatment values, they require only local, shift-specific overlap rather than positivity across the full dose range. We use a duplicated-arm construction that translates the continuous shift problem to a two-arm contrast between the observed and shifted doses, enabling weighting and A-learning losses for both estimands. Under unrestricted function classes for the effects, we show that the population losses identify the causal effects. Under finite-dimensional bases for the effects, they identify corresponding projections of the nudge or CMTP onto the working model. Notably, the direct CMTP estimator requires only the shift-specific exchangeability condition (\Cref{ass:exchange}(A$'$)), which is strictly weaker than the exchangeability condition needed to identify the full nudge surface, so the direct approach is valid even when the conditional nudge effect itself is not identified. We show augmentation {can improve} precision by subtracting an outcome regression without changing the population target.

Our simulations {demonstrate} that the proposed estimators recover the relevant projected targets as sample size increases, and that augmentation generally improves precision. The mechanical ventilation analysis illustrates the practical value of
the framework. Rather than focusing on the entire CADRF, our approach can assess which patient characteristics yield benefit or harm from a modest change in MP. This type of local treatment effect question is relevant for individualized clinical decision-making.
Our asymptotic theory uses Donsker and smoothness conditions for nuisance estimators. Extensions using cross-fitting and double machine learning are straightforward and would allow more general machine-learning nuisance fits. Use of finite-dimensional bases yields projection parameters unless the working model of the effect is correctly specified, so growing spline or sieve bases are a natural approach for reducing approximation error. Finally, the shift magnitude \(\delta\) should be chosen to balance clinical relevance with empirical overlap. Future work could develop diagnostics for feasible shift selection as in \citet{jiang2025modified}, extend our framework to longitudinal
continuous treatments, and provide simultaneous inference for richer effect-modification summaries.

% \vspace*{-5 mm}
% \setlength{\bibsep}{0pt plus 0.3ex}
%\bibliographystyle{Chicago}
{\Large{\textbf{Supplementary Material}}}
\author{}
\date{}
\newline
This Supplementary Material contains additional methodological details, additional simulation and data-analysis results, and proofs of all theoretical results in the main paper. Sections, results, equations, tables, and figures in this document are numbered with the prefix ``S'' (results) or by letter (sections) to distinguish them from those of the main paper.
\def\thesection{\Alph{section}}
\setcounter{section}{0}
\section{Additional Methodological Concepts}

\subsection{Direct CMTP Estimation via Weighting}\label{direct_weighting}
Here, we work with the weighting loss for directly estimating the CMTP. This loss is analogous to the weighting loss in (2) in Section 2.3.1, integrated over the distribution of \(A_\Lambda\) given \(\bX=\bx\). For a measurable function \(f:\mathcal X\to\mathbb R\), we define the direct weighting CMTP loss
{\small\[
\ell_W^{\mathrm{CMTP}}(f)
=
\mathbb{E}\!\left[
  w(\bX,A_\Lambda,\Lambda)\,\big\{Y-\Lambda f(\bX)\big\}^2
  \;\middle|\; X
\right],
\]}
where \(w(\bX,A_\Lambda,\Lambda)\) is any nonnegative arm weight as described earlier. For fixed \(\bx\), differentiating \(\ell_W^{\mathrm{CMTP}}(f)\) with respect to \(f(\bx)\) yields the first-order condition $\mathbb{E}\big[
  w(\bX,A_\Lambda,\Lambda)\,\Lambda\{Y-\Lambda f(\bX)\}
  \,\big|\,
  \bX=\bx
\big] = 0,$
so that the pointwise minimizer takes the ratio form
{\small\[
f^*(\bx)
=
\frac{
  \mathbb{E}\big[w(\bX,A_\Lambda,\Lambda)\,\Lambda Y \mid \bX=\bx\big]
}{
  \mathbb{E}\big[w(\bX,A_\Lambda,\Lambda)\,\Lambda^2 \mid \bX=\bx\big]
}
=
\frac{N_W(\bx)}{D_W(\bx)},
\]}
where we write
\(
N_W(\bx)
:=
\mathbb{E}\big[w(\bX,A_\Lambda,\Lambda)\,\Lambda Y \mid \bX=\bx\big],
\text{ and }
D_W(\bx)
:=
\mathbb{E}\big[w(\bX,A_\Lambda,\Lambda)\,\Lambda^2 \mid \bX=\bx\big].
\)
Subsequent calculations below make these terms explicit for the overlap and density-ratio weighting schemes.

\begin{lemma}\label{lem:weighting_NxDx}
We define the arm contrast $\mu_\Delta(a,\bx) := \mu(a+\delta,\bx) - \mu(a,\bx).$ Then:
\begin{enumerate}
\item \textnormal{(Overlap weighting).}
Setting $h_{\mathrm{ovl}}(\bx,a) := \pi_\Lambda(\bx,a)\{1-\pi_\Lambda(\bx,a)\},$
we get $N_W(\bx)
=
\tfrac12\,
\bbE\bigl[h_{\mathrm{ovl}}(\bx,A_\Lambda)\,\mu_\Delta(A_\Lambda,\bx)\mid \bX=\bx\bigr],
\,
D_W(\bx)
=
\tfrac12\,
\bbE\bigl[h_{\mathrm{ovl}}(\bx,A_\Lambda)\mid \bX=\bx\bigr].$
\item \textnormal{(Density-ratio weighting).}
Setting $h_{\mathrm{DR}}(\bx,a) := 1-\pi_\Lambda(\bx,a),$ we obtain $N_W(\bx)
=
\tfrac12\,
\bbE\bigl[h_{\mathrm{DR}}(\bx,A_\Lambda)\,\mu_\Delta(A_\Lambda,\bx)\mid \bX=\bx\bigr],
\,
D_W(\bx)
=
\tfrac12\,
\bbE\bigl[h_{\mathrm{DR}}(\bx,A_\Lambda)\mid \bX=\bx\bigr].$
\end{enumerate}
In particular, in both cases the direct weighting CMTP estimator
\(
f^*(\bx) = N_W(\bx)/D_W(\bx)
\)
can be written as a weighted CMTP,
{\small\[
f^*(\bx)
=
\frac{\bbE\bigl[h_M(\bx,A_\Lambda)\,\mu_\Delta(A_\Lambda,\bx)\mid \bX=\bx\bigr]}
     {\bbE\bigl[h_M(\bx,A_\Lambda)\mid \bX=\bx\bigr]},
\]}
with $h_M = h_{\mathrm{ovl}}$ for overlap and $h_M = h_{\mathrm{DR}}$ for density-ratio weighting.
\end{lemma}
{\begin{proof}
Write $\mu_\Delta(a,\bx)=m_0(a+\delta,\bx)-m_0(a,\bx)$ with $m_0$ the observed outcome regression. Conditioning on $(A_\Lambda,\bX)=(a,\bx)$ and decomposing over the two arms as in the proof of Proposition~3, $N_W(\bx)=\bbE[\tfrac12\{w_+\pi_\Lambda\,m_0(A_\Lambda+\delta,\bX)-w_-(1-\pi_\Lambda)\,m_0(A_\Lambda,\bX)\}\mid\bX=\bx]$ and $D_W(\bx)=\bbE[\tfrac14\{w_+\pi_\Lambda+w_-(1-\pi_\Lambda)\}\mid\bX=\bx]$. Substituting the overlap weights ($w_+=1-\pi_\Lambda$, $w_-=\pi_\Lambda$) or the density-ratio weights ($w_+=(1-\pi_\Lambda)/\pi_\Lambda$, $w_-=1$) and simplifying yields the displayed expressions.
\end{proof}}
For the direct weighting estimator with method-specific implicit weight
\(h_M\), the population coefficient is the projection of the weighted CMTP
{\(\tau_{\delta,h}^{W}(\bx):=\bbE[h_M(\bx,A_\Lambda)\mu_\Delta(A_\Lambda,\bx)\mid\bX=\bx]/\bbE[h_M(\bx,A_\Lambda)\mid\bX=\bx]\)}, rather than necessarily the projection of
\(\tau_\delta(\bx)\):
$\bbeta^{\mathrm{dir},0}_{W}
    =
    \arg\min_{\bbeta}
    E\left[
    \{\tau_{\delta,h}^{W}(\bX)-\bbeta^\top c(\bX)\}^2
    \right].$
Thus, unless \(h_M(\bX,A_\Lambda)\) is constant in \(A_\Lambda\) conditional on
\(\bX\), or unless the loss is stabilized to remove the implicit weighting, the
direct weighting coefficient should be interpreted as a projection of a
weighted CMTP.

\subsection{Weighted NLL loss for Nudge Effect}\label{NLL_weighting}
We are given that $Y\mid(A,\bX)$ follow a canonical exponential family with NLL $M$ and mean function $\mu({\boldeta})=b'({\boldeta})$, and consider the duplicated-arm construction $(A_\Lambda,\Lambda)$ with $\Lambda\in\{\tfrac12,-\tfrac12\}$ as in Section 2.3. Consider a candidate link-scale contrast $f(a,\bx)$ evaluated on the duplicated-arm sample $(A_\Lambda,\Lambda)$. For any nonnegative weight function $w(a,\bx,\Lambda)$, we define the weighted NLL loss at $(a,\bx)$ by
{\small\begin{equation}\label{eq:W-NLL-loss}
\ell_W^{\mathrm{NLL}}(f;a,\bx)
=
\mathbb{E}\Big[
  w(a,\bx,\Lambda)\,
  M\bigl(Y,\Lambda f(a,\bx)\bigr)
  \,\Big|\,
 \bX = \bx,\ A_\Lambda = a
\Big].
\end{equation}}
The population FOC for $f$ can be written as
{\small\begin{equation}\label{eq:W-NLL-FOC}
\mathbb{E}\Big[
  \omega(a,\bx,\Lambda)\,\Lambda\,
  \bigl\{\mu\bigl(\Lambda f(a,\bx)\bigr) - Y\bigr\}
  \,\Big|\,
 \bX = \bx,\ A_\Lambda = a
\Big] = 0,
\end{equation}}
{where, for the canonical NLL, $\omega(a,\bx,\Lambda)=w(a,\bx,\Lambda)$ exactly, since $\partial M(y,\eta)/\partial\eta=\mu(\eta)-y$.} The next result shows that, under the duplicated-arm construction, the solution to \eqref{eq:W-NLL-FOC} again corresponds to an arm contrast at $(a,\bx)$, now on the link scale.

\begin{proposition}\label{prop:nudge_weighting_nll}
 {Suppose $\pi_\Lambda(\bx,a)\in(0,1)$, the weights $w>0$ satisfy the arm-balancing condition (3), and $\ell_W^{\mathrm{NLL}}(f;a,\bx)$ admits a minimizer $f_W^*(a,\bx)$. Then, under the conditions of Proposition~1,} for each $(a,\bx)$, we have $\tau_\delta(a,\bx) = \mathbb{E}\bigl[
  Y \mid \bX=\bx,\ A_\Lambda=a,\ \Lambda=\tfrac{1}{2}
\bigr]
-
\mathbb{E}\bigl[
  Y \mid \bX=\bx,\ A_\Lambda=a,\ \Lambda=-\tfrac{1}{2}
\bigr]$ and the minimizer $f_W^*(a,\bx)$ solves
{\small\begin{equation}\label{eq:mu-contrast-generic}
\mu\!\left(\frac{f_W^*(a,\bx)}{2}\right)
-
\mu\!\left(-\frac{f_W^*(a,\bx)}{2}\right)
= \tau_\delta(a,\bx)
.
\end{equation}}
In particular, in the logistic case $\mu({\boldeta})=\exp({\boldeta})/(1+\exp({\boldeta}))$, the left-hand side reduces to $\mu({f}/{2})-\mu(-{f}/{2})
=
{({e^{f/2}-1})/({e^{f/2}+1}),}$
and hence $f_W^*(a,\bx)$ uniquely encodes the arm-specific difference in success probabilities at $(a,\bx)$ via
\[
{\frac{e^{f_W^*(a,\bx)/2} - 1}{e^{f_W^*(a,\bx)/2} + 1}
=
\tau_\delta(a,\bx),
\qquad\text{equivalently}\qquad
f_W^*(a,\bx) = 2\log\frac{1+\tau_\delta(a,\bx)}{1-\tau_\delta(a,\bx)}.}
\]
We therefore interpret $f_W^*(a,\bx)$ as a {link-scale nudge effect} whose induced change in the mean outcome is given by the right-hand side of \eqref{eq:mu-contrast-generic}.
\end{proposition}

\noindent
As in the squared-loss case, one can subtract any function $m(a,\bx)$ of treatment and covariates only (e.g., an outcome regression) inside $M\bigl(Y,\Lambda f(a,\bx)\bigr)$ to reduce variance without changing the target of estimation, provided $m$ is treated as fixed when differentiating with respect to $f$. This merely recenters the first order conditions \eqref{eq:W-NLL-FOC} and does not impact the minimizer $f_W^*$ at the population level.

{\noindent The proof of Proposition~\ref{prop:nudge_weighting_nll} is as follows.}
\begin{proof}[Proof of Proposition~\ref{prop:nudge_weighting_nll}]
Fix $(a,\bx)$, write $\pi=\pi_\Lambda(\bx,a)$, $w_\pm=w(a,\bx,\pm\tfrac12)$, and
let $\mu_\pm(a,\bx)$ be the arm-specific means as in the proof of
Proposition~3, so that
$\mu_+(a,\bx)-\mu_-(a,\bx)=\tau_\delta(a,\bx)$ by
Proposition~1; this establishes the first claim. Since
$\partial M(y,\eta)/\partial\eta=\mu(\eta)-y$ and the natural parameter is
$\Lambda f$ with $\Lambda\in\{\pm\tfrac12\}$,
\[
\frac{\partial}{\partial f}\ell_W^{\mathrm{NLL}}(f;a,\bx)
=\tfrac12 w_+\pi\bigl\{\mu(f/2)-\mu_+(a,\bx)\bigr\}
-\tfrac12 w_-(1-\pi)\bigl\{\mu(-f/2)-\mu_-(a,\bx)\bigr\}.
\]
Under the balancing condition (3),
$w_+\pi=w_-(1-\pi)=:c(a,\bx)>0$, so the first-order condition reduces to
$\mu(f/2)-\mu(-f/2)=\mu_+(a,\bx)-\mu_-(a,\bx)=\tau_\delta(a,\bx)$, which is
\eqref{eq:mu-contrast-generic}. Moreover,
$\frac{\partial^2}{\partial f^2}\ell_W^{\mathrm{NLL}}(f;a,\bx)
=\tfrac{c(a,\bx)}{4}\bigl\{b''(f/2)+b''(-f/2)\bigr\}>0$
by strict convexity of $b$, so the conditional loss is strictly convex and the
minimizer $f_W^*(a,\bx)$, when it exists, is unique.
\end{proof}

\subsection{Direct CMTP Estimation via NLL Loss}\label{direct_NLL_CMTP}

In the exponential-family framework of Section 2.5, the direct CMTP estimators from Section 2.4 can be extended by replacing the squared-error criteria with the corresponding canonical NLL. Rather than modeling a local nudge effect at a fixed dose \(a\), we now posit a link-scale CMTP contrast \(f:\mathcal X \to \mathbb{R}\) that depends only on the covariates \(\bX\) and is evaluated under the duplicated-arm mixture for \(A_\Lambda\) given \(\bX\). In both the weighting and A-learning formulations, a natural first step is to take the corresponding nudge NLL loss (from \eqref{eq:W-NLL-loss} and (9)), simply replacing the local contrast $f(a,\bx)$ by a function $f(\bx)$ and conditioning only on $X$. However, exactly as in the squared-error case in Lemma 1 and \Cref{lem:weighting_NxDx}, the resulting first-order conditions show that the population minimizer $f^*(\bx)$ does {not} recover the CMTP $\tau_\delta(\bx)$ directly. Instead, the numerator of $f^*(\bx)$ involves an average of the arm-specific mean shift $\mu(a+\delta,\bx)-\mu(a,\bx)$ weighted by a method-specific factor $h_M(\bx,a)$ that depends on the arm propensity $\pi_\Lambda(\bx,a)$ and the curvature of the NLL, while the denominator averages the same $h_M(\bx,a)$ alone. Thus, the unstabilized NLL losses identify a {weighted} CMTP of the form
{\small\[
f^*(\bx)
=
\frac{\mathbb{E}\!\big[h_M(\bX,A_\Lambda)\{\mu(A_\Lambda+\delta,\bX)-\mu(A_\Lambda,\bX)\}\mid \bX=\bx\big]}
     {\mathbb{E}\!\big[h_M(\bX,A_\Lambda)\mid \bX=\bx\big]},
\]}
with $h_M$ taking different but explicit forms for overlap, density-ratio, and A-learning variants, rather than the unweighted CMTP $\tau_\delta(\bx) = \mathbb{E}[\mu(A+\delta,\bx)-\mu(A,x)\mid \bX=\bx]$. {To target $\tau_\delta(\bx)$ itself, we therefore stabilize the direct NLL losses by dividing each duplicated observation's contribution by $h_M(\bX,A_\Lambda)/\{1-\pi_\Lambda(\bX,A_\Lambda)\}$. The additional factor $1-\pi_\Lambda$ is essential: as in the squared-error case (see Proposition~5 of the main paper), the expectations above are taken under the duplicated-arm mixture law of $A_\Lambda$ given $\bX$, and the identity $\{1-\pi_\Lambda(\bx,a)\}f_{A_\Lambda\mid\bX}(a\mid\bx)=\tfrac12 g_{A\mid\bX}(a\mid\bx)$ converts the mixture average into an average over the natural law of $A$ given $\bX$. With this stabilization, the implicit weighting cancels in both the numerator and denominator of $f^*(\bx)$, so that the unique minimizer of the stabilized NLL loss coincides with the CMTP $\tau_\delta(\bx)$ on the appropriate scale (mean scale for squared error, link scale for NLL). For the A-learning estimator, $h_M=\pi_\Lambda(1-\pi_\Lambda)$ and the stabilizer reduces to $\pi_\Lambda(\bX,A_\Lambda)$.} As in the squared-loss setting, one can also subtract an arbitrary function $m(A,\bX)$ inside the NLL to form augmented versions of these stabilized losses. At the population level this only recenters the first order conditions and does not change the fact that $\tau_\delta(\bx)$ is the target of estimation, while yielding variance-reduction benefits in finite and large samples. {Detailed derivations of the direct CMTP losses for the logistic case may be found in the proof of Theorem~\ref{thm:logistic_direct_cmtp_asymptotic} in Section~\ref{sec:supp_proofs}.}

\subsection{Additional Empirical Loss Functions}\label{weighting_empirical}
Let $\widehat w_i^\lambda
    :=
    \widehat w(A_i-\lambda\delta,\bX_i,\lambda)${, where the third argument $\lambda\in\{0,1\}$ indicates the duplicated arm (corresponding to $\Lambda=\lambda-\tfrac12$ in the population weight of (2)).} The estimated weighting functions \(\widehat w_i^\lambda\) may be obtained from the
original sample \(\{A_i,\bX_i\}_{i=1}^n\) when density-ratio weights are used, or from the
duplicated sample when overlap or IPW weights are used. The conditional nudge effect function is obtained by minimizing
\[
    \widehat t_{\delta,W}
    =
    \argmin_{t_\delta\in\mathcal T}
    \bbP_n^\Lambda
    \widehat w_i^\lambda
    \left\{
        Y_i - {\bigl(\lambda-\tfrac12\bigr)\,}t_\delta(\bZ_i^\lambda)
    \right\}^2 .
\]
{The factor $\lambda-\tfrac12$ is the arm coding $\Lambda$ of the population loss (2); omitting it would yield an unweighted regression of $Y$ on the basis, which targets the average of the two arm means rather than their contrast.}
If \(t_\delta\) is represented in a linear basis,
$t_\delta(a,\bx)=\btheta^\top \bb(a,\bx),$
then \(\widehat t_{\delta,W}(a,\bx)=\widehat\btheta_W^\top \bb(a,\bx)\), where
\begin{align}
    \widehat{\btheta}_W
    =
    \argmin_{\btheta}
    \bbP_n^\Lambda
    \widehat w_i^\lambda
    \left\{
        Y_i-{\bigl(\lambda-\tfrac12\bigr)}\btheta^\top {\bb}_i^\lambda
    \right\}^2
    \label{eqn:weighted_sq_loss_basis_empirical}
    =
    \left\{
        {\tfrac14}\,\bbP_n^{\Lambda}
        \widehat w_i^\lambda {\bb}_i^\lambda {\bb}_i^{\lambda\top}
    \right\}^{-1}
    \left\{
        {\tfrac12}\,\bbP_n^{\Lambda}
        {(2\lambda-1)}\,\widehat w_i^\lambda {\bb}_i^\lambda Y_i
    \right\}.
    \nonumber
\end{align}
Thus the weighting estimator can be implemented by standard weighted least squares on
the duplicated sample.
\begin{lemma}[Population limits of linear-basis nudge estimators for weighting]
\label{lem:nudge-linear-projection_weighting}
Let \(\tau_\delta(a,\bx)\) denote the conditional nudge effect and suppose we approximate it
in a linear basis, $\tau_\delta(a,\bx)\approx \btheta^\top \bb(a,\bx),$
where \(\bb(a,\bx)\in\mathbb R^p\) is a fixed vector of basis functions and
\(\btheta\in\mathbb R^p\) is unknown. {Suppose the weights are strictly positive, satisfy the balancing condition (3), and $\widehat w$ is uniformly consistent for $w$; suppose also the limiting Gram matrix below is nonsingular.} Then
 \(\widehat\btheta_W\to\btheta_W^*\) in probability, where{, with $h_W(a,\bx):=\pi_\Lambda(\bx,a)\,w(a,\bx,\tfrac12)$,}
\[
    \btheta_W^*
    =
    \argmin_{\btheta\in\mathbb R^p}
    \mathbb E{^\Lambda}
    \left[
        h_W({A_\Lambda},\bX)
        \left\{
            \tau_\delta({A_\Lambda},\bX)-\btheta^\top \bb({A_\Lambda},\bX)
        \right\}^2
    \right]{,}
\]
{where $\mathbb E^\Lambda$ denotes expectation under the duplicated-arm law of $(A_\Lambda,\bX)$, as in Lemma~2 of the main paper.}
\end{lemma}
{\begin{proof}
The argument parallels the proof of Lemma~2 in the main paper: with the arm coding included, $\mathbb E[w\,(\widetilde\Lambda-\tfrac12)^2\mid A_\Lambda,\bX]=\tfrac14\{w_+\pi_\Lambda+w_-(1-\pi_\Lambda)\}=\tfrac12\,\pi_\Lambda w_+$ and $\mathbb E[w\,(\widetilde\Lambda-\tfrac12)Y\mid A_\Lambda,\bX]=\tfrac12\,\pi_\Lambda w_+\,\tau_\delta(A_\Lambda,\bX)$ under the balancing condition and the conditions of Proposition~1; the limit then follows from the law of large numbers, uniform consistency of $\widehat w$, and the continuous mapping theorem.
\end{proof}}
\noindent In particular, for overlap weights \(w\), one has
$h_W(a,\bx)=\pi_\Lambda({\bx,a})\{1-\pi_\Lambda({\bx,a})\},$
so the A-learning and overlap-weighting estimators project \(\tau_\delta\) onto the same
basis with the same weight function{; for IPW weights $h_W\equiv 1$, and for density-ratio weights $h_W=1-\pi_\Lambda$}.

\section{\texorpdfstring{{Additional Simulation Results}}{Additional Simulation Results}}\label{sec:supp_sim}
First, we report the inference results of nudge estimation below in Table~\ref{tab:linear_nudge}. Next, in \Cref{subsec:setup2} we introduce a new simulation experiment that studies the performance of the proposed methods under the logistic loss. 

\begin{table}[!htbp]
\centering
\footnotesize
\caption{Linear outcome, Nudge coefficients: empirical bias, RMSE, and 95\% Wald coverage.}
\label{tab:linear_nudge}
\resizebox{\linewidth}{!}{%
\begin{tabular}{llllrrrrrrrrr}
\toprule
 & & & & \multicolumn{3}{c}{$n=500$} & \multicolumn{3}{c}{$n=1{,}000$} & \multicolumn{3}{c}{$n=5{,}000$} \\
\cmidrule(lr){5-7} \cmidrule(lr){8-10} \cmidrule(lr){11-13}
Coef. & Truth & Weighting & Aug. & Bias & RMSE & Cov. & Bias & RMSE & Cov. & Bias & RMSE & Cov. \\
\midrule
\multirow{6}{*}{$A$} & \multirow{6}{*}{-0.0600} & \multirow{2}{*}{A-learning} & No  & 0.0049 & 0.0173 & 0.8450 & 0.0022 & 0.0135 & 0.8850 & 0.0004 & 0.0065 & 0.9100 \\
 & & & Yes & -0.0001 & 0.0016 & 0.9950 & -0.0000 & 0.0011 & 0.9750 & -0.0000 & 0.0005 & 0.9700 \\
\cmidrule(lr){3-13}
 & & \multirow{2}{*}{Density Ratio} & No  & 0.0049 & 0.0173 & 0.8450 & 0.0022 & 0.0136 & 0.8850 & 0.0004 & 0.0065 & 0.9100 \\
 & & & Yes & -0.0001 & 0.0016 & 0.9900 & -0.0000 & 0.0012 & 0.9750 & -0.0000 & 0.0005 & 0.9700 \\
\cmidrule(lr){3-13}
 & & \multirow{2}{*}{Overlap} & No  & 0.0049 & 0.0173 & 0.8450 & 0.0022 & 0.0135 & 0.8850 & 0.0004 & 0.0065 & 0.9100 \\
 & & & Yes & -0.0001 & 0.0016 & 0.9900 & -0.0000 & 0.0011 & 0.9750 & -0.0000 & 0.0005 & 0.9700 \\
\midrule
\multirow{6}{*}{$A{:}X_1$} & \multirow{6}{*}{0.0800} & \multirow{2}{*}{A-learning} & No  & -0.0048 & 0.0288 & 0.7750 & -0.0022 & 0.0183 & 0.8350 & -0.0004 & 0.0093 & 0.9100 \\
 & & & Yes & -0.0001 & 0.0005 & 0.9700 & 0.0000 & 0.0004 & 0.9200 & 0.0000 & 0.0001 & 0.9650 \\
\cmidrule(lr){3-13}
 & & \multirow{2}{*}{Density Ratio} & No  & -0.0048 & 0.0288 & 0.7750 & -0.0022 & 0.0183 & 0.8350 & -0.0004 & 0.0093 & 0.9100 \\
 & & & Yes & -0.0001 & 0.0005 & 0.9700 & 0.0000 & 0.0004 & 0.9250 & 0.0000 & 0.0001 & 0.9650 \\
\cmidrule(lr){3-13}
 & & \multirow{2}{*}{Overlap} & No  & -0.0048 & 0.0288 & 0.7750 & -0.0022 & 0.0183 & 0.8350 & -0.0004 & 0.0093 & 0.9100 \\
 & & & Yes & -0.0001 & 0.0005 & 0.9700 & 0.0000 & 0.0004 & 0.9200 & 0.0000 & 0.0001 & 0.9650 \\
\midrule
\multirow{6}{*}{$A{:}X_2$} & \multirow{6}{*}{0.0000} & \multirow{2}{*}{A-learning} & No  & -0.0015 & 0.0211 & 0.9400 & -0.0010 & 0.0153 & 0.9600 & -0.0001 & 0.0083 & 0.9250 \\
 & & & Yes & 0.0002 & 0.0016 & 0.9850 & 0.0001 & 0.0011 & 0.9700 & 0.0000 & 0.0005 & 0.9650 \\
\cmidrule(lr){3-13}
 & & \multirow{2}{*}{Density Ratio} & No  & -0.0015 & 0.0212 & 0.9450 & -0.0010 & 0.0154 & 0.9600 & -0.0001 & 0.0083 & 0.9250 \\
 & & & Yes & 0.0002 & 0.0016 & 0.9850 & 0.0001 & 0.0012 & 0.9700 & 0.0000 & 0.0005 & 0.9650 \\
\cmidrule(lr){3-13}
 & & \multirow{2}{*}{Overlap} & No  & -0.0015 & 0.0211 & 0.9400 & -0.0010 & 0.0153 & 0.9600 & -0.0001 & 0.0083 & 0.9250 \\
 & & & Yes & 0.0002 & 0.0016 & 0.9850 & 0.0001 & 0.0011 & 0.9700 & 0.0000 & 0.0005 & 0.9650 \\
\midrule
\multirow{6}{*}{$X_1$} & \multirow{6}{*}{0.0580} & \multirow{2}{*}{A-learning} & No  & 0.0487 & 0.3016 & 0.8000 & 0.0252 & 0.1874 & 0.8400 & 0.0045 & 0.0958 & 0.8950 \\
 & & & Yes & 0.0018 & 0.0165 & 0.9700 & -0.0017 & 0.0136 & 0.9300 & -0.0005 & 0.0055 & 0.9350 \\
\cmidrule(lr){3-13}
 & & \multirow{2}{*}{Density Ratio} & No  & 0.0486 & 0.3012 & 0.8000 & 0.0251 & 0.1872 & 0.8400 & 0.0044 & 0.0957 & 0.8950 \\
 & & & Yes & 0.0018 & 0.0167 & 0.9650 & -0.0017 & 0.0136 & 0.9350 & -0.0005 & 0.0056 & 0.9350 \\
\cmidrule(lr){3-13}
 & & \multirow{2}{*}{Overlap} & No  & 0.0486 & 0.3016 & 0.8000 & 0.0252 & 0.1874 & 0.8400 & 0.0045 & 0.0958 & 0.8950 \\
 & & & Yes & 0.0018 & 0.0165 & 0.9700 & -0.0017 & 0.0136 & 0.9300 & -0.0005 & 0.0055 & 0.9350 \\
\midrule
\multirow{6}{*}{$X_2$} & \multirow{6}{*}{0.0000} & \multirow{2}{*}{A-learning} & No  & -0.0267 & 0.2856 & 0.8350 & -0.0062 & 0.1858 & 0.9250 & -0.0052 & 0.1026 & 0.9100 \\
 & & & Yes & -0.0034 & 0.0279 & 0.9800 & -0.0018 & 0.0201 & 0.9600 & -0.0002 & 0.0079 & 0.9650 \\
\cmidrule(lr){3-13}
 & & \multirow{2}{*}{Density Ratio} & No  & -0.0268 & 0.2865 & 0.8350 & -0.0062 & 0.1860 & 0.9250 & -0.0052 & 0.1028 & 0.9050 \\
 & & & Yes & -0.0034 & 0.0277 & 0.9850 & -0.0017 & 0.0203 & 0.9600 & -0.0002 & 0.0079 & 0.9650 \\
\cmidrule(lr){3-13}
 & & \multirow{2}{*}{Overlap} & No  & -0.0267 & 0.2857 & 0.8350 & -0.0062 & 0.1858 & 0.9250 & -0.0052 & 0.1026 & 0.9100 \\
 & & & Yes & -0.0034 & 0.0279 & 0.9800 & -0.0018 & 0.0201 & 0.9600 & -0.0002 & 0.0079 & 0.9650 \\
\midrule
\multirow{6}{*}{Int.} & \multirow{6}{*}{-0.1060} & \multirow{2}{*}{A-learning} & No  & -0.0253 & 0.1042 & 0.8950 & -0.0116 & 0.0813 & 0.9250 & -0.0018 & 0.0397 & 0.8900 \\
 & & & Yes & 0.0019 & 0.0257 & 0.9900 & 0.0007 & 0.0182 & 0.9750 & 0.0005 & 0.0076 & 0.9700 \\
\cmidrule(lr){3-13}
 & & \multirow{2}{*}{Density Ratio} & No  & -0.0253 & 0.1043 & 0.8950 & -0.0117 & 0.0815 & 0.9250 & -0.0018 & 0.0398 & 0.8850 \\
 & & & Yes & 0.0019 & 0.0256 & 0.9900 & 0.0007 & 0.0184 & 0.9700 & 0.0005 & 0.0076 & 0.9700 \\
\cmidrule(lr){3-13}
 & & \multirow{2}{*}{Overlap} & No  & -0.0253 & 0.1043 & 0.8950 & -0.0116 & 0.0813 & 0.9250 & -0.0018 & 0.0397 & 0.8900 \\
 & & & Yes & 0.0019 & 0.0257 & 0.9900 & 0.0007 & 0.0182 & 0.9750 & 0.0005 & 0.0076 & 0.9700 \\
\bottomrule
\end{tabular}}
\end{table}
\subsection{Logistic Loss}\label{subsec:setup2}
In a second simulation, we consider a binary outcome under a correctly specified logistic model, so effects are defined on the log-odds scale. We evaluate whether the proposed weighting and A-learning estimators of the conditional nudge and CMTP effects remain approximately unbiased under a nonlinear link, and compare direct CMTP estimation with averaging the estimated nudge under logistic loss. {Unlike in} the linear setting, we omit outcome-regression augmentation so that differences reflect only the loss and CMTP construction.
To facilitate comparison with Setup 1, we retain the same covariate and
treatment distributions. We generate i.i.d. covariates
\(\bX=(X_1,X_2)^\top\), where
\(X_1\sim\mathrm{Unif}(-1,1)\) and
\(X_2\sim\mathrm{Bernoulli}(0.7)\), independently. Conditional on \(\bX\),
\(A\mid\bX\sim\mathrm{Exp}\{\lambda(\bX)\}\), where
\(\lambda(\bX)=\exp\{-(2+0.5X_1+X_2)\}\), so that
\(\mathbb E[A\mid\bX=\bx]=\exp(2+0.5x_1+x_2)\). Given \((A,\bX)\), we generate
\(Y\mid(A,\bX)\sim\mathrm{Bernoulli}\{\mu(A,\bX)\}\), with
\(\mu(A,\bX)=\operatorname{expit}\{\eta(A,\bX)\}\), where
\(\eta(A,\bX)\) has the form as in Section 5, with
\(\delta=0.5\), \(\alpha=-1.5\), \(\beta=2.0\),
\(\beta_A=-0.50\), and \(\beta_{AX}=0.25\). This model specification induces a nonlinear, covariate-modified treatment effect on the log-odds scale while a parametric structure that allows for closed-form nudge and CMTP effects.

The conditional nudge effect is defined on the log-odds scale as
\(\tau_\delta(a,\bx)=\eta(a+\delta,\bx)-\eta(a,\bx)\), corresponding to an odds ratio \(\exp\{\tau_\delta(a,\bx)\}\) under the shift
\(a\mapsto a+\delta\). 
Under a polynomial working model, this contrast has the same form as that in Section 5 and with coefficients interpreted on the log-odds scale. The CMTP effect averages this contrast over the observed conditional treatment distribution, namely 
\(\tau_\delta(\bx)=\mathbb E\{\tau_\delta(A,\bX)\mid\bX=\bx\}\), where \(\mathbb E[A\mid\bX=\bx]=\exp(2+0.5x_1+x_2)\). As in the linear setting, we
summarize \(\tau_\delta(\bx)\) by its least-squares projection onto the span of
\(\{1,X_1,X_2\}\). The resulting coefficients then define the  CMTP parameter on the log-odds scale.

The estimators here parallel those in Setup 1 but use logistic rather than squared error loss and for simplicity regression augmentation. Density, Overlap, and A-learning approaches are each fit under a linear working model for the log-odds contrast{, and for the CMTP we again compare the ``average'' and ``direct'' constructions, yielding six distinct CMTP estimators}. As in Setup 1, for \(n\in\{500,1000,5000\}\) we generate \(R=200\) independent datasets and compute empirical bias and root mean squared error relative to the known polynomial nudge coefficients and projected CMTP coefficients.

{Figure~\ref{fig:logistic_nudge_cmtp} displays the root integrated mean squared error of the estimated nudge (left) and CMTP (right) functions: error decreases with $n$ for all methods, and the direct CMTP estimators attain uniformly lower error than their averaging counterparts in this design. Supplementary Tables~\ref{tab:logistic_nudge} and \ref{tab:logistic_cmtp} report coefficient-level bias, RMSE and 95\% Wald coverage for the logistic estimators under the nudge and CMTP models.}

\begin{figure}
    \centering
    \includegraphics[width=0.8\linewidth]{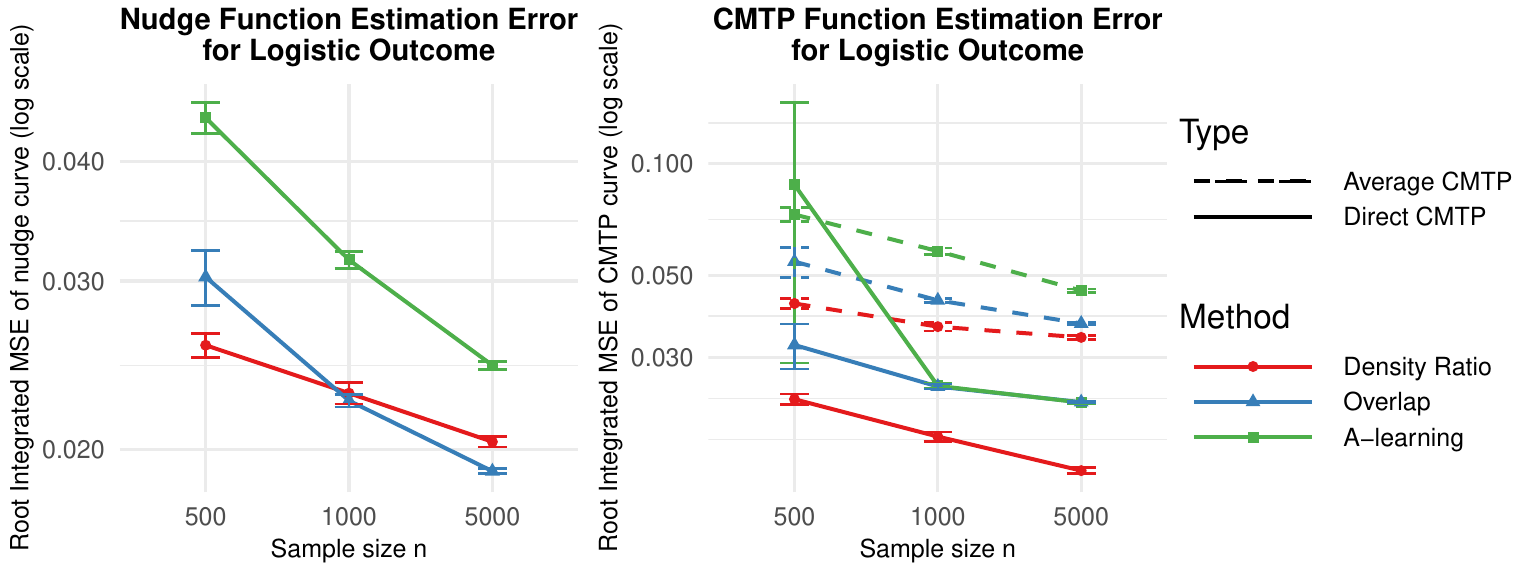}
    \caption{{Root integrated MSE (log scale) of the estimated nudge (left) and CMTP (right) functions for Setup 2. Colors distinguish the three methods; line type separates the average and direct CMTP estimators.}}
    \label{fig:logistic_nudge_cmtp}
\end{figure}

\begin{table}[!htbp]
\centering
\footnotesize
\caption{Logistic outcome, Nudge coefficients: empirical bias, RMSE, and 95\% Wald coverage.}
\label{tab:logistic_nudge}
\resizebox{\textwidth}{!}{%
\begin{tabular}{lllrrrrrrrrr}
\toprule
 & & & \multicolumn{3}{c}{$n=500$} & \multicolumn{3}{c}{$n=1{,}000$} & \multicolumn{3}{c}{$n=5{,}000$} \\
\cmidrule(lr){4-6} \cmidrule(lr){7-9} \cmidrule(lr){10-12}
Coef. & Truth & Weighting & Bias & RMSE & Cov. & Bias & RMSE & Cov. & Bias & RMSE & Cov. \\
\midrule
\multirow{3}{*}{$A$} & \multirow{3}{*}{-0.0060} & A-learning & 0.0004 & 0.0119 & 0.9550 & 0.0002 & 0.0079 & 0.9500 & 0.0001 & 0.0035 & 0.9750 \\
 &  & Density Ratio & 0.0005 & 0.0121 & 0.9650 & 0.0003 & 0.0080 & 0.9550 & 0.0001 & 0.0036 & 0.9750 \\
 &  & Overlap & 0.0005 & 0.0121 & 0.9650 & 0.0003 & 0.0080 & 0.9550 & 0.0001 & 0.0036 & 0.9750 \\
\midrule
\multirow{3}{*}{$A{:}X_1$} & \multirow{3}{*}{0.0080} & A-learning & 0.0007 & 0.0124 & 0.9650 & 0.0001 & 0.0083 & 0.9650 & 0.0001 & 0.0039 & 0.9800 \\
 &  & Density Ratio & 0.0006 & 0.0128 & 0.9650 & 0.0000 & 0.0084 & 0.9800 & 0.0001 & 0.0039 & 0.9800 \\
 &  & Overlap & 0.0006 & 0.0127 & 0.9700 & 0.0000 & 0.0084 & 0.9800 & 0.0001 & 0.0039 & 0.9800 \\
\midrule
\multirow{3}{*}{$A{:}X_2$} & \multirow{3}{*}{0.0000} & A-learning & -0.0009 & 0.0081 & 0.9650 & -0.0003 & 0.0053 & 0.9500 & -0.0001 & 0.0022 & 0.9500 \\
 &  & Density Ratio & -0.0010 & 0.0082 & 0.9650 & -0.0003 & 0.0054 & 0.9500 & -0.0001 & 0.0022 & 0.9500 \\
 &  & Overlap & -0.0010 & 0.0082 & 0.9650 & -0.0003 & 0.0053 & 0.9500 & -0.0001 & 0.0022 & 0.9500 \\
\midrule
\multirow{3}{*}{$X_1$} & \multirow{3}{*}{0.0058} & A-learning & -0.0019 & 0.0828 & 0.9250 & -0.0002 & 0.0548 & 0.9550 & -0.0010 & 0.0253 & 0.9550 \\
 &  & Density Ratio & -0.0004 & 0.0848 & 0.9300 & 0.0006 & 0.0558 & 0.9550 & -0.0009 & 0.0255 & 0.9650 \\
 &  & Overlap & -0.0004 & 0.0847 & 0.9300 & 0.0006 & 0.0557 & 0.9550 & -0.0009 & 0.0254 & 0.9650 \\
\midrule
\multirow{3}{*}{$X_2$} & \multirow{3}{*}{0.0000} & A-learning & 0.0099 & 0.0904 & 0.9450 & 0.0019 & 0.0657 & 0.9250 & 0.0015 & 0.0269 & 0.9500 \\
 &  & Density Ratio & 0.0112 & 0.0923 & 0.9500 & 0.0024 & 0.0664 & 0.9300 & 0.0016 & 0.0270 & 0.9500 \\
 &  & Overlap & 0.0112 & 0.0923 & 0.9500 & 0.0024 & 0.0664 & 0.9300 & 0.0016 & 0.0270 & 0.9500 \\
\midrule
\multirow{3}{*}{Int.} & \multirow{3}{*}{-0.0106} & A-learning & -0.0074 & 0.0885 & 0.9500 & -0.0020 & 0.0609 & 0.9250 & -0.0011 & 0.0262 & 0.9550 \\
 &  & Density Ratio & -0.0085 & 0.0900 & 0.9500 & -0.0025 & 0.0616 & 0.9250 & -0.0012 & 0.0264 & 0.9550 \\
 &  & Overlap & -0.0085 & 0.0899 & 0.9500 & -0.0025 & 0.0615 & 0.9250 & -0.0012 & 0.0263 & 0.9550 \\
\bottomrule
\end{tabular}%
}
\end{table}

\begin{table}[!htbp]
\centering
\footnotesize
\caption{Logistic outcome, CMTP coefficients: empirical bias, RMSE, and 95\% Wald coverage.}
\label{tab:logistic_cmtp}
\resizebox{\textwidth}{!}{%
\begin{tabular}{llllrrrrrrrrr}
\toprule
 & & & & \multicolumn{3}{c}{$n=500$} & \multicolumn{3}{c}{$n=1{,}000$} & \multicolumn{3}{c}{$n=5{,}000$} \\
\cmidrule(lr){5-7} \cmidrule(lr){8-10} \cmidrule(lr){11-13}
Coef. & Truth & Constr. & Weighting & Bias & RMSE & Cov. & Bias & RMSE & Cov. & Bias & RMSE & Cov. \\
\midrule
\multirow{6}{*}{$X_1$} & \multirow{6}{*}{0.0959} & \multirow{3}{*}{Average} & A-learning & 0.0061 & 0.0745 & 0.9800 & 0.0016 & 0.0506 & 0.9650 & 0.0000 & 0.0234 & 0.9800 \\
 &  &  & Density Ratio & 0.0062 & 0.0769 & 0.9800 & 0.0016 & 0.0517 & 0.9700 & 0.0001 & 0.0238 & 0.9850 \\
 &  &  & Overlap & 0.0062 & 0.0767 & 0.9800 & 0.0016 & 0.0516 & 0.9700 & 0.0001 & 0.0238 & 0.9850 \\
\cmidrule(lr){3-13}
 &  & \multirow{3}{*}{Direct} & A-learning & 0.0061 & 0.0745 & 0.9800 & 0.0016 & 0.0506 & 0.9650 & 0.0000 & 0.0234 & 0.9800 \\
 &  &  & Density Ratio & 0.0062 & 0.0769 & 0.9800 & 0.0016 & 0.0517 & 0.9700 & 0.0001 & 0.0238 & 0.9850 \\
 &  &  & Overlap & 0.0062 & 0.0767 & 0.9800 & 0.0016 & 0.0516 & 0.9700 & 0.0001 & 0.0238 & 0.9850 \\
\midrule
\multirow{6}{*}{$X_2$} & \multirow{6}{*}{-0.0621} & \multirow{3}{*}{Average} & A-learning & -0.0036 & 0.0932 & 0.9700 & -0.0005 & 0.0598 & 0.9800 & 0.0000 & 0.0278 & 0.9700 \\
 &  &  & Density Ratio & -0.0030 & 0.0954 & 0.9700 & -0.0001 & 0.0609 & 0.9800 & 0.0000 & 0.0282 & 0.9700 \\
 &  &  & Overlap & -0.0031 & 0.0952 & 0.9700 & -0.0001 & 0.0608 & 0.9800 & 0.0000 & 0.0281 & 0.9700 \\
\cmidrule(lr){3-13}
 &  & \multirow{3}{*}{Direct} & A-learning & -0.0036 & 0.0932 & 0.9700 & -0.0005 & 0.0598 & 0.9800 & 0.0000 & 0.0278 & 0.9700 \\
 &  &  & Density Ratio & -0.0030 & 0.0954 & 0.9700 & -0.0001 & 0.0609 & 0.9800 & 0.0000 & 0.0282 & 0.9700 \\
 &  &  & Overlap & -0.0031 & 0.0952 & 0.9700 & -0.0001 & 0.0608 & 0.9800 & 0.0000 & 0.0281 & 0.9700 \\
\midrule
\multirow{6}{*}{Int.} & \multirow{6}{*}{-0.0466} & \multirow{3}{*}{Average} & A-learning & -0.0037 & 0.0228 & 0.9750 & -0.0007 & 0.0176 & 0.9600 & -0.0006 & 0.0073 & 0.9650 \\
 &  &  & Density Ratio & -0.0042 & 0.0234 & 0.9800 & -0.0009 & 0.0178 & 0.9600 & -0.0007 & 0.0074 & 0.9650 \\
 &  &  & Overlap & -0.0042 & 0.0235 & 0.9800 & -0.0009 & 0.0179 & 0.9600 & -0.0007 & 0.0075 & 0.9650 \\
\cmidrule(lr){3-13}
 &  & \multirow{3}{*}{Direct} & A-learning & -0.0037 & 0.0228 & 0.9750 & -0.0007 & 0.0176 & 0.9600 & -0.0006 & 0.0073 & 0.9650 \\
 &  &  & Density Ratio & -0.0042 & 0.0234 & 0.9800 & -0.0009 & 0.0178 & 0.9600 & -0.0007 & 0.0074 & 0.9650 \\
 &  &  & Overlap & -0.0042 & 0.0235 & 0.9800 & -0.0009 & 0.0179 & 0.9600 & -0.0007 & 0.0075 & 0.9650 \\
\bottomrule
\end{tabular}%
}
\end{table}

\section{\texorpdfstring{{Additional Details and Results for the Mechanical Power Data Analysis}}{Additional Details and Results for the Mechanical Power Data Analysis}}\label{sec:supp_rda}

{This section provides the full list of candidate effect modifiers used in Section~6 of the main paper, coefficient estimates under the overlap and density-ratio weighting schemes (Figures~\ref{fig:overlap} and \ref{fig:dr}), and full coefficient tables with standard errors for the A-learning CMTP estimators (Tables~\ref{tab:al_cmtp_delta-1} and \ref{tab:al_cmtp_delta1}).}

\paragraph{Effect modifiers.}
We consider a rich set of $p = 24$ pre-treatment covariates
$\bX_i \in \mathcal{X} \subset \mathbb{R}^p$ measured on ICU day 1,
comprising:
\begin{itemize}
  \item \emph{Demographics and severity:} sex, Simplified Acute
    Physiology Score II (SAPS~II), smoking history.
  \item \emph{Respiratory mechanics:} minimum and maximum respiratory
    rate, minimum and maximum peak inspiratory pressure, plateau
    pressure, fraction of inspired oxygen (FiO$_2$).
  \item \emph{Gas exchange:} minimum and maximum partial pressure of
    arterial oxygen (PaO$_2$), PaO$_2$/FiO$_2$ ratio, minimum and
    maximum partial pressure of arterial carbon dioxide (PaCO$_2$),
    minimum and maximum arterial pH.
  \item \emph{Haemodynamics and organ support:} minimum mean arterial
    blood pressure, vasopressor use on day 1.
  \item \emph{Admission characteristics:} primary ICU service
    (cardiac surgery, cardiac medical, general medical, neurological
    medical, neurological surgical, oncological medical, orthopaedic,
    plastic surgery, thoracic surgery, trauma, vascular surgery, and
    other), insurance category (Medicaid, Medicare, private,
    self-pay), ARDS severity classification and its missingness
    indicator.
  \item \emph{Comorbidities:} asthma.
\end{itemize}
Categorical variables were expanded into $K-1$ dummy variable columns and constant columns were removed. Continuous covariates were
mean-centered prior to analysis so that all CMTP coefficients represent deviations from the population-mean covariate profile and the model intercept is interpretable as the predicted nudge effect at the population-mean.

\begin{figure}[htbp]
    \centering
    \includegraphics[width=\linewidth]{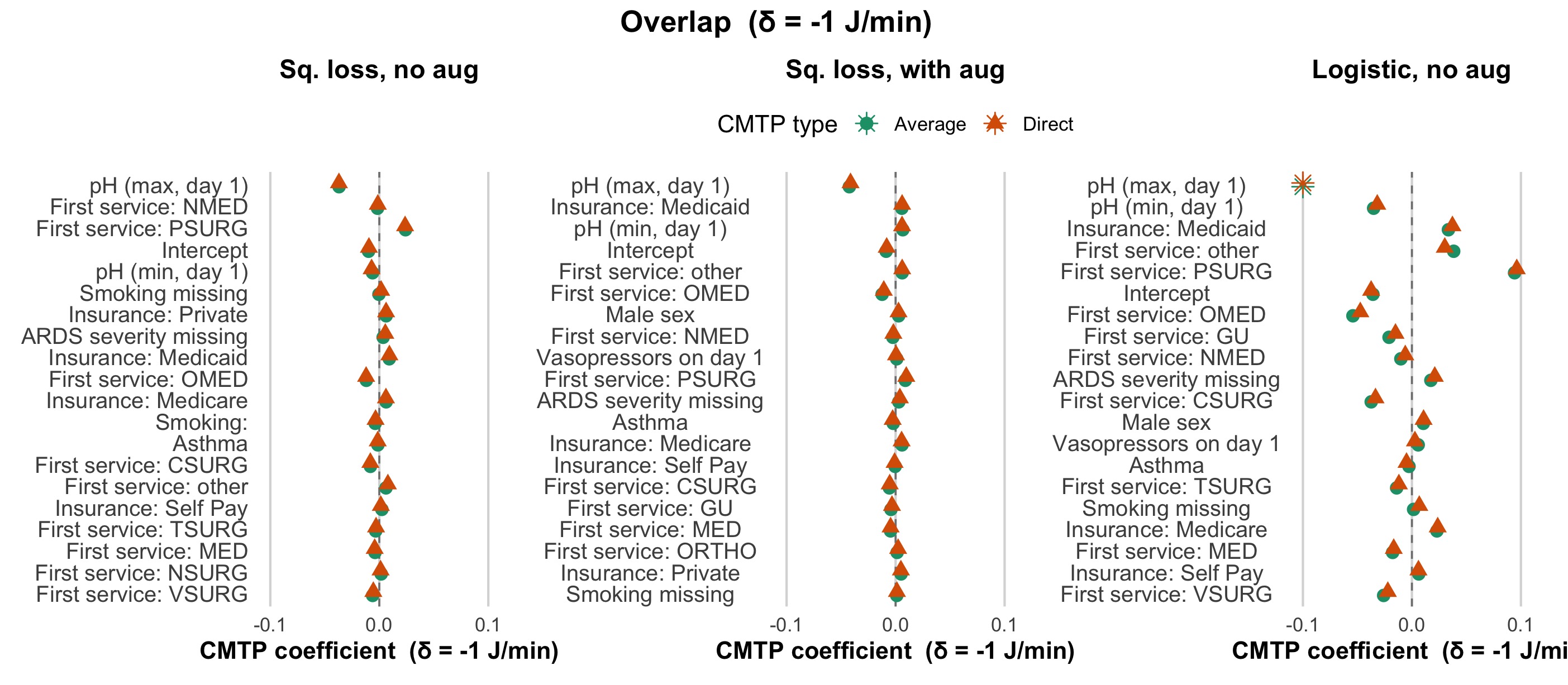}
    \includegraphics[width=\linewidth]{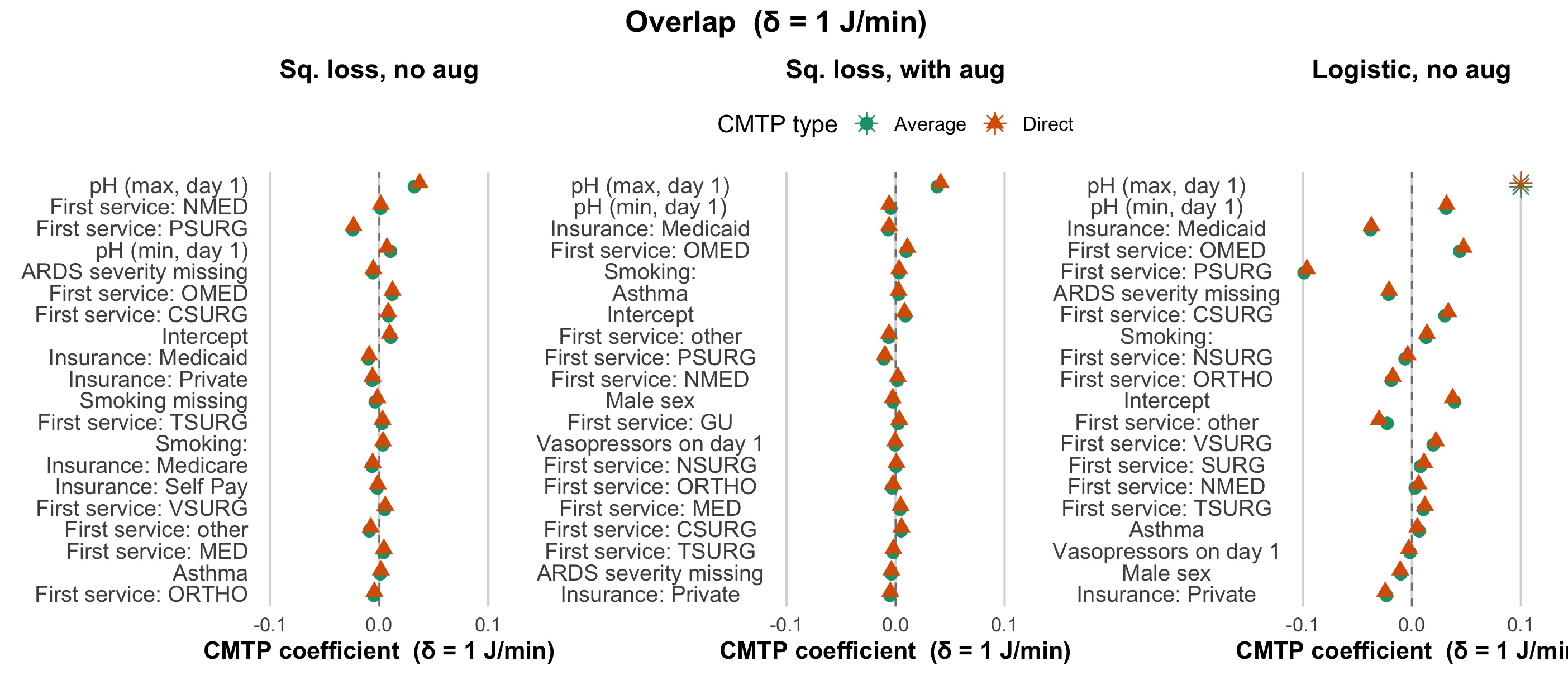}
    \caption[CMTP coefficient estimates under overlap weighting]{{CMTP coefficient estimates for the mechanical power analysis under overlap weighting, for $\delta=-1$~J/min (top) and $\delta=+1$~J/min (bottom). Panels correspond to the squared-error loss without and with augmentation and the logistic loss; colors distinguish the average (plug-in) and direct CMTP estimators.}}
    \label{fig:overlap}
\end{figure}

\begin{figure}[htbp]
    \centering
    \includegraphics[width=\linewidth]{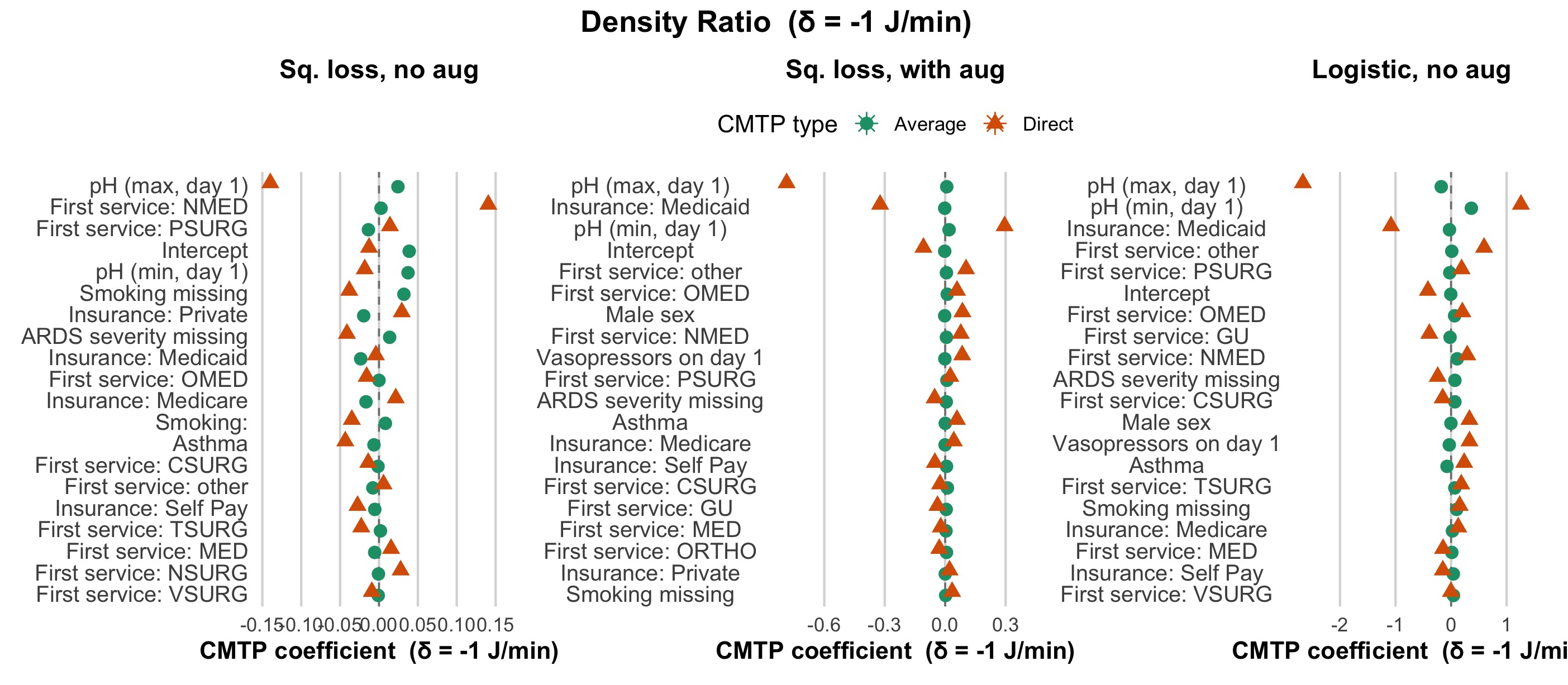}
    \includegraphics[width=\linewidth]{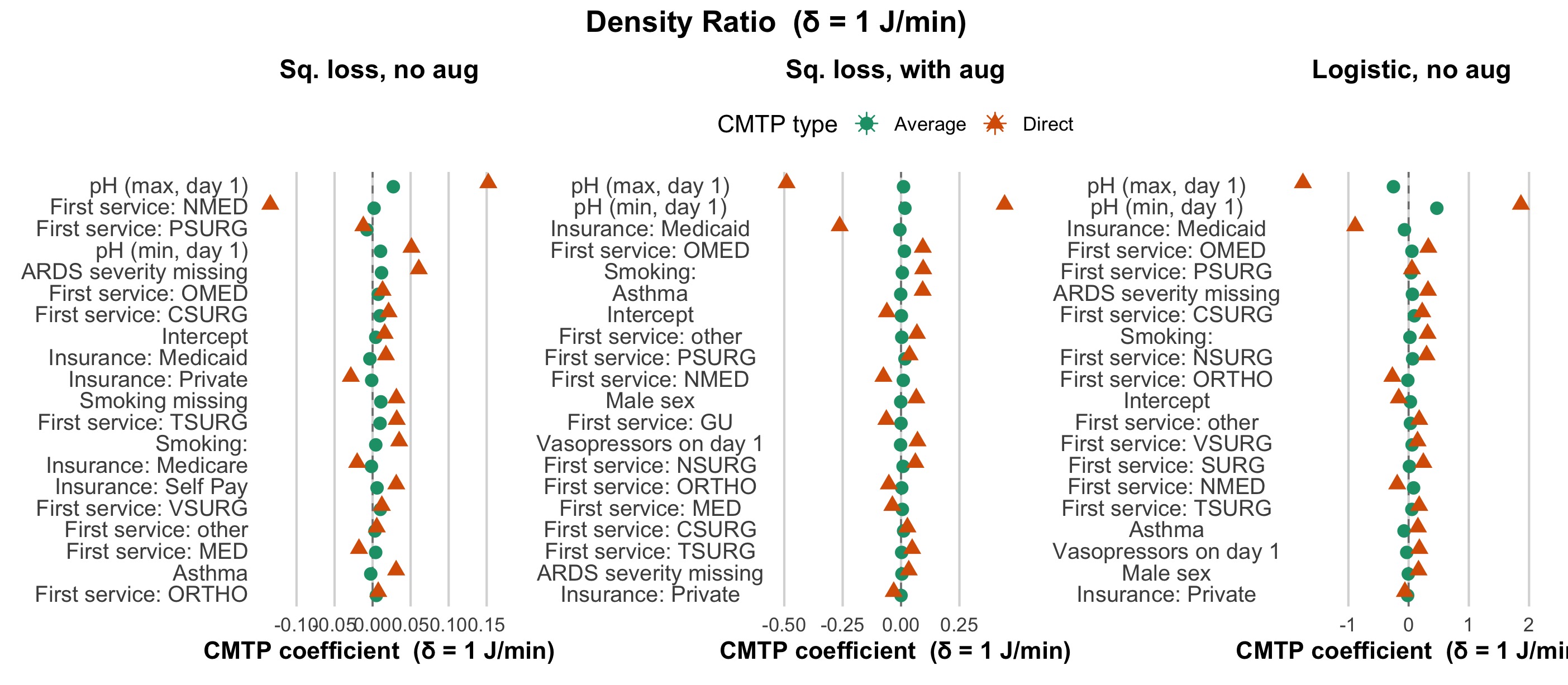}
    \caption[CMTP coefficient estimates under density-ratio weighting]{{CMTP coefficient estimates for the mechanical power analysis under density-ratio weighting, for $\delta=-1$~J/min (top) and $\delta=+1$~J/min (bottom). Panels correspond to the squared-error loss without and with augmentation and the logistic loss; colors distinguish the average (plug-in) and direct CMTP estimators.}}
    \label{fig:dr}
\end{figure}

\begin{figure}[htbp]
    \centering
    \includegraphics[width=\linewidth]{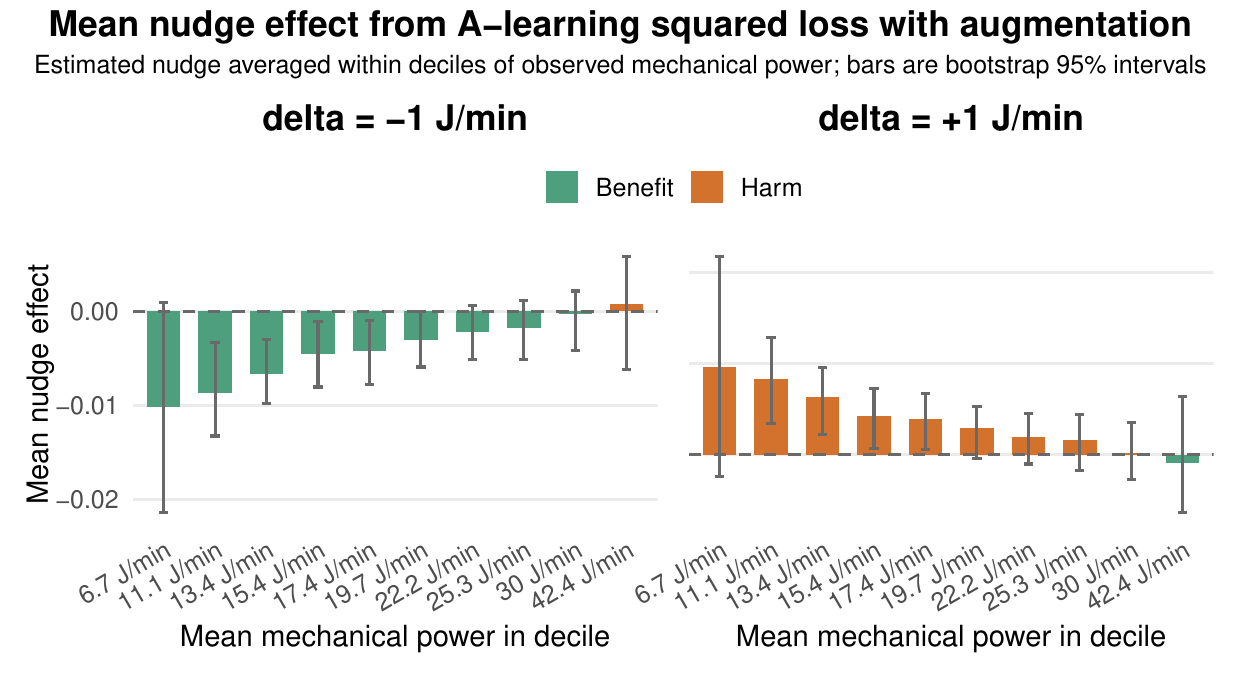}
    \caption{{Mean estimated nudge effect within deciles of observed mechanical power (A-learning, squared loss with augmentation), under $\delta=-1$ (left) and $\delta=+1$~J/min (right); negative values indicate predicted mortality reduction.}}
    \label{fig:benefit_vs_harm}
\end{figure}

\begin{sidewaystable}[htbp]
  \centering
  \small
  \caption{A-learning CMTP coefficients, $\delta = -1$ J/min. Format: estimate (SE). SEs are plug-in cluster-sandwich asymptotic standard deviations from the estimating-equation implementation used for the coefficient plots.}
  \label{tab:al_cmtp_delta-1}

  \resizebox{0.8\textheight}{!}{%

\begin{tabular}{lrrrrrr}
\toprule
\multicolumn{1}{c}{ } & \multicolumn{2}{c}{Squared Loss (no aug)} & \multicolumn{2}{c}{Squared Loss (with aug)} & \multicolumn{2}{c}{Logistic Loss} \\
\cmidrule(l{3pt}r{3pt}){2-3} \cmidrule(l{3pt}r{3pt}){4-5} \cmidrule(l{3pt}r{3pt}){6-7}
Covariate & Average & Direct & Average & Direct & Average & Direct\\
\midrule
pH (max, day 1) & -0.0372 (0.0299) & -0.0455 (0.0282) & -0.0464 (0.0241) & -0.0491 (0.0229) & -0.1614 (0.1121) & -0.1485 (0.1030)\\
First service: OMED & -0.0157 (0.0181) & -0.0173 (0.0170) & -0.0158 (0.0114) & -0.0170 (0.0117) & -0.0684 (0.0587) & -0.0764 (0.0560)\\
First service: PSURG & 0.0205 (0.0209) & 0.0175 (0.0173) & 0.0051 (0.0146) & 0.0017 (0.0151) & 0.0883 (0.0929) & 0.0737 (0.0934)\\
First service: CSURG & -0.0086 (0.0051) & -0.0082 (0.0050) & -0.0066 (0.0044) & -0.0053 (0.0044) & -0.0542 (0.0245) & -0.0414 (0.0235)\\
First service: GU & -0.0087 (0.0082) & -0.0072 (0.0084) & -0.0065 (0.0084) & -0.0067 (0.0094) & -0.0278 (0.0791) & -0.0388 (0.1079)\\
\addlinespace
Intercept & -0.0085 (0.0080) & -0.0075 (0.0080) & -0.0085 (0.0063) & -0.0063 (0.0061) & -0.0392 (0.0340) & -0.0249 (0.0325)\\
Insurance: Medicaid & 0.0068 (0.0072) & 0.0049 (0.0069) & 0.0031 (0.0060) & 0.0026 (0.0056) & 0.0254 (0.0333) & 0.0235 (0.0316)\\
First service: other & 0.0005 (0.0092) & 0.0034 (0.0069) & 0.0067 (0.0066) & 0.0055 (0.0058) & 0.0342 (0.0927) & 0.0154 (0.0735)\\
First service: VSURG & -0.0060 (0.0066) & -0.0059 (0.0063) & -0.0031 (0.0050) & -0.0035 (0.0050) & -0.0235 (0.0276) & -0.0224 (0.0277)\\
pH (min, day 1) & -0.0021 (0.0244) & -0.0002 (0.0246) & 0.0091 (0.0188) & 0.0140 (0.0187) & -0.0247 (0.0922) & -0.0119 (0.0892)\\
\addlinespace
Smoking: & -0.0049 (0.0023) & -0.0042 (0.0020) & -0.0048 (0.0019) & -0.0038 (0.0015) & -0.0219 (0.0086) & -0.0182 (0.0074)\\
Insurance: Medicare & 0.0064 (0.0065) & 0.0048 (0.0066) & 0.0049 (0.0052) & 0.0044 (0.0050) & 0.0185 (0.0303) & 0.0173 (0.0290)\\
Insurance: Self Pay & -0.0032 (0.0115) & -0.0056 (0.0117) & -0.0061 (0.0092) & -0.0071 (0.0088) & -0.0149 (0.0474) & -0.0187 (0.0454)\\
Insurance: Private & 0.0058 (0.0063) & 0.0051 (0.0064) & 0.0037 (0.0051) & 0.0034 (0.0049) & 0.0184 (0.0301) & 0.0189 (0.0289)\\
First service: MED & -0.0040 (0.0049) & -0.0038 (0.0049) & -0.0044 (0.0039) & -0.0040 (0.0039) & -0.0207 (0.0173) & -0.0166 (0.0165)\\
\addlinespace
First service: NMED & 0.0059 (0.0112) & -0.0031 (0.0088) & -0.0022 (0.0095) & -0.0037 (0.0068) & -0.0181 (0.0373) & -0.0204 (0.0286)\\
First service: SURG & -0.0039 (0.0054) & -0.0028 (0.0055) & -0.0014 (0.0044) & -0.0009 (0.0044) & -0.0195 (0.0210) & -0.0152 (0.0205)\\
ARDS severity missing & 0.0020 (0.0088) & 0.0028 (0.0064) & -0.0009 (0.0082) & 0.0021 (0.0055) & -0.0195 (0.0363) & 0.0146 (0.0271)\\
Male sex & 0.0022 (0.0029) & 0.0033 (0.0027) & 0.0037 (0.0024) & 0.0027 (0.0021) & 0.0197 (0.0116) & 0.0093 (0.0104)\\
First service: ORTHO & 0.0012 (0.0081) & 0.0022 (0.0076) & 0.0014 (0.0068) & 0.0001 (0.0070) & 0.0198 (0.0579) & 0.0090 (0.0561)\\
\bottomrule
\end{tabular}
  }
\end{sidewaystable}

\begin{sidewaystable}[htbp]
  \centering
  \small
  \caption{A-learning CMTP coefficients, $\delta = 1$ J/min. Format: estimate (SE). SEs are plug-in cluster-sandwich asymptotic standard deviations from the estimating-equation implementation used for the coefficient plots.}
  \label{tab:al_cmtp_delta1}

  \resizebox{0.8\textheight}{!}{%

\begin{tabular}{lrrrrrr}
\toprule
\multicolumn{1}{c}{ } & \multicolumn{2}{c}{Squared Loss (no aug)} & \multicolumn{2}{c}{Squared Loss (with aug)} & \multicolumn{2}{c}{Logistic Loss} \\
\cmidrule(l{3pt}r{3pt}){2-3} \cmidrule(l{3pt}r{3pt}){4-5} \cmidrule(l{3pt}r{3pt}){6-7}
Covariate & Average & Direct & Average & Direct & Average & Direct\\
\midrule
pH (max, day 1) & 0.0346 (0.0292) & 0.0337 (0.0310) & 0.0413 (0.0228) & 0.0416 (0.0223) & 0.1605 (0.1080) & 0.1696 (0.1065)\\
First service: PSURG & -0.0207 (0.0208) & -0.0230 (0.0187) & -0.0070 (0.0143) & -0.0094 (0.0155) & -0.0917 (0.0935) & -0.0988 (0.1066)\\
First service: OMED & 0.0154 (0.0176) & 0.0150 (0.0172) & 0.0136 (0.0114) & 0.0122 (0.0130) & 0.0590 (0.0584) & 0.0518 (0.0595)\\
Intercept & 0.0091 (0.0081) & 0.0101 (0.0095) & 0.0085 (0.0064) & 0.0098 (0.0069) & 0.0405 (0.0334) & 0.0473 (0.0333)\\
First service: CSURG & 0.0088 (0.0051) & 0.0091 (0.0053) & 0.0053 (0.0043) & 0.0051 (0.0044) & 0.0424 (0.0233) & 0.0390 (0.0241)\\
\addlinespace
Insurance: Medicaid & -0.0076 (0.0073) & -0.0093 (0.0088) & -0.0051 (0.0062) & -0.0049 (0.0065) & -0.0322 (0.0329) & -0.0336 (0.0327)\\
Insurance: Medicare & -0.0066 (0.0067) & -0.0079 (0.0084) & -0.0053 (0.0054) & -0.0062 (0.0060) & -0.0209 (0.0298) & -0.0270 (0.0298)\\
First service: other & -0.0043 (0.0100) & -0.0045 (0.0077) & -0.0070 (0.0067) & -0.0083 (0.0062) & -0.0180 (0.0861) & -0.0292 (0.0777)\\
Insurance: Private & -0.0059 (0.0066) & -0.0069 (0.0082) & -0.0041 (0.0053) & -0.0049 (0.0059) & -0.0195 (0.0297) & -0.0246 (0.0296)\\
First service: ORTHO & -0.0022 (0.0083) & -0.0034 (0.0088) & -0.0032 (0.0067) & -0.0037 (0.0071) & -0.0235 (0.0567) & -0.0239 (0.0555)\\
\addlinespace
First service: GU & 0.0078 (0.0077) & 0.0072 (0.0075) & 0.0050 (0.0077) & 0.0048 (0.0093) & 0.0170 (0.0734) & 0.0176 (0.0908)\\
First service: VSURG & 0.0048 (0.0065) & 0.0052 (0.0065) & 0.0020 (0.0049) & 0.0020 (0.0050) & 0.0181 (0.0266) & 0.0214 (0.0275)\\
Smoking: & 0.0042 (0.0021) & 0.0040 (0.0022) & 0.0044 (0.0017) & 0.0037 (0.0017) & 0.0197 (0.0081) & 0.0167 (0.0079)\\
First service: MED & 0.0038 (0.0048) & 0.0042 (0.0050) & 0.0038 (0.0037) & 0.0042 (0.0037) & 0.0167 (0.0167) & 0.0180 (0.0165)\\
pH (min, day 1) & 0.0044 (0.0243) & 0.0049 (0.0261) & -0.0080 (0.0180) & -0.0075 (0.0179) & 0.0129 (0.0888) & 0.0110 (0.0883)\\
\addlinespace
Male sex & -0.0019 (0.0028) & -0.0012 (0.0029) & -0.0030 (0.0023) & -0.0029 (0.0023) & -0.0164 (0.0110) & -0.0161 (0.0110)\\
First service: SURG & 0.0034 (0.0054) & 0.0042 (0.0055) & 0.0005 (0.0043) & 0.0009 (0.0043) & 0.0111 (0.0203) & 0.0137 (0.0205)\\
ARDS severity: & 0.0027 (0.0016) & 0.0028 (0.0016) & 0.0014 (0.0013) & 0.0013 (0.0013) & 0.0127 (0.0057) & 0.0120 (0.0059)\\
First service: TRAUM & -0.0019 (0.0058) & -0.0022 (0.0059) & -0.0020 (0.0046) & -0.0019 (0.0046) & -0.0126 (0.0217) & -0.0115 (0.0214)\\
Insurance: Self Pay & 0.0035 (0.0118) & -0.0000 (0.0139) & 0.0056 (0.0092) & 0.0035 (0.0100) & 0.0135 (0.0463) & 0.0034 (0.0473)\\
\bottomrule
\end{tabular}
  }
\end{sidewaystable}

\section{Proofs}\label{sec:supp_proofs}

\subsection{Proofs of Propositions}

\begin{proof}[Proof of Proposition~1]
{Fix \(\bx\in\mathcal X\) and \(a\in\mathcal A\) with \(f_{A\mid\bX}(a\mid\bx)>0\); Assumption~2 then guarantees \(f_{A\mid\bX}(a+\delta\mid\bx)>0\), so both conditional means below are well defined. By consistency (Assumption~1),
\begin{align*}
    &\mathbb E[Y\mid \bX=\bx,A=a+\delta]
=
\mathbb E\{Y(a+\delta)\mid \bX=\bx,A=a+\delta\},\\
&\mathbb E[Y\mid \bX=\bx,A=a]
=
\mathbb E\{Y(a)\mid \bX=\bx,A=a\}.
\end{align*}
By Assumption~3(A$'$) with \(q(a,\bx)=a+\delta\),
\[
\mathbb E\{Y(a+\delta)\mid \bX=\bx,A=a+\delta\}
=
\mathbb E\{Y(a+\delta)\mid \bX=\bx,A=a\}.
\]
Subtracting the two consistency identities and applying the display above,
\begin{align*}
    \mathbb E[Y\mid \bX=\bx,A=a+\delta]-\mathbb E[Y\mid \bX=\bx,A=a]
&=\mathbb E\{Y(a+\delta)-Y(a)\mid \bX=\bx,A=a\}\\
&=\mathbb E\{Y(a+\delta)-Y(a)\mid \bX=\bx\}\\
&=\tau_\delta(a,\bx),
\end{align*}
where the second equality is Assumption~3(A). Note that the proof makes it clear where each condition is used: (A$'$) moves the shifted-arm potential outcome across conditioning arms, while (A) removes the dose from the conditioning set in the contrast.}
\end{proof}

\begin{proof}[Proof of Proposition~2]
{Fix \(\bx\) with \(f_{\bX}(\bx)>0\) and write \(m_0(a,\bx)=\mathbb E[Y\mid\bX=\bx,A=a]\) for the observed outcome regression. By consistency (Assumption~1), \(Y=Y(A)\) almost surely, so
\[
\mathbb E\{Y(A)\mid\bX=\bx\}=\mathbb E[Y\mid\bX=\bx].
\]
Next, by the tower property,
\[
\mathbb E\{Y(A+\delta)\mid \bX=\bx\}
=\mathbb E\big[\,\mathbb E\{Y(A+\delta)\mid \bX=\bx,A\}\,\big|\,\bX=\bx\big].
\]
For \(a\) in the conditional support of \(A\mid\bX=\bx\), Assumption~2 gives \(f_{A\mid\bX}(a+\delta\mid\bx)>0\), so Assumption~3(A$'$) (with \(q(a,\bx)=a+\delta\)) followed by consistency yields
\begin{align*}
    \mathbb E\{Y(a+\delta)\mid \bX=\bx,A=a\}
&=\mathbb E\{Y(a+\delta)\mid \bX=\bx,A=a+\delta\}\\
&=\mathbb E[Y\mid \bX=\bx,A=a+\delta]
=m_0(a+\delta,\bx).
\end{align*}
Hence \(\mathbb E\{Y(A+\delta)\mid\bX=\bx\}=\mathbb E[m_0(A+\delta,\bx)\mid\bX=\bx]\), and subtracting the two identities gives
\[
\tau_\delta(\bx)
=\mathbb E\{Y(A+\delta)-Y(A)\mid\bX=\bx\}
=\mathbb E[m_0(A+\delta,\bx)\mid\bX=\bx]-\mathbb E[Y\mid\bX=\bx].
\]
Note that only Assumption~3(A$'$) is used: the CMTP is identified without the mean nudge exchangeability condition (A), and the identification formula involves only the observed outcome regression \(m_0\), never the counterfactual CADRF.}
\end{proof}

\begin{proof}[Proof of Proposition~3]
{Fix \((a,\bx)\) with \(f_{A\mid\bX}(a\mid\bx)>0\), so that \(f_{A\mid\bX}(a+\delta\mid\bx)>0\) by Assumption~2 and hence \(\pi_\Lambda(\bx,a)\in(0,1)\). Throughout this proof we work with the recoded arm indicator \(\widetilde\Lambda=\Lambda+\tfrac12\in\{0,1\}\) of the main text, and for notational simplicity write}
\[
\pi(a,\bx)
=
\Pr({\widetilde\Lambda}=1\mid A_\Lambda=a,\bX=\bx){=\pi_\Lambda(\bx,a)},
\]
and let \(w_+(a,\bx)\) and \(w_-(a,\bx)\) denote the weights assigned to the
\({\widetilde\Lambda}=1\) and \({\widetilde\Lambda}=0\) duplicated arms, respectively. Also define
\[
\mu_+(a,\bx)
=
\mathbb E[Y\mid A_\Lambda=a,\bX=\bx,{\widetilde\Lambda}=1],
\qquad
\mu_-(a,\bx)
=
\mathbb E[Y\mid A_\Lambda=a,\bX=\bx,{\widetilde\Lambda}=0].
\]
By construction of the duplicated data distribution, any row with \({\widetilde\Lambda}=1\) with
\(A_\Lambda=a\) corresponds to the observed treatment value \(A=a+\delta\),
whereas any row with \({\widetilde\Lambda}=0\)  corresponds to the observed treatment value
\(A=a\). Because the data-duplication label is assigned independently of \(Y\) given \((A,\bX)\), conditioning on \((\bX=\bx, A_\Lambda=a, \widetilde\Lambda=1)\) is equivalent to conditioning on \((\bX=\bx, A=a+\delta)\). Hence
\[
\mu_+(a,\bx)=\mathbb E[Y\mid \bX=\bx,A=a+\delta],
\qquad
\mu_-(a,\bx)=\mathbb E[Y\mid \bX=\bx,A=a].
\]
Now condition on \(A_\Lambda=a\) and \(\bX=\bx\). Writing
\(\Gamma={\widetilde\Lambda}-1/2\) {(which coincides with the main text's \(\Lambda\))}, the population weighted squared loss is
\[
\ell_W(f\mid a,\bx)
=
\mathbb E\left[
w(A_\Lambda,\bX,\Lambda)
\{Y-\Gamma f(a,\bx)\}^2
\mid A_\Lambda=a,\bX=\bx
\right].
\]
Since \(f(a,\bx)\) is scalar at the fixed point \((a,\bx)\), the first-order
condition is
\[
0
=
\frac{\partial}{\partial f}
\ell_W(f\mid a,\bx)
=
-2\mathbb E\left[
w\Gamma\{Y-\Gamma f(a,\bx)\}
\mid A_\Lambda=a,\bX=\bx
\right].
\]
{Since \(w_+,w_->0\) and \(0<\pi(a,\bx)<1\), we have \(\mathbb E[w\Gamma^2\mid A_\Lambda=a,\bX=\bx]>0\), so \(f\mapsto\ell_W(f\mid a,\bx)\) is a strictly convex quadratic and the first-order condition characterizes its unique minimizer.} Therefore {the} minimizer satisfies
\[
f_W^*(a,\bx)
=
\frac{
\mathbb E\left[
w\Gamma Y\mid A_\Lambda=a,\bX=\bx
\right]
}{
\mathbb E\left[
w\Gamma^2\mid A_\Lambda=a,\bX=\bx
\right]
}.
\]
Because \(\Gamma=1/2\) when \(\Lambda=1\) and \(\Gamma=-1/2\) when
\(\Lambda=0\), the numerator is
\[
\mathbb E\left[
w\Gamma Y\mid A_\Lambda=a,\bX=\bx
\right]
=
\frac12 w_+(a,\bx)\pi(a,\bx)\mu_+(a,\bx)
-
\frac12 w_-(a,\bx)\{1-\pi(a,\bx)\}\mu_-(a,\bx),
\]
and the denominator is
\[
\mathbb E\left[
w\Gamma^2\mid A_\Lambda=a,\bX=\bx
\right]
=
\frac14 w_+(a,\bx)\pi(a,\bx)
+
\frac14 w_-(a,\bx)\{1-\pi(a,\bx)\}.
\]
By the balancing condition (3),
\[
w_+(a,\bx)\pi(a,\bx)
=
w_-(a,\bx)\{1-\pi(a,\bx)\}.
\]
Let the common value be \(c(a,\bx)\). Then
\[
\mathbb E\left[
w\Gamma Y\mid A_\Lambda=a,\bX=\bx
\right]
=
\frac12 c(a,\bx)\{\mu_+(a,\bx)-\mu_-(a,\bx)\},
\]
whereas
\[
\mathbb E\left[
w\Gamma^2\mid A_\Lambda=a,\bX=\bx
\right]
=
\frac12 c(a,\bx).
\]
Thus
\[
f_W^*(a,\bx)
=
\mu_+(a,\bx)-\mu_-(a,\bx).
\]
Using the duplicated-data identities above,
\[
f_W^*(a,\bx)
=
\mathbb E[Y\mid \bX=\bx,A=a+\delta]
-
\mathbb E[Y\mid \bX=\bx,A=a].
\]
By Proposition~1, this observed data
contrast is equal to \(\tau_\delta(a,\bx)\). Therefore
\[
f_W^*(a,\bx)=\tau_\delta(a,\bx).
\]
It remains to show that augmentation does not change the population minimizer.
For the augmented loss, again conditioning on \(A_\Lambda=a\) and
\(\bX=\bx\), the first-order condition gives
\[
f_{W,\mathrm{aug}}^*(a,\bx)
=
\frac{
\mathbb E\left[
w\Gamma\{Y-m(A_\Lambda,\bX)\}
\mid A_\Lambda=a,\bX=\bx
\right]
}{
\mathbb E\left[
w\Gamma^2
\mid A_\Lambda=a,\bX=\bx
\right]
}.
\]
At the fixed point \((A_\Lambda,\bX)=(a,\bx)\), the augmentation term
\(m(A_\Lambda,\bX)=m(a,\bx)\) is common to both duplicated arms. Hence, using
the same balancing condition as above,
\begin{align*}
&\mathbb E\left[
w\Gamma\{Y-m(A_\Lambda,\bX)\}
\mid A_\Lambda=a,\bX=\bx
\right]\\
&=
\frac12 c(a,\bx)\{\mu_+(a,\bx)-m(a,\bx)\}  -
\frac12 c(a,\bx)\{\mu_-(a,\bx)-m(a,\bx)\} \\
&=
\frac12 c(a,\bx)\{\mu_+(a,\bx)-\mu_-(a,\bx)\}.
\end{align*}
The denominator remains \(\frac12 c(a,\bx)\). Therefore
\[
f_{W,\mathrm{aug}}^*(a,\bx)
=
\mu_+(a,\bx)-\mu_-(a,\bx)
=
\mathbb E[Y\mid \bX=\bx,A=a+\delta]
-
\mathbb E[Y\mid \bX=\bx,A=a].
\]
Applying Proposition~1 once more gives
\[
f_{W,\mathrm{aug}}^*(a,\bx)=\tau_\delta(a,\bx).
\]
Combining the two parts,
\[
f_W^*(a,\bx)
=
f_{W,\mathrm{aug}}^*(a,\bx)
=
\tau_\delta(a,\bx),
\]
as claimed. {Finally, since the marginal risk satisfies \(\ell_W(f)=\mathbb E\{\ell_W(f\mid A_\Lambda,\bX)\}\) and the pointwise minimizer \((a,\bx)\mapsto\mu_+(a,\bx)-\mu_-(a,\bx)\) is measurable, minimizing the marginal risk over all measurable \(f\) is achieved by minimizing the conditional risk at each \((a,\bx)\); hence \(f_W^*\) and \(f_{W,\mathrm{aug}}^*\) are the \(P_{(A_\Lambda,\bX)}\)-almost-everywhere unique minimizers of the marginal losses.}
\end{proof}

\begin{proof}[Proof of Proposition~4]
{Fix \((a,\bx)\) with \(f_{A\mid\bX}(a\mid\bx)>0\), so that \(\pi_\Lambda(\bx,a)\in(0,1)\) by Assumption~2. As in the previous proof we work with the recoded arm indicator \(\widetilde\Lambda=\Lambda+\tfrac12\in\{0,1\}\), and write}
\[
\pi(a,\bx)
=
\Pr({\widetilde\Lambda}=1\mid A_\Lambda=a,\bX=\bx){=\pi_\Lambda(\bx,a)},
\qquad
g(a,\bx,{\widetilde\Lambda})
=
{\widetilde\Lambda}-\pi(a,\bx).
\]
Also define the two conditional outcome means in the duplicated sample by
\[
\mu_+(a,\bx)
=
\mathbb E[Y\mid A_\Lambda=a,\bX=\bx,\widetilde\Lambda=1],
\qquad
\mu_-(a,\bx)
=
\mathbb E[Y\mid A_\Lambda=a,\bX=\bx,\widetilde\Lambda=0].
\]
By construction of the duplicated data, the \(\widetilde\Lambda=1\) row at
\(A_\Lambda=a\) corresponds to the observed treatment value \(A=a+\delta\),
whereas the \(\widetilde\Lambda=0\) row corresponds to the observed treatment value
\(A=a\). Therefore,
\[
\mu_+(a,\bx)=\mathbb E[Y\mid \bX=\bx,A=a+\delta],
\qquad
\mu_-(a,\bx)=\mathbb E[Y\mid \bX=\bx,A=a].
\]

Conditioning on \(A_\Lambda=a\) and \(\bX=\bx\), the population A-learning
loss is
\[
\ell_A(f\mid a,\bx)
=
\mathbb E\left[
\{Y-g(a,\bx,\widetilde\Lambda)f(a,\bx)\}^2
\mid A_\Lambda=a,\bX=\bx
\right].
\]
Since \(f(a,\bx)\) is scalar at the fixed point \((a,\bx)\), the first-order
condition is
\[
0
=
\frac{\partial}{\partial f}\ell_A(f\mid a,\bx)
=
-2\mathbb E\left[
g(a,\bx,\widetilde\Lambda)
\{Y-g(a,\bx,\widetilde\Lambda)f(a,\bx)\}
\mid A_\Lambda=a,\bX=\bx
\right].
\]
{Since \(0<\pi(a,\bx)<1\), we have \(\mathbb E[g^2\mid A_\Lambda=a,\bX=\bx]=\pi(a,\bx)\{1-\pi(a,\bx)\}>0\), so the conditional loss is strictly convex in \(f(a,\bx)\) and the first-order condition characterizes its unique minimizer.} Thus {the} minimizer satisfies
\[
f_A^*(a,\bx)
=
\frac{
\mathbb E\left[
g(a,\bx,\widetilde\Lambda)Y
\mid A_\Lambda=a,\bX=\bx
\right]
}{
\mathbb E\left[
g(a,\bx,\widetilde\Lambda)^2
\mid A_\Lambda=a,\bX=\bx
\right]
}.
\]
Now, conditional on \((A_\Lambda,\bX)=(a,\bx)\), we have
\(g=1-\pi(a,\bx)\) when \(\widetilde\Lambda=1\), and
\(g=-\pi(a,\bx)\) when \(\widetilde\Lambda=0\). Hence
\begin{align*}
&\mathbb E\left[
g(a,\bx,\widetilde\Lambda)Y
\mid A_\Lambda=a,\bX=\bx
\right]\\
&=
\pi(a,\bx)\{1-\pi(a,\bx)\}\mu_+(a,\bx) 
-\{1-\pi(a,\bx)\}\pi(a,\bx)\mu_-(a,\bx) \\
&=
\pi(a,\bx)\{1-\pi(a,\bx)\}
\{\mu_+(a,\bx)-\mu_-(a,\bx)\}.
\end{align*}
Similarly,
\begin{align*}
\mathbb E\left[
g(a,\bx,\widetilde\Lambda)^2
\mid A_\Lambda=a,\bX=\bx
\right]
&=
\pi(a,\bx)\{1-\pi(a,\bx)\}^2
+
\{1-\pi(a,\bx)\}\pi(a,\bx)^2 \\
&=
\pi(a,\bx)\{1-\pi(a,\bx)\}.
\end{align*}
Therefore,
\[
f_A^*(a,\bx)
=
\mu_+(a,\bx)-\mu_-(a,\bx).
\]
Using the duplicated-data identities above,
\[
f_A^*(a,\bx)
=
\mathbb E[Y\mid \bX=\bx,A=a+\delta]
-
\mathbb E[Y\mid \bX=\bx,A=a].
\]
By Proposition~1, this contrast equals
\(\tau_\delta(a,\bx)\). Hence
\[
f_A^*(a,\bx)=\tau_\delta(a,\bx).
\]

It remains to verify that augmentation does not alter the target. Conditioning
again on \(A_\Lambda=a\) and \(\bX=\bx\), the augmented A-learning loss gives
the first-order condition
\[
f_{A,\mathrm{aug}}^*(a,\bx)
=
\frac{
\mathbb E\left[
g(a,\bx,\widetilde\Lambda)\{Y-m(A_\Lambda,\bX)\}
\mid A_\Lambda=a,\bX=\bx
\right]
}{
\mathbb E\left[
g(a,\bx,\widetilde\Lambda)^2
\mid A_\Lambda=a,\bX=\bx
\right]
}.
\]
At the fixed point \((A_\Lambda,\bX)=(a,\bx)\), the augmentation term
\(m(A_\Lambda,\bX)=m(a,\bx)\) is common to both duplicated arms. Therefore,
\begin{align*}
&\mathbb E\left[
g(a,\bx,\widetilde\Lambda)\{Y-m(A_\Lambda,\bX)\}
\mid A_\Lambda=a,\bX=\bx
\right] \\
&\quad =
\pi(a,\bx)\{1-\pi(a,\bx)\}\{\mu_+(a,\bx)-m(a,\bx)\}
-
\{1-\pi(a,\bx)\}\pi(a,\bx)\{\mu_-(a,\bx)-m(a,\bx)\} \\
&\quad =
\pi(a,\bx)\{1-\pi(a,\bx)\}
\{\mu_+(a,\bx)-\mu_-(a,\bx)\}.
\end{align*}
The denominator is again
\[
\mathbb E\left[
g(a,\bx,\widetilde\Lambda)^2
\mid A_\Lambda=a,\bX=\bx
\right]
=
\pi(a,\bx)\{1-\pi(a,\bx)\}.
\]
Thus
\[
f_{A,\mathrm{aug}}^*(a,\bx)
=
\mu_+(a,\bx)-\mu_-(a,\bx)
=
\mathbb E[Y\mid \bX=\bx,A=a+\delta]
-
\mathbb E[Y\mid \bX=\bx,A=a].
\]
Applying Proposition~1 once more gives
\[
f_{A,\mathrm{aug}}^*(a,\bx)=\tau_\delta(a,\bx).
\]
Combining the two parts,
\[
f_A^*(a,\bx)
=
f_{A,\mathrm{aug}}^*(a,\bx)
=
\tau_\delta(a,\bx),
\]
as claimed. {As in the proof of Proposition~3, minimizing the marginal risks \(\ell_A(f)\) and \(\ell_{A,\mathrm{aug}}(f)\) over measurable \(f\) is equivalent to this pointwise minimization for \(P_{(A_\Lambda,\bX)}\)-almost every \((a,\bx)\).}
\end{proof}

\begin{proof}[Proof of Proposition~5]
Fix \(\bx\) with \(f_{\bX}(\bx)>0\). As in the preceding proofs, write \(\widetilde\Lambda=\Lambda+\tfrac12\in\{0,1\}\),
\begin{align*}
    &\pi(a,\bx)
=
\Pr(\widetilde\Lambda=1\mid A_\Lambda=a,\bX=\bx)=\pi_\Lambda(\bx,a),\\
&g(a,\bx,\widetilde\Lambda)
=
\widetilde\Lambda-\pi(a,\bx),
\\
&h(a,\bx)=\pi(a,\bx)\{1-\pi(a,\bx)\},
\end{align*}
and let \(m_0(a,\bx)=\mathbb E[Y\mid \bX=\bx, A=a]\) denote the observed outcome regression. Define the arm-specific means \(\mu_+(a,\bx)=\mathbb E[Y\mid A_\Lambda=a,\bX=\bx,\widetilde\Lambda=1]\) and \(\mu_-(a,\bx)=\mathbb E[Y\mid A_\Lambda=a,\bX=\bx,\widetilde\Lambda=0]\); by construction of the duplicated data, \(\mu_+(a,\bx)=m_0(a+\delta,\bx)\) and \(\mu_-(a,\bx)=m_0(a,\bx)\).

\emph{Step 1 (evaluation of the unstabilized loss; proof of Lemma~1 in the main paper).}
Conditioning on \((A_\Lambda,\bX)=(a,\bx)\), the computations in the proof of Proposition~3 give
\begin{align*}
    &\mathbb E[g(A_\Lambda,\bx,\widetilde\Lambda)\,Y\mid A_\Lambda=a,\bX=\bx]
=
h(a,\bx)\{\mu_+(a,\bx)-\mu_-(a,\bx)\},
\\
&\mathbb E[g(A_\Lambda,\bx,\widetilde\Lambda)^2\mid A_\Lambda=a,\bX=\bx]
=
h(a,\bx).
\end{align*}
Averaging over \(A_\Lambda\mid\bX=\bx\) yields
\(N(\bx)=\mathbb E[h(A_\Lambda,\bx)\{m_0(A_\Lambda+\delta,\bx)-m_0(A_\Lambda,\bx)\}\mid\bX=\bx]\) and
\(D(\bx)=\mathbb E[h(A_\Lambda,\bx)\mid\bX=\bx]\), which is the claim of Lemma~1.

\emph{Step 2 (the \(\pi_\Lambda\)-stabilized loss).}
The stabilized loss (7), conditional on \(\bX=\bx\), is
\[
\ell_A^{\mathrm{CMTP}}(f\mid \bx)
=
\mathbb E\left[
\frac{\{Y-g(A_\Lambda,\bx,\widetilde\Lambda)f(\bx)\}^2}
{\pi(A_\Lambda,\bx)}
\Bigm| \bX=\bx
\right].
\]
Since \(f(\bx)\) is scalar after conditioning on \(\bX=\bx\), the first-order condition gives
\[
f_A^*(\bx)
=
\frac{
\mathbb E\left[
\pi(A_\Lambda,\bx)^{-1}g(A_\Lambda,\bx,\widetilde\Lambda)\,Y
\mid \bX=\bx
\right]
}{
\mathbb E\left[
\pi(A_\Lambda,\bx)^{-1}g(A_\Lambda,\bx,\widetilde\Lambda)^2
\mid \bX=\bx
\right]
},
\]
and by Step 1, after division by \(\pi(a,\bx)\),
\[
f_A^*(\bx)
=
\frac{
\mathbb E\big[\{1-\pi(A_\Lambda,\bx)\}\{\mu_+(A_\Lambda,\bx)-\mu_-(A_\Lambda,\bx)\}\mid \bX=\bx\big]
}{
\mathbb E\big[1-\pi(A_\Lambda,\bx)\mid \bX=\bx\big]
}.
\]
The uniqueness of the pointwise minimizer follows from the fact that the denominator is strictly positive, as we show next.

\emph{Step 3 (removing the mixture tilt).}
The conditional density of \(A_\Lambda\) given \(\bX=\bx\) can be expressed as the mixture
\(f_{A_\Lambda\mid\bX}(a\mid\bx)=\tfrac12\{g_{A\mid\bX}(a\mid\bx)+g_{A\mid\bX}(a+\delta\mid\bx)\}\). The definition of \(\pi(a,\bx)\) yields 
\[
\{1-\pi(a,\bx)\}\,f_{A_\Lambda\mid\bX}(a\mid\bx)
=
\tfrac12\, g_{A\mid\bX}(a\mid\bx).
\]
Consequently, the denominator is
\(\int \{1-\pi(a,\bx)\}f_{A_\Lambda\mid\bX}(a\mid\bx)\,\mathrm da
=\tfrac12\int g_{A\mid\bX}(a\mid\bx)\,\mathrm da=\tfrac12\),
and the numerator equals
\begin{align*}
    &\int \{1-\pi(a,\bx)\}\{m_0(a+\delta,\bx)-m_0(a,\bx)\}f_{A_\Lambda\mid\bX}(a\mid\bx)\,\mathrm da\\
&=
\tfrac12\int \{m_0(a+\delta,\bx)-m_0(a,\bx)\}\,g_{A\mid\bX}(a\mid\bx)\,\mathrm da.
\end{align*}
Therefore
\[
f_A^*(\bx)
=
\mathbb E\big[m_0(A+\delta,\bx)-m_0(A,\bx)\mid \bX=\bx\big],
\]
where the expectation is now taken over the \textit{natural} conditional distribution of \(A\) given \(\bX=\bx\). 
Since by the tower property, \(\mathbb E[m_0(A,\bx)\mid\bX=\bx]=\mathbb E[Y\mid\bX=\bx]\), Proposition~2 results in \(f_A^*(\bx)=\tau_\delta(\bx)\) under Assumptions~1, 2, and 3(A$'$).

\emph{Step 4 (augmentation invariance).}
For the augmented stabilized loss (8), the first-order condition involves \(Y-m(A_\Lambda,\bX)\) instead of \(Y\) in the numerator. Then by conditioning on \((A_\Lambda,\bX)=(a,\bx)\), the augmentation term \(m(a,\bx)\) can be seen as the same in each of the duplicated arms, so
\[
\mathbb E\big[g(A_\Lambda,\bx,\widetilde\Lambda)\{Y-m(A_\Lambda,\bX)\}\mid A_\Lambda=a,\bX=\bx\big]
=
h(a,\bx)\{\mu_+(a,\bx)-\mu_-(a,\bx)\},
\]
as in Step 1; the denominator is unchanged. Thus \(f_{A,\mathrm{aug}}^*(\bx)=f_A^*(\bx)=\tau_\delta(\bx)\) for any measurable \(m\).

Finally, since the marginal risks can be constructed by averaging the conditional risks over the distribution of \(\bX\), the marginal minimizers thus agree with the pointwise minimizers, which completes the proof.
\end{proof}

\begin{proof}[Proof of Proposition~6]
Fix \((a,\bx)\) with \(\pi:=\pi_\Lambda(\bx,a)\in(0,1)\), and let
\(\mu_+(a,\bx)\) and \(\mu_-(a,\bx)\) denote the two arm-specific means, as in the proof of Proposition~3. Conditional on \((A_\Lambda,\bX)=(a,\bx)\),
\(\widetilde\Lambda-\pi\) equals \(1-\pi\) with probability \(\pi\) and \(-\pi\) with
probability \(1-\pi\), so
\[
\ell_A^{\mathrm{NLL}}(f;a,\bx)
=\pi\bigl\{-\mu_+(a,\bx)(1-\pi)f+b\bigl((1-\pi)f\bigr)\bigr\}
+(1-\pi)\bigl\{\mu_-(a,\bx)\pi f+b(-\pi f)\bigr\}.
\]
Differentiating in the scalar \(f\) and using \(b'=\mu\),
\[
\frac{\partial}{\partial f}\ell_A^{\mathrm{NLL}}(f;a,\bx)
=\pi(1-\pi)\Bigl[\mu\bigl((1-\pi)f\bigr)-\mu(-\pi f)
-\{\mu_+(a,\bx)-\mu_-(a,\bx)\}\Bigr],
\]
and
\(\frac{\partial^2}{\partial f^2}\ell_A^{\mathrm{NLL}}(f;a,\bx)
=\pi(1-\pi)^2 b''\bigl((1-\pi)f\bigr)+(1-\pi)\pi^2 b''(-\pi f)>0\),
since \(b\) is strictly convex. Hence the conditional loss is strictly convex in
\(f\), and any minimizer satisfies
\(\mu\bigl((1-\pi)f\bigr)-\mu(-\pi f)=\mu_+(a,\bx)-\mu_-(a,\bx)\);
because the left-hand side is strictly increasing in \(f\), this solution is
unique when it exists. By Proposition~1,
\(\mu_+(a,\bx)-\mu_-(a,\bx)=\tau_\delta(a,\bx)\). Finally, since
\(\ell_A^{\mathrm{NLL}}(f)=\mathbb E\{\ell_A^{\mathrm{NLL}}(f;A_\Lambda,\bX)\}\),
the assumed unique marginal minimizer \(f_A^*\) coincides with the pointwise minimizer and therefore satisfies the
displayed identity for almost every \((a,\bx)\).
\end{proof}

\begin{proof}[Proof of Lemma~2]
Write \(g_i^\lambda=\widetilde\Lambda_i^\lambda-\pi_\Lambda(\bZ_i^\lambda)\) for the centered arm score with the true propensity, and \(\widetilde\bb_i^\lambda=g_i^\lambda\bb_i^\lambda\). By the law of large numbers on the duplicated sample and uniform consistency of \(\widehat\pi_\Lambda\),
\[
\bbP_n^\Lambda\,\widetilde\bb_i^\lambda\widetilde\bb_i^{\lambda\top}
\;\overset{p}{\to}\;
\mathbb E^\Lambda\big[(\widetilde\Lambda-\pi_\Lambda)^2\,\bb\bb^\top\big]
=
\mathbb E^\Lambda\big[h_A(A_\Lambda,\bX)\,\bb(A_\Lambda,\bX)\bb(A_\Lambda,\bX)^\top\big],
\]
using \(\mathbb E[(\widetilde\Lambda-\pi_\Lambda)^2\mid A_\Lambda,\bX]=h_A(A_\Lambda,\bX)\). Similarly, since
$$\mathbb E[(\widetilde\Lambda-\pi_\Lambda)Y\mid A_\Lambda,\bX]
=h_A(A_\Lambda,\bX)\{\mu_+(A_\Lambda,\bX)-\mu_-(A_\Lambda,\bX)\}
=h_A(A_\Lambda,\bX)\,\tau_\delta(A_\Lambda,\bX)$$
under the conditions of Proposition~1,
\[
\bbP_n^\Lambda\,\widetilde\bb_i^\lambda Y_i
\;\overset{p}{\to}\;
\mathbb E^\Lambda\big[h_A(A_\Lambda,\bX)\,\bb(A_\Lambda,\bX)\,\tau_\delta(A_\Lambda,\bX)\big].
\]
By nonsingularity of the limiting Gram matrix and the continuous mapping theorem,
\(\widehat\btheta_A\overset{p}{\to}\btheta_A^*\) with
\[
\btheta_A^*
=
\mathbb E^\Lambda\big[h_A\,\bb\bb^\top\big]^{-1}
\mathbb E^\Lambda\big[h_A\,\bb\,\tau_\delta\big],
\]
which is exactly the normal-equation characterization of
\(\argmin_{\btheta}\mathbb E^\Lambda[h_A(A_\Lambda,\bX)\{\tau_\delta(A_\Lambda,\bX)-\btheta^\top\bb(A_\Lambda,\bX)\}^2]\).
\end{proof}

\subsection{Proofs of Theorems}

Recall the notations $Z_i^\lambda := (A_i-\lambda\delta,\bX_i), 
    b_i^\lambda := b(Z_i^\lambda) \text{ and }
    Y_i^\lambda := Y_i \text{ for }
    \lambda\in\{0,1\}.$
Thus \(Z_i^0=(A_i,\bX_i)\) denotes the observed arm and
\(Z_i^1=(A_i-\delta,\bX_i)\) denotes the shifted arm. The duplicated empirical average
\[
    \bbP_n^\Lambda f_i^\lambda
    :=
    \frac1n\sum_{i=1}^n\sum_{\lambda=0}^1 f_i^\lambda .
\]
For A-learning, define the transformed basis $\widetilde b_i^\lambda
    :=
    \left\{
        \widetilde\Lambda_i^\lambda
        -
        \widehat\pi_\Lambda(Z_i^\lambda)
    \right\}
    b_i^\lambda,$
where \(\widetilde\Lambda_i^\lambda\) denotes the treatment-arm coding used in the
duplicated-arm construction and $\widehat\pi_\Lambda(Z_i^\lambda) = \widehat{\pi}(A_{i}, \bX_{i}, \Lambda_i)$. Then the A-learning empirical loss is
$$\widehat\theta_A=
    \argmin_\theta
    \bbP_n^\Lambda
    \left\{
        Y_i-\theta^\top \widetilde b_i^\lambda
    \right\}^2 = \left\{
        \bbP_n^\Lambda
        \widetilde b_i^\lambda\widetilde b_i^{\lambda\top}
    \right\}^{-1}
    \left\{
        \bbP_n^\Lambda
        \widetilde b_i^\lambda Y_i
    \right\}.$$
We now prove the theorems.

\begin{proof}[Proof of Theorem~1]
\noindent\textbf{Part (i): asymptotic normality of $\widehat{\btheta}_A$.}
We first note that $\widehat{\btheta}_A = \widehat{\bD}^{-1}_n\widehat{\bd}_n$, where
	\begin{align*}
		\frac{1}{2}\widehat{\bd}_n = {} & \frac{1}{2}\bbP_n \left[ \left\{(1-\widehat{\pi}_\Lambda(A-\delta, \bX)){\bb}(A-\delta, \bX) - \widehat{\pi}_\Lambda(A, \bX){\bb}(A, \bX)\right\}Y \right] \\
		\text{and} & \\
		\frac{1}{2}\widehat{\bD}_n = {} & \frac{1}{2}\bbP_n \left[ (1-\widehat{\pi}_\Lambda(A-\delta, \bX))^2{\bb}(A-\delta, \bX){\bb}(A-\delta, \bX)\trans + \widehat{\pi}_\Lambda(A, \bX)^2{\bb}(A, \bX){\bb}(A, \bX)\trans \right],
	\end{align*}
with the $(1-\pi)$-factor attached to the shifted row $(A-\delta,\bX)$ (arm indicator $\widetilde\Lambda=1$) and the $\pi$-factor to the observed row $(A,\bX)$, consistent with the construction in Section~3 of the main paper. Similarly, $\btheta^*_A = \bD^{-1}\bd$ by Lemma~2, where
	\begin{align*}
	\frac{1}{2}\bd = {} & \bbE^\Lambda\left[(\widetilde{\Lambda}- {\pi}_\Lambda(A_\Lambda, \bX)){\bb}(A_\Lambda, \bX) Y\right] \\
	= {} & \frac{1}{2}\bbE_{\bX, A, Y}\left[\left((1-\pi_\Lambda(A-\delta, \bX)){\bb}(A-\delta, \bX) - \pi_\Lambda(A, \bX){\bb}(A, \bX)\right)Y\right] \\
	\equiv {} & \frac{1}{2}P\left[\left((1-\pi_\Lambda(A-\delta, \bX)){\bb}(A-\delta, \bX) - \pi_\Lambda(A, \bX){\bb}(A, \bX)\right)Y\right] \\
	\text{and} & \\
	\frac{1}{2}\bD = {} & \bbE^\Lambda\left[\{\widetilde{\Lambda}- {\pi}_\Lambda(A_\Lambda, \bX)\}^2{\bb}(A_\Lambda, \bX){\bb}(A_\Lambda, \bX)\trans\right] \\
	= {} &  \frac{1}{2}\bbE_{\bX, A}\left[(1-\pi_\Lambda(A-\delta, \bX))^2{\bb}(A-\delta, \bX){\bb}(A-\delta, \bX)\trans + \pi_\Lambda(A, \bX)^2{\bb}(A, \bX){\bb}(A, \bX)\trans\right] \\
	\equiv {} &  \frac{1}{2}P\left[(1-\pi_\Lambda(A-\delta, \bX))^2{\bb}(A-\delta, \bX){\bb}(A-\delta, \bX)\trans + \pi_\Lambda(A, \bX)^2{\bb}(A, \bX){\bb}(A, \bX)\trans\right].
	\end{align*}
We further define the intermediate quantities
\begin{align*}
	\frac{1}{2}\widehat{\bd}
	\equiv {} & \frac{1}{2}P\left[\left\{(1-\widehat{\pi}_\Lambda(A-\delta, \bX)){\bb}(A-\delta, \bX) - \widehat{\pi}_\Lambda(A, \bX){\bb}(A, \bX)\right\}Y\right], \\
	\frac{1}{2}\widehat{\bD}
	\equiv {} &  \frac{1}{2}P\left[(1-\widehat{\pi}_\Lambda(A-\delta, \bX))^2{\bb}(A-\delta, \bX){\bb}(A-\delta, \bX)\trans + \widehat{\pi}_\Lambda(A, \bX)^2{\bb}(A, \bX){\bb}(A, \bX)\trans\right], \\
	\frac{1}{2}{\bd}_n  \equiv {} & \frac{1}{2}\bbP_n \left[ \left\{(1-{\pi}_\Lambda(A-\delta, \bX)){\bb}(A-\delta, \bX) - {\pi}_\Lambda(A, \bX){\bb}(A, \bX)\right\}Y \right], \\
	\text{and} & \\
	\frac{1}{2}{\bD}_n  \equiv {} & \frac{1}{2}\bbP_n \left[ (1-{\pi}_\Lambda(A-\delta, \bX))^2{\bb}(A-\delta, \bX){\bb}(A-\delta, \bX)\trans + {\pi}_\Lambda(A, \bX)^2{\bb}(A, \bX){\bb}(A, \bX)\trans \right].
\end{align*}

\noindent We decompose
	\begin{align}
		\sqrt{n}(\widehat{\btheta}_A - \btheta^*_A) = {} &  \sqrt{n}\widehat{\bD}^{-1}_n\left( \widehat{\bd}_n - \widehat{\bD}_n\bD^{-1}\bd \right) =  \sqrt{n}\widehat{\bD}^{-1}_n\left(\widehat{\bd}_n - \widehat{\bD}_n\btheta^*\right) \nonumber \\
		= {} & \widehat{\bD}^{-1}_n \sqrt{n} \Big( \underbrace{{\bd}_n - {\bD}_n\bD^{-1}\bd }_{\text{\circled{1}}}\Big) + \widehat{\bD}^{-1}_n\sqrt{n}\Big( \underbrace{\left\{\widehat{\bd}_n - \widehat{\bD}_n\btheta^*\right\} - \left\{{\bd}_n  - {\bD}_n\btheta^*\right\}}_{\text{\circled{2}}} \Big). \label{eqn:a_learning_error_decomp}
	\end{align}
We first focus on \circled{1}. We note that $\bbE[{\bd}_n - {\bD}_n\bD^{-1}\bd] = \bd - \bD\btheta^* = \bzero$, so
\begin{align}
	\sqrt{n}\text{\circled{1}} = {} & \sqrt{n}({\bd}_n - {\bD}_n\bD^{-1}\bd) \nonumber \\
  = {} &\bbG_n \left[ (1-\pi_\Lambda(A-\delta, \bX))\bb(A-\delta, \bX) \left\{Y - (1-\pi_\Lambda(A-\delta, \bX))\bb(A-\delta, \bX)\trans \btheta^* \right\}  \right. \nonumber \\
  &	\left.\quad\quad- \pi_\Lambda(A, \bX)\bb(A, \bX)\left\{Y + \pi_\Lambda(A, \bX)\bb(A, \bX)\trans \btheta^*\right\}\right] \nonumber \\
  = {} & O_p(1) \nonumber
\end{align}
by the central limit theorem, using the moment conditions assumed in Section~4 of the main paper.

\noindent We now focus on \circled{2}. Let $\Psi_n(\bdeta, \btheta)$ be
\begin{align*}
	&  \bbP_n\big[ \left\{(1-{\pi}_\Lambda(A-\delta, \bX; \bdeta)){\bb}(A-\delta, \bX) - {\pi}_\Lambda(A, \bX; \bdeta){\bb}(A, \bX)\right\}Y  \\
	& - \left\{(1-{\pi}_\Lambda(A-\delta, \bX; \bdeta))^2{\bb}(A-\delta, \bX){\bb}(A-\delta, \bX)\trans + {\pi}_\Lambda(A, \bX; \bdeta)^2{\bb}(A, \bX){\bb}(A, \bX)\trans\right\}\btheta \big],
\end{align*}
so that $\Psi_n(\bdeta,\btheta)=\bbP_n\psi_A(\bO;\btheta,\bdeta)$. Then by a Taylor expansion we have
\begin{align}
	\sqrt{n}\text{\circled{2}} = {} & \sqrt{n}\left( \left\{\widehat{\bd}_n - \widehat{\bD}_n\btheta^*\right\} - \left\{{\bd}_n  - {\bD}_n\btheta^*\right\} \right) \nonumber \\
	= {} & \sqrt{n}\left(\Psi_n(\widehat{\bdeta}, \btheta^*) - \Psi_n(\bdeta^*, \btheta^*)\right) \nonumber \\
	= {} &  \dot{\Psi}_n(\bdeta^*, \btheta^*)\sqrt{n}(\widehat{\bdeta} - \bdeta^*) + o_p(1),
\end{align}
where
\begin{align*}
	& \dot{\Psi}_n(\bdeta^*, \btheta^*) \\ & = \bbP_n\left[  -\left\{{\bb}(A-\delta, \bX)\frac{\partial {\pi}_\Lambda(A-\delta, \bX; \bdeta)}{\partial \bdeta\trans}\Bigg\vert_{\bdeta = \bdeta^*}  + {\bb}(A, \bX)\frac{\partial {\pi}_\Lambda(A, \bX; \bdeta)}{\partial \bdeta\trans}\Bigg\vert_{\bdeta = \bdeta^*} \right\}Y  \right. \\
	&\quad\quad  + 2(1-{\pi}_\Lambda(A-\delta, \bX; \bdeta^*)) {\bb}(A-\delta, \bX){\bb}(A-\delta, \bX)\trans\btheta^* \frac{\partial {\pi}_\Lambda(A-\delta, \bX; \bdeta)}{\partial \bdeta\trans}\Bigg\vert_{\bdeta = \bdeta^*}    \\
	& \left. \quad\quad\quad - 2{\pi}_\Lambda(A, \bX; \bdeta^*){\bb}(A, \bX){\bb}(A, \bX)\trans\btheta^*  \frac{\partial {\pi}_\Lambda(A, \bX; \bdeta)}{\partial \bdeta\trans}\Bigg\vert_{\bdeta = \bdeta^*}     \right].
\end{align*}
Because of the assumptions on the continuity and differentiability of ${\pi}_\Lambda(a, \bx; \cdot)$, $\{{\pi}_\Lambda(A, \bX; \bdeta): \bdeta \in \calH^*\}$ is Lipschitz and hence Donsker (see example 3.2.12 in \citet{van1996weak}) and its derivative is also Donsker. Thus, $\{ {\Psi}_n\}$ and $\{ \dot{\Psi}_n\}$ are both Donsker. Thus, as $n\rightarrow\infty$,
\begin{align*}
	 \dot{\Psi}_n(\bdeta^*, \btheta^*)  \rightarrow {} & -\bbE\left[ \left\{\bb(A-\delta, \bX) \dot{\pi}_\Lambda(A-\delta, \bX; \bdeta^*)\trans  + \bb(A, \bX)\dot{\pi}_\Lambda(A, \bX; \bdeta^*)\trans  \right\}Y \right] \\
	 & + \bbE\left[  2(1-{\pi}_\Lambda(A-\delta, \bX; \bdeta^*)) \bb(A-\delta, \bX)\bb(A-\delta, \bX)\trans\btheta^* \dot{\pi}_\Lambda(A-\delta, \bX; \bdeta^*)\trans \right] \\
	 & - \bbE\left[  2{\pi}_\Lambda(A, \bX; \bdeta^*)\bb(A, \bX)\bb(A, \bX)\trans\btheta^* \dot{\pi}_\Lambda(A, \bX; \bdeta^*)\trans \right] \\
	 \equiv {} & \bB_{\bdeta}
\end{align*}
with probability 1,
where $\dot{\pi}_\Lambda(A, \bX; \bdeta^*) = \frac{\partial {\pi}_\Lambda(A, \bX; \bdeta)}{\partial \bdeta}\Big\vert_{\bdeta = \bdeta^*}$; note $\bB_{\bdeta}$ is the population derivative of $\bbE\{\psi_A(\bO;\btheta_A^*,\bdeta)\}$ with respect to $\bdeta\trans$ at $\bdeta^*$, as in the main paper. Thus, by condition (C1) and Slutsky's theorem, we have
\begin{align*}
	\sqrt{n}\text{\circled{2}} = {} & -\bB_{\bdeta} J_{\bdeta}^{-1} \bbG_n U_{\bdeta}(\bO; \bdeta^*) + o_p(1)
	= O_p(1).
\end{align*}
Thus, by Slutsky's theorem since as $n\rightarrow\infty$, $\widehat{\bD}_n \rightarrow \bD$  with probability 1, we have
\begin{align*}
	&\sqrt{n}(\widehat{\btheta}_A - \btheta^*_A)  \\  &= \bD^{-1}\bbG_n\left\{   \left[ (1-\pi_\Lambda(A-\delta, \bX))\bb(A-\delta, \bX) \left\{Y - (1-\pi_\Lambda(A-\delta, \bX))\bb(A-\delta, \bX)\trans \btheta^* \right\}  \right. \right.  \\
	&	\left.\quad\quad\quad\quad- \pi_\Lambda(A, \bX)\bb(A, \bX)\left\{Y + \pi_\Lambda(A, \bX)\bb(A, \bX)\trans \btheta^*\right\}\right] \nonumber \\
	& \left. \quad\quad\quad\quad-\bB_{\bdeta} J_{\bdeta}^{-1} U_{\bdeta}(\bO;\bdeta^*)  \right\} + o_p(1).
\end{align*}
Thus, as $n\rightarrow\infty$,
\begin{align*}
	\sqrt{n}(\widehat{\btheta}_A - \btheta^*_A) \xrightarrow{d} N(\bzero, \bD^{-1}\bSigma_A(\btheta^*; \bdeta^*)\bD^{-1}),
\end{align*}
where
\begin{align*}
	\bSigma_A(\btheta^*; \bdeta^*) = {} &  \bbE\Bigg[   \Big\{\left[ (1-\pi_\Lambda(A-\delta, \bX))\bb(A-\delta, \bX) \left\{Y - (1-\pi_\Lambda(A-\delta, \bX))\bb(A-\delta, \bX)\trans \btheta^* \right\}  \right.   \\
	&	\left.\quad\quad- \pi_\Lambda(A, \bX)\bb(A, \bX)\left\{Y + \pi_\Lambda(A, \bX)\bb(A, \bX)\trans \btheta^*\right\}\right] \nonumber \\
	&  \quad\quad-\bB_{\bdeta} J_{\bdeta}^{-1} U_{\bdeta}(\bO; \bdeta^*) \Big\}^{\bigotimes 2}  \Bigg],
\end{align*}
\noindent where for a vector or matrix $\bV$, $\bV^{\bigotimes 2} = \bV\bV\trans$. This establishes the first display of the theorem.

\medskip
\noindent\textbf{Part (ii): asymptotic normality of $\widehat{\btheta}_{A,\mathrm{aug}}$.}
We note that $\widehat{\btheta}_{A,aug} = \widehat{\bD}^{-1}_n\widehat{\bd}_{n,\widehat{aug}}$, where
	\begin{align*}
		\frac{1}{2}\widehat{\bd}_{n,\widehat{aug}} = {} & \frac{1}{2}\bbP_n \left[ (1-{\pi}_\Lambda(A-\delta, \bX; \widehat{\bdeta})){\bb}(A-\delta, \bX)(Y - {m}(A-\delta, \bX;\widehat{\bnu})) \right. \\
		& \left. \quad\quad - {\pi}_\Lambda(A, \bX; \widehat{\bdeta}){\bb}(A, \bX)(Y - {m}(A, \bX; \widehat{\bnu})) \right].
	\end{align*}
	Further, define
	\begin{align*}
		\frac{1}{2}{\bd}_{n,aug}  = {} & \frac{1}{2} \bbP_n \left[ (1-{\pi}_\Lambda(A-\delta, \bX; \bdeta^*)){\bb}(A-\delta, \bX)(Y - {m}(A-\delta, \bX;\bnu^*)) \right. \\
		&\left. \quad\quad - {\pi}_\Lambda(A, \bX; \bdeta^*){\bb}(A, \bX)(Y - {m}(A, \bX; \bnu^*)) \right], \\
				& \text{and} \\
		\frac{1}{2}{\bd}_{aug} = {} & \frac{1}{2}P \left[ (1-{\pi}_\Lambda(A-\delta, \bX; \bdeta^*)){\bb}(A-\delta, \bX)(Y - {m}(A-\delta, \bX;\bnu^*)) \right. \\
		& \left. \quad\quad - {\pi}_\Lambda(A, \bX; \bdeta^*){\bb}(A, \bX)(Y - {m}(A, \bX;\bnu^*)) \right], \\
				& \text{and} \\
		\frac{1}{2}\widehat{\bd}_{n,aug} = {} & \frac{1}{2}\bbP_n \left[ (1-{\pi}_\Lambda(A-\delta, \bX; \widehat{\bdeta})){\bb}(A-\delta, \bX)(Y - {m}(A-\delta, \bX;\bnu^*)) \right. \\
		& \left. \quad\quad - {\pi}_\Lambda(A, \bX; \widehat{\bdeta}){\bb}(A, \bX)(Y - {m}(A, \bX;\bnu^*)) \right].
	\end{align*}

	\noindent We decompose
	\begin{align}
		\sqrt{n}(\widehat{\btheta}_{A,aug} - \btheta^*_A) = {} & \widehat{\bD}^{-1}_n \sqrt{n} \Big( \underbrace{{\bd}_{n,aug} - {\bD}_n\bD^{-1}\bd }_{\text{\circled{1b}}}\Big) \nonumber\\ &+ \widehat{\bD}^{-1}_n\sqrt{n}\Big( \underbrace{\left\{\widehat{\bd}_{n,aug} - \widehat{\bD}_n\btheta_A^*\right\} - \left\{{\bd}_{n,aug}  - {\bD}_n\btheta_A^*\right\}}_{\text{\circled{2b}}} \Big) \nonumber \\
		&+ \widehat{\bD}^{-1}_n\sqrt{n}\Big( \underbrace{\widehat{\bd}_{n,\widehat{aug}}  - \widehat{\bd}_{n,aug}}_{\text{\circled{3b}}}  \Big). \label{eqn:a_learning_aug_error_decomp}
	\end{align}

	\noindent We first focus on \circled{1b}. We note that $\btheta_A^* = \bD^{-1}\bd = \bD^{-1}\bd_{aug}$, since
$P[(1-\pi_\Lambda(A-\delta,\bX))\bb(A-\delta,\bX)m(A-\delta,\bX;\bnu^*)-\pi_\Lambda(A,\bX)\bb(A,\bX)m(A,\bX;\bnu^*)]=\bzero$ by the change of variables $a\mapsto a-\delta$ and the identity $\{1-\pi_\Lambda(a,\bx)\}g_{A\mid\bX}(a+\delta\mid\bx)=\pi_\Lambda(a,\bx)g_{A\mid\bX}(a\mid\bx)$. Hence $\bbE[{\bd}_{n,aug} - {\bD}_n\bD^{-1}\bd] = \bd_{aug} - \bD\btheta_A^* = \bzero$, so $\sqrt{n}\text{\circled{1b}}$ becomes
	\begin{align}
		&\bbG_n \big[ (1-\pi_\Lambda(A-\delta, \bX))\bb(A-\delta, \bX) \big\{Y - {m}(A-\delta, \bX;\bnu^*) \nonumber \\& \qquad \qquad \qquad \qquad - (1-\pi_\Lambda(A-\delta, \bX))\bb(A-\delta, \bX)\trans \btheta_A^* \big\} \nonumber \\
		&	\left.\quad\quad- \pi_\Lambda(A, \bX)\bb(A, \bX)\left\{Y - {m}(A, \bX;\bnu^*) + \pi_\Lambda(A, \bX)\bb(A, \bX)\trans \btheta_A^*\right\}\right] = O_p(1) \nonumber
	\end{align}
	by the central limit theorem.

    \noindent Next, we focus on \circled{2b}. Let $\Psi_n(\bdeta, \bnu, \btheta)$ be
\begin{align*}
	  &  \bbP_n\big[ \big\{(1-{\pi}_\Lambda(A-\delta, \bX; \bdeta)){\bb}(A-\delta, \bX)(Y - m(A-\delta,\bX; \bnu)) \\& \qquad \qquad \qquad \qquad - {\pi}_\Lambda(A, \bX; \bdeta){\bb}(A, \bX)(Y - m(A,\bX; \bnu)) \big\} \\
	& - \left\{(1-{\pi}_\Lambda(A-\delta, \bX; \bdeta))^2{\bb}(A-\delta, \bX){\bb}(A-\delta, \bX)\trans + {\pi}_\Lambda(A, \bX; \bdeta)^2{\bb}(A, \bX){\bb}(A, \bX)\trans\right\}\btheta \big].
\end{align*}
Then, since \circled{2b} varies only the propensity parameter (the outcome-model parameter is fixed at $\bnu^*$), a Taylor expansion gives
\begin{align}
	\sqrt{n}\text{\circled{2b}} = {} & \sqrt{n}\left( \left\{\widehat{\bd}_{n,aug} - \widehat{\bD}_n\btheta_A^*\right\} - \left\{{\bd}_{n,aug}  - {\bD}_n\btheta_A^*\right\} \right) \nonumber \\
	= {} & \sqrt{n}\left(\Psi_n(\widehat{\bdeta}, \bnu^*, \btheta^*_A) - \Psi_n(\bdeta^*, \bnu^*, \btheta^*_A)\right) \nonumber \\
	= {} &  \dot{\Psi}_{n,\bdeta}(\bdeta^*, \bnu^*, \btheta^*_A)\sqrt{n}(\widehat{\bdeta} - \bdeta^*) + o_p(1),
\end{align}
where $\dot{\Psi}_{n,\bdeta}(\bdeta^*, \bnu^*, \btheta^*_A)$ is
\begin{align*}
	& \bbP_n\Bigg[  -\bigg\{{\bb}(A-\delta, \bX)(Y - m(A-\delta,\bX; \bnu^*))\frac{\partial {\pi}_\Lambda(A-\delta, \bX; \bdeta)}{\partial \bdeta\trans}\bigg\vert_{\bdeta = \bdeta^*}\\ & \qquad \qquad \qquad \qquad \qquad + {\bb}(A, \bX)(Y - m(A,\bX; \bnu^*))\frac{\partial {\pi}_\Lambda(A, \bX; \bdeta)}{\partial \bdeta\trans}\Bigg\vert_{\bdeta = \bdeta^*} \bigg\}   \\
	&\quad\quad  + 2(1-{\pi}_\Lambda(A-\delta, \bX; \bdeta^*)) {\bb}(A-\delta, \bX){\bb}(A-\delta, \bX)\trans\btheta^*_A \frac{\partial {\pi}_\Lambda(A-\delta, \bX; \bdeta)}{\partial \bdeta\trans}\bigg\vert_{\bdeta = \bdeta^*}    \\
	& \left. \quad\quad - 2{\pi}_\Lambda(A, \bX; \bdeta^*){\bb}(A, \bX){\bb}(A, \bX)\trans\btheta^*_A  \frac{\partial {\pi}_\Lambda(A, \bX; \bdeta)}{\partial \bdeta\trans}\bigg\vert_{\bdeta = \bdeta^*}     \right].
\end{align*}
Because of the assumptions on the continuity and differentiability of ${\pi}_\Lambda(a, \bx; \cdot)$, $\{{\pi}_\Lambda(A, \bX; \bdeta): \bdeta \in \calH^*\}$ is Lipschitz and hence Donsker (see example 3.2.12 in \citet{van1996weak}) and its derivative is also Donsker. Thus, as $n\rightarrow\infty$, $\dot{\Psi}_{n,\bdeta}(\bdeta^*,\bnu^*, \btheta^*_A) \rightarrow \bB_{\bdeta,\mathrm{aug}}$ with probability 1, the population derivative of $\bbE\{\psi_{A,\mathrm{aug}}(\bO;\btheta_A^*,\bdeta,\bnu^*)\}$ with respect to $\bdeta\trans$ at $\bdeta^*$. Thus, by condition (C1) and Slutsky's theorem,
\begin{align*}
	\sqrt{n}\text{\circled{2b}} = {} & -\bB_{\bdeta,\mathrm{aug}} J_{\bdeta}^{-1} \bbG_n U_{\bdeta}(\bO; \bdeta^*) + o_p(1) = O_p(1).
\end{align*}

\noindent	Finally, we focus on the term \circled{3b}.  Let
	\begin{align*}
		\Psi_{n,m}(\bnu, \bdeta) ={} & \bbP_n\left[ {\pi}_\Lambda(A, \bX; {\bdeta}){\bb}(A, \bX){m}(A, \bX; {\bnu})  \right. \\
		&\quad\quad\quad - \left.  (1-{\pi}_\Lambda(A-\delta, \bX;{\bdeta})){\bb}(A-\delta, \bX) {m}(A-\delta, \bX; {\bnu}) \right].
	\end{align*}
	Then by a Taylor expansion and an application of Slutsky's Theorem, we have
	\begin{align*}
		\sqrt{n}\circled{3b} = {} & \sqrt{n}\left( \widehat{\bd}_{n,\widehat{aug}}  - \widehat{\bd}_{n,aug} \right) \\
		= {} & \sqrt{n} \bbP_n\left[ {\pi}_\Lambda(A, \bX; \widehat{\bdeta}){\bb}(A, \bX)\left\{{m}(A, \bX; \widehat{\bnu}) - {m}(A, \bX; \bnu^*) \right\} \right. \\
		&\quad\quad\quad - \left.  (1-{\pi}_\Lambda(A-\delta, \bX;\widehat{\bdeta})){\bb}(A-\delta, \bX) \left\{{m}(A-\delta, \bX; \widehat{\bnu}) - {m}(A-\delta, \bX; \bnu^*) \right\} \right] \\
		= {} & \sqrt{n} \left(\Psi_{n,m}(\bnuhat, \bdetahat) - \Psi_{n,m}(\bnu^*, \bdetahat) \right) \\
		= {} & \dot{\Psi}_{n,\bnu}(\bnu^*, \bdeta^*)\sqrt{n}(\bnuhat - \bnu^*) + o_p(1),
	\end{align*}
where
\begin{align}\label{Psi_deriv_wrt_nu}
    & \dot{\Psi}_{n,\bnu}(\bnu^*, \bdeta^*) \nonumber \\ & = \bbP_n\bigg[  {\pi}_\Lambda(A, \bX; \bdeta^*){\bb}(A, \bX)\frac{\partial m(A, \bX; \bnu)}{\partial \bnu\trans}\bigg\vert_{\bnu = \bnu^*} \nonumber \\ & \qquad \qquad- (1-{\pi}_\Lambda(A-\delta, \bX; \bdeta^*)){\bb}(A-\delta, \bX)\dfrac{\partial m(A - \delta, \bX; \bnu)}{\partial \bnu\trans}\bigg\vert_{\bnu = \bnu^*} \bigg].
\end{align}
Because of the assumptions on the continuity and differentiability of $m(a, \bx; \cdot)$, $\{m(A, \bX; \bnu): \bnu \in \calH_m^*\}$ is Lipschitz and hence Donsker, and its derivative is also Donsker. Thus, as $n \rightarrow \infty$, $\dot{\Psi}_{n,\bnu}(\bnu^*, \bdeta^*) \rightarrow \bB_{\bnu,\mathrm{aug}}$ with probability 1, the population derivative of $\bbE\{\psi_{A,\mathrm{aug}}(\bO;\btheta_A^*,\bdeta^*,\bnu)\}$ with respect to $\bnu\trans$ at $\bnu^*$. Thus, by condition (C4) and Slutsky's theorem, we have
\begin{align*}
	\sqrt{n}\text{\circled{3b}} = {} & -\bB_{\bnu,\mathrm{aug}} J_{\bnu}^{-1} \bbG_n U_{\bnu}(\bO; \bnu^*) + o_p(1) = O_p(1).
\end{align*}
Thus, since as $n\rightarrow\infty$, $\widehat{\bD}_n \rightarrow \bD$  with probability 1, we have
\begin{align*}
	&\sqrt{n}(\widehat{\btheta}_{A,aug} - \btheta^*_A) \\ &= \bD^{-1}\bbG_n \Big[ \psi_{A,\mathrm{aug}}(\bO;\btheta_A^*,\bdeta^*,\bnu^*)
	-\bB_{\bdeta,\mathrm{aug}} J_{\bdeta}^{-1} U_{\bdeta}(\bO; \bdeta^*) - \bB_{\bnu,\mathrm{aug}} J_{\bnu}^{-1} U_{\bnu}(\bO; \bnu^*)  \Big] + o_p(1).
\end{align*}
Thus, as $n\rightarrow\infty$,
\begin{align*}
	\sqrt{n}(\widehat{\btheta}_{A,aug} - \btheta^*_A) \xrightarrow{d} N(\bzero, \bD^{-1}\bSigma_{A,\mathrm{aug}}(\btheta_A^*; \bdeta^*,\bnu^*)\bD^{-1}),
\end{align*}
\noindent where $\bSigma_{A,\mathrm{aug}}(\btheta_A^*; \bdeta^*,\bnu^*)$ is
\begin{align*}
	 &  {\small{\bbE\Bigg[   \Big\{\left[ (1-\pi_\Lambda(A-\delta, \bX))\bb(A-\delta, \bX) \left\{Y - {m}(A-\delta, \bX;\bnu^*) - (1-\pi_\Lambda(A-\delta, \bX))\bb(A-\delta, \bX)\trans \btheta_A^* \right\}  \right.}}  \\
	&	\left.\quad\quad- \pi_\Lambda(A, \bX)\bb(A, \bX)\left\{Y - {m}(A, \bX;\bnu^*) + \pi_\Lambda(A, \bX)\bb(A, \bX)\trans \btheta_A^*\right\}\right] \nonumber \\
	&  \quad\quad-\bB_{\bdeta,\mathrm{aug}} J_{\bdeta}^{-1} U_{\bdeta}(\bO; \bdeta^*) - \bB_{\bnu,\mathrm{aug}} J_{\bnu}^{-1} U_{\bnu}(\bO; \bnu^*)   \Big\}^{\bigotimes 2}  \Bigg],
\end{align*}
matching the second display of the theorem.

\medskip
\noindent\textbf{Part (iii): variance ordering.}
Write $\phi_A(\bO)$ and $\phi_{A,\mathrm{aug}}(\bO)$ for the influence functions displayed in the theorem, so that $\bSigma_A=\bbE[\phi_A^{\otimes 2}]$ and $\bSigma_{A,\mathrm{aug}}=\bbE[\phi_{A,\mathrm{aug}}^{\otimes 2}]$, with the kernels
\[
\begin{aligned}
\psi_A(\bO)
&=(1-\pi_\Lambda(A-\delta,\bX))\,\bb(A-\delta,\bX)\Big\{Y-(1-\pi_\Lambda(A-\delta,\bX))\,\bb(A-\delta,\bX)^\top\btheta^*\Big\}\\[-2pt]
&\hspace{1.8em}-\;\pi_\Lambda(A,\bX)\,\bb(A,\bX)\Big\{Y+\pi_\Lambda(A,\bX)\,\bb(A,\bX)^\top\btheta^*\Big\},\\
\psi_{A,aug}(\bO)
&=(1-\pi_\Lambda(A-\delta,\bX))\,\bb(A-\delta,\bX)\Big\{Y-m(A-\delta,\bX;\bnu^*)\\[-2pt]
&\qquad \qquad \qquad \qquad \qquad \qquad \qquad \qquad-(1-\pi_\Lambda(A-\delta,\bX))\,\bb(A-\delta,\bX)^\top\btheta^*_A\Big\}\\[-2pt]
&\hspace{1.8em}-\;\pi_\Lambda(A,\bX)\,\bb(A,\bX)\Big\{Y-m(A,\bX;\bnu^*)+\pi_\Lambda(A,\bX)\,\bb(A,\bX)^\top\btheta^*_A\Big\}.
\end{aligned}
\]
First, observe the exact identity (obtained by replacing $Y$ with $Y-m$ in $\psi_A$)
\begin{align*}
    \psi_A(\bO)&=\psi_{A,aug}(\bO)\\
&+\underbrace{\big[(1-\pi_\Lambda(A-\delta,\bX))\,\bb(A-\delta,\bX)\,m(A-\delta,\bX;\bnu^*)-\pi_\Lambda(A,\bX)\,\bb(A,\bX)\,m(A,\bX;\bnu^*)\big]}_{=:D(A,\bX)}.
\end{align*}
Consequently,
\[
\phi_A=\phi_{A,aug}+ \Big\{D(A,\bX)+\bB_{\bnu,\mathrm{aug}}\,J_{\bnu}^{-1}U_{\bnu}(\bO;\bnu^*)\Big\}
=:\phi_{A,aug}+R.
\]
Using $\bSigma=\bbE[\phi\,\phi^\top]$, we therefore obtain
\[
\bSigma_A(\btheta^*;\bdeta^*)-\bSigma_{A,aug}(\btheta^*_A;\bdeta^*,\bnu^*)
=\bbE\!\Big[\phi_{A,aug}\,R^\top + R\,\phi_{A,aug}^\top + R\,R^\top\Big].
\]
We now show that the two cross terms vanish. First, when the nuisance-score
coefficients are their population $L_2$-projection coefficients,
\[
\bbE\!\Big[\phi_{A,aug}\,U_{\bnu}(\bO;\bnu^*)^\top\Big]=\bzero
\text{ and }
\bbE\!\Big[\phi_{A,aug}\,U_{\bdeta}(\bO;\bdeta^*)^\top\Big]=\bzero.
\]
Second, by the additional orthogonality condition assumed in the theorem,
\[
\bbE\!\Big[\phi_{A,aug}\,D(A,\bX)^\top\Big]=\bzero.
\]
Combining these two orthogonality relations yields $\bbE[\phi_{A,aug}\,R^\top]=\bzero$; by symmetry also $\bbE[R\,\phi_{A,aug}^\top]=\bzero$.
Therefore,
\[
\bSigma_A(\btheta^*;\bdeta^*)-\bSigma_{A,aug}(\btheta^*_A;\bdeta^*,\bnu^*)
=\bbE\!\Big[R\,R^\top\Big],
\]
which is positive semidefinite. This proves $\bSigma_{A,aug}\preceq \bSigma_A$.
Equality holds if and only if $R\equiv 0$ almost surely, that is,
\begin{align*}
    &(1-\pi_\Lambda(A-\delta,\bX))\,\bb(A-\delta,\bX)\,m(A-\delta,\bX;\bnu^*)
-\pi_\Lambda(A,\bX)\,\bb(A,\bX)\,m(A,\bX;\bnu^*)
\\  &\qquad \qquad +\bB_{\bnu,\mathrm{aug}}\,J_{\bnu}^{-1}U_{\bnu}(\bO;\bnu^*)\equiv 0,
\end{align*}
which is equivalent to $\phi_{A,\mathrm{aug}}=\phi_A$ almost surely, the equality condition in the theorem.
\end{proof}

\begin{remark}
{The following remark is restated from the main paper's Section~4 (where it now appears in abbreviated form).} If \(\widehat{\boldeta}_n\) is an efficient estimator of \(\boldeta^*\),
the general sandwich variance admits the usual projection simplification $\bSigma_A
=
\widetilde\bSigma_A
-
B_{\boldeta}V_{\boldeta}B_{\boldeta}^{\top},
\text{ and }
\widetilde\bSigma_A
=
\mathbb E\!\left[
\psi_A(\bO;\btheta_A^*,\boldeta^*)^{\otimes2}
\right].$
An analogous expression holds for the augmented estimator.
\end{remark}

\begin{proof}[Proof of the Remark following Theorem~2]
We first treat the unaugmented estimator.
Define $\bthetatilde_A = {\bD}^{-1}_n\bd_n$ to be the estimator using the true weights. We decompose
\begin{align}
	\sqrt{n}(\widehat{\btheta}_A - \btheta^*_A)
	= {} & \widehat{\bD}^{-1}_n \sqrt{n} \Big( {\bd}_n - {\bD}_n\bD^{-1}\bd\Big) + \widehat{\bD}^{-1}_n\sqrt{n}\Big( \left\{\widehat{\bd}_n - \widehat{\bD}_n\btheta^*\right\} - \left\{{\bd}_n  - {\bD}_n\btheta^*\right\} \Big) \nonumber \\
	= {} & \sqrt{n}(\bthetatilde_A - \btheta^*_A) - {\bD}^{-1}\bB_{\bdeta}\sqrt{n}(\widehat{\bdeta} - \bdeta^*) + o_p(1). \nonumber
\end{align}
Thus, if $\widehat{\bdeta}$ is efficient and since $\sqrt{n}(\bthetatilde_A - \btheta^*_A)$ and $\sqrt{n}(\widehat{\bdeta} - \bdeta^*)$ are jointly asymptotically normal, by the results of \citet{pierce1982asymptotic}, as $n\rightarrow\infty$
\begin{equation*}
	\sqrt{n}(\widehat{\btheta}_A - \btheta^*_A)  \xrightarrow{d} N(\bzero, \bD^{-1}\left(\bSigmatilde_A(\btheta^*; \bdeta^*) - \bB_{\bdeta} \bV_{\bdeta} \bB_{\bdeta}\trans\right)\bD^{-1}),
\end{equation*}
where
\begin{align*}
	\bSigmatilde_A(\btheta^*; \bdeta^*) = {} &  \bbE\Big[   \big\{ (1-\pi_\Lambda(A-\delta, \bX))\bb(A-\delta, \bX) \left\{Y - (1-\pi_\Lambda(A-\delta, \bX))\bb(A-\delta, \bX)\trans \btheta^* \right\}    \\
	&	\quad\quad- \pi_\Lambda(A, \bX)\bb(A, \bX)\left\{Y + \pi_\Lambda(A, \bX)\bb(A, \bX)\trans \btheta^*\right\}\big\}^{\bigotimes 2}  \Big].
\end{align*}
Now we treat the augmented estimator.
Define $\bthetatilde_{A,aug} = {\bD}^{-1}_n{\bd}_{n,aug}$, the estimator using the true propensity $\pi_\Lambda(\cdot;\bdeta^*)$ but the estimated $\widehat{\bnu}$. By the arguments in Part (ii) of the proof of the theorem,
\begin{align*}
\sqrt{n}(\widehat{\btheta}_{A,aug} - \btheta^*_A)
= \sqrt{n}(\bthetatilde_{A,aug} - \btheta^*_A) - {\bD}^{-1}\bB_{\bdeta,\mathrm{aug}}\sqrt{n}(\widehat{\bdeta} - \bdeta^*) + o_p(1),
\end{align*}
and the two terms are jointly asymptotically normal. Thus, if $\widehat{\bdeta}$ is efficient, by the results of \citet{pierce1982asymptotic}, as $n\rightarrow\infty$,
\begin{equation*}
\sqrt{n}(\widehat{\btheta}_{A,aug} - \btheta^*_A)  \xrightarrow{d} N\big(\bzero,\ \bD^{-1}\big(\bSigmatilde_{A,aug}(\btheta_A^*; \bdeta^*,\bnu^*) - \bB_{\bdeta,\mathrm{aug}} \bV_{\bdeta} \bB_{\bdeta,\mathrm{aug}}\trans\big)\bD^{-1}\big),
\end{equation*}
where
\begin{align*}
\bSigmatilde_{A,aug}(\btheta_A^*; \bdeta^*,\bnu^*) = \bbE\Big[ \big\{ \psi_{A,aug}(\bO;\btheta_A^*,\bdeta^*,\bnu^*) - \bB_{\bnu,\mathrm{aug}}J_{\bnu}^{-1}U_{\bnu}(\bO;\bnu^*) \big\}^{\bigotimes 2} \Big].
\end{align*}
\end{proof}

\noindent
We next translate the asymptotic normality of the conditional nudge estimator
into inference for the plug-in CMTP estimator obtained by averaging the fitted
nudge over the conditional distribution of \(A\mid \bX=\bx\).
\begin{corollary}[Asymptotic normality of plug-in CMTP from A-learning nudge estimator]
\label{cor:cmtp_plugin}
Suppose the conditions of Proposition~1
and conditions {\normalfont(C1)-(C3)} hold, and that the nudge working model is
correctly specified, i.e., \(\tau_\delta(a,\bx)=\bb(a,\bx)^\top\btheta_A^*\) for
all \((a,\bx)\) in the support of the duplicated-arm law, so that
Theorem~1 yields
\(\sqrt{n}(\widehat{\btheta}_A-\btheta_A^*)\rightsquigarrow
N(\bzero,\,\bD^{-1}\bSigma_A\bD^{-1})\).
Fix \(\bx\), assume \(\bbE\{\|\bb(A,\bx)\|\mid\bX=\bx\}<\infty\), and let
\(\bar\bb(\bx):=\bbE\{\bb(A,\bx)\mid \bX=\bx\}\). Define the plug-in CMTP
estimator \(\widehat{\tau}_\delta^{\,\mathrm{plug}}(\bx):=
\widehat{\bar\bb}(\bx)^\top\widehat{\btheta}_A\), where
\(\widehat{\bar\bb}(\bx)\) is any estimator satisfying
\(\|\widehat{\bar\bb}(\bx)-\bar\bb(\bx)\|=o_p(n^{-1/2})\); this rate condition
holds, e.g., when the conditional dose distribution \(g_{A\mid\bX}\) is known by
design, or when \(\bar\bb\) is estimated from an independent auxiliary sample
whose size grows faster than \(n\). Then
\[
\sqrt{n}\big(\widehat{\tau}_\delta^{\,\mathrm{plug}}(\bx)-\tau_\delta(\bx)\big)
\;\rightsquigarrow\;
N\big(0,\;\bar\bb(\bx)^\top\bD^{-1}\bSigma_A\bD^{-1}\bar\bb(\bx)\big).
\]
If instead only \(\|\widehat{\bar\bb}(\bx)-\bar\bb(\bx)\|=o_p(1)\), then
\(\widehat{\tau}_\delta^{\,\mathrm{plug}}(\bx)-\tau_\delta(\bx)
=O_p(n^{-1/2})+O_p\big(\|\widehat{\bar\bb}(\bx)-\bar\bb(\bx)\|\big)\), so the
plug-in estimator inherits the (possibly slower) rate of
\(\widehat{\bar\bb}(\bx)\); in particular, generic nonparametric or cross-fitted
estimators of \(\bar\bb(\bx)\) converge more slowly than \(\sqrt n\), and the
display above then no longer describes the limiting distribution.
\end{corollary}
{Thus, Corollary~\ref{cor:cmtp_plugin} shows that when $\bar\bb(\bx)$ can be estimated faster than root-$n$ in rate, e.g. when the treatment's distribution is known by design, the plug-in CMTP estimator inherits the root-$n$ asymptotic normality of the estimator of the nudge effect coefficients. When $\bar\bb(\bx)$ is instead estimated nonparametrically, the plug-in estimator converges at the slower rate of $\widehat{\bar\bb}(\bx)$ and the stated limiting normal distribution no longer applies. Remark~\ref{rem:cmtp_plugin_projection} describes the same-sample projection implementation we use in practice, which explicitly accounts for the additional averaging step.}

\begin{proof}[Proof of Corollary~\ref{cor:cmtp_plugin} (Asymptotic normality of plug-in CMTP from A-learning nudge estimator)]
Under the conditions of Proposition~1,
the nudge surface \(\tau_\delta(a,\bx)\) is identified and, by
Assumption~3(A),
\(\tau_\delta(\bx)=\bbE\{\tau_\delta(A,\bx)\mid\bX=\bx\}\). Correct
specification of the working model then gives
\[
\tau_\delta(\bx)
=\bbE\{\bb(A,\bx)^\top\btheta_A^*\mid\bX=\bx\}
=\bar\bb(\bx)^\top\btheta_A^*,
\]
the interchange being justified by
\(\bbE\{\|\bb(A,\bx)\|\mid\bX=\bx\}<\infty\). Decompose
\[
\widehat{\tau}_\delta^{\,\mathrm{plug}}(\bx)-\tau_\delta(\bx)
=\bar\bb(\bx)^\top(\widehat{\btheta}_A-\btheta_A^*)
+\{\widehat{\bar\bb}(\bx)-\bar\bb(\bx)\}^\top\btheta_A^*
+\{\widehat{\bar\bb}(\bx)-\bar\bb(\bx)\}^\top(\widehat{\btheta}_A-\btheta_A^*).
\]
By Theorem~1,
\(\widehat{\btheta}_A-\btheta_A^*=O_p(n^{-1/2})\) and, by the continuous
mapping theorem,
\(\sqrt n\,\bar\bb(\bx)^\top(\widehat{\btheta}_A-\btheta_A^*)
\rightsquigarrow N\big(0,\,\bar\bb(\bx)^\top\bD^{-1}\bSigma_A\bD^{-1}\bar\bb(\bx)\big)\).
Under the rate condition
\(\|\widehat{\bar\bb}(\bx)-\bar\bb(\bx)\|=o_p(n^{-1/2})\), the second term is
\(o_p(n^{-1/2})\,\|\btheta_A^*\|=o_p(n^{-1/2})\) and the third is
\(o_p(n^{-1/2})\,O_p(n^{-1/2})=o_p(n^{-1/2})\); Slutsky's theorem gives the
first display of the corollary. If only
\(\|\widehat{\bar\bb}(\bx)-\bar\bb(\bx)\|=o_p(1)\), the same decomposition
yields
\(\widehat{\tau}_\delta^{\,\mathrm{plug}}(\bx)-\tau_\delta(\bx)
=O_p(n^{-1/2})+O_p(\|\widehat{\bar\bb}(\bx)-\bar\bb(\bx)\|)\).
\end{proof}

\begin{remark}[Same-sample projection implementation]\label{rem:cmtp_plugin_projection}
In practice (Section~2 of the main paper) we project the fitted nudge onto a
basis \(\bc(\bx)\) using the same sample, i.e.,
\(\widehat\bbeta^{\mathrm{plug}}
=\{\bbP_n \bc_i\bc_i^\top\}^{-1}\bbP_n \bc_i\,\bb(A_i,\bX_i)^\top\widehat\btheta_A\),
which targets \(\bbeta^*=\bM\btheta_A^*\) with
\(\bM:=\{\bbE(\bc\bc^\top)\}^{-1}\bbE\{\bc\,\bb(A,\bX)^\top\}\); under the
conditions of Corollary~\ref{cor:cmtp_plugin}, \(\bbeta^*\) is the
\(L_2(P_{\bX})\) projection coefficient of \(\tau_\delta(\bX)\) onto
\(\bc(\bX)\). Averaging the fitted nudge over the same sample contributes an
additional empirical-process term to the influence function: under
{\normalfont(C1)-(C3)}, finite fourth moments of \((\bb,\bc)\), and
nonsingular \(\bbE(\bc\bc^\top)\),
\(\sqrt n(\widehat\bbeta^{\mathrm{plug}}-\bbeta^*)\rightsquigarrow
N\big(\bzero,\ \bbE\{h(\bO)^{\otimes2}\}\big)\), where
\[
h(\bO)=\bM\,\bD^{-1}\phi_A(\bO)
+\{\bbE(\bc\bc^\top)\}^{-1}\bc(\bX)
\big\{\bb(A,\bX)^\top\btheta_A^*-\bc(\bX)^\top\bbeta^*\big\}.
\]
The second term is not asymptotically negligible in general and is correlated
with the first; ignoring it (as a naive application of
Corollary~\ref{cor:cmtp_plugin} would) misstates the asymptotic variance.
\end{remark}

\begin{lemma}[Population characterization of the $\pi_\Lambda$-stabilized direct CMTP
A-learning minimizer]\label{lem:cmtp_projection}
Suppose Assumptions~1, 2, and
3(A$'$) of the main paper hold, $\pi_\Lambda(\bx,a)\in(0,1)$, and the
propensity model is correctly specified, $\pi_\Lambda(\cdot)=\pi_\Lambda(\cdot;\bdeta^*)$.
Write $\pi_i^\lambda=\pi_\Lambda(\bZ_i^\lambda;\bdeta^*)$,
$g_i^\lambda=\lambda-\pi_i^\lambda$, $m_i^\lambda(\bnu)=m(\bZ_i^\lambda;\bnu)$, and
$\bb_i=b(\bX_i)$, where $\bZ_i^\lambda=(A_i-\lambda\delta,\bX_i)$, $\lambda\in\{0,1\}$.
Then, for any measurable working outcome model $m$ and any $\bnu$, almost surely,
\begin{equation}\label{eq:cmtp_proj_identities}
\bbE\Bigl[\sum_{\lambda\in\{0,1\}}\frac{(g_i^\lambda)^2}{\pi_i^\lambda}\Bigm|\bX_i\Bigr]=1,
\qquad
\bbE\Bigl[\sum_{\lambda\in\{0,1\}}\frac{g_i^\lambda}{\pi_i^\lambda}
\bigl\{Y_i-m_i^\lambda(\bnu)\bigr\}\Bigm|\bX_i\Bigr]
=\tau_\delta(\bX_i).
\end{equation}
Consequently, with
$\psi_{C,\mathrm{aug}}(\bO_i;\btheta,\bdeta,\bnu)
=\sum_{\lambda\in\{0,1\}}\{g_i^\lambda(\bdeta)/\pi_i^\lambda(\bdeta)\}\,\bb_i
\{Y_i-m_i^\lambda(\bnu)-g_i^\lambda(\bdeta)\bb_i^\top\btheta\}$ as in the main paper,
\[
\bbE\bigl\{\psi_{C,\mathrm{aug}}(\bO_i;\btheta,\bdeta^*,\bnu)\bigr\}
=
\bbE\bigl[b(\bX)\bigl\{\tau_\delta(\bX)-b(\bX)^\top\btheta\bigr\}\bigr]
\qquad\text{for every }\btheta\text{ and every }\bnu,
\]
so that, provided $\bD_C:=\bbE\{b(\bX)b(\bX)^\top\}$ is nonsingular, the population
estimating equation has the unique solution
\[
\btheta_C^*
=
\bbE\{b(\bX)b(\bX)^\top\}^{-1}\bbE\{b(\bX)\tau_\delta(\bX)\}
=
\argmin_{\btheta}\bbE\bigl[\{\tau_\delta(\bX)-b(\bX)^\top\btheta\}^2\bigr],
\]
which does not depend on $(m,\bnu)$; in particular
$\bB_{C,\bnu}
=\partial\,\bbE\{\psi_{C,\mathrm{aug}}(\bO_i;\btheta_C^*,\bdeta^*,\bnu)\}
/\partial\bnu^\top\vert_{\bnu=\bnu^*}=\bzero$.
\end{lemma}

\begin{proof}
The definition of $\pi_\Lambda$ gives, for all $(u,\bx)$, the natural-law identity
\begin{equation}\label{eq:natural_transfer}
\pi_\Lambda(\bx,u)\,g_{A\mid\bX}(u\mid\bx)
=\{1-\pi_\Lambda(\bx,u)\}\,g_{A\mid\bX}(u+\delta\mid\bx),
\end{equation}
which is the natural-law counterpart of the mixture identity
$\{1-\pi_\Lambda(\bx,a)\}f_{A_\Lambda\mid\bX}(a\mid\bx)=\tfrac12 g_{A\mid\bX}(a\mid\bx)$
used in the proof of Proposition~5. For the $\lambda=0$ row,
$g_i^0/\pi_i^0=-1$ and $(g_i^0)^2/\pi_i^0=\pi_i^0$ identically. For the $\lambda=1$
row, substituting $u=a-\delta$ and applying \eqref{eq:natural_transfer}, for any
integrable $\omega$,
\begin{align}\label{eq:row1_transfer}
& \bbE\Bigl[\frac{1-\pi_i^1}{\pi_i^1}\,\omega(A_i-\delta,\bX_i)\Bigm|\bX_i=\bx\Bigr]\nonumber\\ &
=\int\frac{1-\pi_\Lambda(\bx,u)}{\pi_\Lambda(\bx,u)}\,\omega(u,\bx)\,
g_{A\mid\bX}(u+\delta\mid\bx)\,\mathrm du
=\bbE\bigl[\omega(A_i,\bX_i)\mid\bX_i=\bx\bigr].
\end{align}
Taking $\omega(u,\bx)=1-\pi_\Lambda(\bx,u)$ in \eqref{eq:row1_transfer} yields
$\bbE[(g_i^1)^2/\pi_i^1\mid\bX_i]=\bbE[1-\pi_i^0\mid\bX_i]$; adding
$\bbE[(g_i^0)^2/\pi_i^0\mid\bX_i]=\bbE[\pi_i^0\mid\bX_i]$ gives the first identity in
\eqref{eq:cmtp_proj_identities}. Next, conditioning on $(A_i,\bX_i)$ so that
$\bbE(Y_i\mid A_i,\bX_i)=m_0(A_i,\bX_i)$, and taking
$\omega(u,\bx)=m_0(u+\delta,\bx)-m(u,\bx;\bnu)$ in \eqref{eq:row1_transfer},
\[
\bbE\Bigl[\frac{g_i^1}{\pi_i^1}\{Y_i-m_i^1(\bnu)\}\Bigm|\bX_i\Bigr]
=\bbE\bigl[m_0(A_i+\delta,\bX_i)-m(A_i,\bX_i;\bnu)\mid\bX_i\bigr],
\]
while the $\lambda=0$ row contributes
$-\bbE[m_0(A_i,\bX_i)-m(A_i,\bX_i;\bnu)\mid\bX_i]$. Summing, the terms in $m$ cancel
for every $\bnu$, and
$\bbE[\sum_\lambda (g_i^\lambda/\pi_i^\lambda)\{Y_i-m_i^\lambda(\bnu)\}\mid\bX_i]
=\bbE[m_0(A_i+\delta,\bX_i)-m_0(A_i,\bX_i)\mid\bX_i]=\tau_\delta(\bX_i)$ by
Proposition~2 of the main paper. The remaining
claims follow because $b(\bX_i)$ is $\bX_i$-measurable and
$\bbE\{\psi_{C,\mathrm{aug}}(\bO_i;\btheta,\bdeta^*,\bnu)\}$ is constant in $\bnu$
(differentiation under the integral is justified by (C6)).
\end{proof}

\begin{proof}[Proof of Theorem~2
(Asymptotic normality of augmented direct CMTP A-learning)]
Recall $\bZ_i^\lambda=(A_i-\lambda\delta,\bX_i)$,
$\pi_i^\lambda(\bdeta)=\pi_\Lambda(\bZ_i^\lambda;\bdeta)$,
$g_i^\lambda(\bdeta)=\lambda-\pi_i^\lambda(\bdeta)$,
$m_i^\lambda(\bnu)=m(\bZ_i^\lambda;\bnu)$, and $\bb_i=b(\bX_i)$; write
$\pi_i^\lambda=\pi_i^\lambda(\bdeta^*)$, $g_i^\lambda=g_i^\lambda(\bdeta^*)$, and
$m_i^\lambda=m_i^\lambda(\bnu^*)$. Since
$\psi_{C,\mathrm{aug}}(\bO_i;\btheta,\bdeta,\bnu)$ is linear in $\btheta$,
$\widehat\btheta_{C,\mathrm{aug}}=\widehat\bD_{C,n}^{-1}\widehat\bd_{C,n,\mathrm{aug}}$,
where
\begin{align*}
\widehat\bd_{C,n,\mathrm{aug}}
&:=\bbP_n\Bigl[\sum_{\lambda\in\{0,1\}}
\frac{g_i^\lambda(\widehat\bdeta_n)}{\pi_i^\lambda(\widehat\bdeta_n)}\,
\bb_i\bigl\{Y_i-m_i^\lambda(\widehat\bnu_n)\bigr\}\Bigr],
&
\widehat\bD_{C,n}
&:=\bbP_n\Bigl[\sum_{\lambda\in\{0,1\}}
\frac{g_i^\lambda(\widehat\bdeta_n)^2}{\pi_i^\lambda(\widehat\bdeta_n)}\,
\bb_i\bb_i^\top\Bigr],
\end{align*}
with true-nuisance analogues
\begin{align*}
\bd_{C,n,\mathrm{aug}}
&:=\bbP_n\Bigl[\sum_{\lambda}\frac{g_i^\lambda}{\pi_i^\lambda}\,
\bb_i\{Y_i-m_i^\lambda\}\Bigr],
&
\bD_{C,n}
&:=\bbP_n\Bigl[\sum_{\lambda}\frac{(g_i^\lambda)^2}{\pi_i^\lambda}\,
\bb_i\bb_i^\top\Bigr],
\end{align*}
and population versions $\bd_{C,\mathrm{aug}}$ and $\bD_C$
obtained by replacing $\bbP_n$ with $P$. By Lemma~\ref{lem:cmtp_projection},
\[
\bD_C=\bbE\{b(\bX)b(\bX)^\top\},
\qquad
\bd_{C,\mathrm{aug}}=\bbE\{b(\bX)\tau_\delta(\bX)\},
\qquad
\btheta_C^*=\bD_C^{-1}\bd_{C,\mathrm{aug}},
\]
where $\bd_{C,\mathrm{aug}}$ does not depend on $\bnu^*$. We decompose
\begin{align}
\sqrt{n}(\widehat{\btheta}_{C,\mathrm{aug}} - \btheta^*_C)
={}& \widehat{\bD}^{-1}_{C,n} \sqrt{n}
\Bigl( \underbrace{\bd_{C,n,\mathrm{aug}} - \bD_{C,n}\btheta_C^*}_{\text{\circled{1}}}\Bigr)
\nonumber\\
&+ \widehat{\bD}^{-1}_{C,n}\sqrt{n}
\Bigl( \underbrace{\bigl\{\widehat\bd_{C,n,\mathrm{aug}}
- \widehat\bD_{C,n}\btheta_C^*\bigr\}
- \bigl\{\bd_{C,n,\mathrm{aug}} - \bD_{C,n}\btheta_C^*\bigr\}}_{\text{\circled{2}}}\Bigr).
\label{eqn:cmtp_aug_error_decomp}
\end{align}

\medskip\noindent
{\bf Step 1: limit of $\sqrt{n}\,\text{\circled{1}}$.}
By Lemma~\ref{lem:cmtp_projection},
$\bbE[\bd_{C,n,\mathrm{aug}}-\bD_{C,n}\btheta_C^*]
=\bd_{C,\mathrm{aug}}-\bD_C\btheta_C^*=\bzero$, and
$\bd_{C,n,\mathrm{aug}}-\bD_{C,n}\btheta_C^*
=\bbP_n\psi_{C,\mathrm{aug}}(\bO_i;\btheta_C^*,\bdeta^*,\bnu^*)$. Hence
\[
\sqrt{n}\,\text{\circled{1}}
=\bbG_n\,\psi_{C,\mathrm{aug}}(\bO_i;\btheta_C^*,\bdeta^*,\bnu^*)
=O_p(1)
\]
by the central limit theorem: since $\pi_\Lambda\ge c_0$, the stabilized weights are
bounded by $1/c_0$, and $\bbE\|b(\bX)\|^2<\infty$, $\bbE(Y^2)<\infty$, and the
boundedness of $m(\cdot;\bnu^*)$ implied by (C6) give
$\bbE\|\psi_{C,\mathrm{aug}}\|^2<\infty$.

\medskip\noindent
{\bf Step 2: limit of $\sqrt{n}\,\text{\circled{2}}$ (effect of nuisance estimation).}
Define
\[
\Psi_{C,n}(\bdeta,\bnu,\btheta)
:=\bbP_n\Bigl[\sum_{\lambda\in\{0,1\}}
\frac{g_i^\lambda(\bdeta)}{\pi_i^\lambda(\bdeta)}\,\bb_i
\bigl\{Y_i-m_i^\lambda(\bnu)-g_i^\lambda(\bdeta)\,\bb_i^\top\btheta\bigr\}\Bigr],
\]
so that
$\text{\circled{2}}
=\Psi_{C,n}(\widehat\bdeta_n,\widehat\bnu_n,\btheta_C^*)
-\Psi_{C,n}(\bdeta^*,\bnu^*,\btheta_C^*)$.
By the smoothness conditions (C3) and (C6), the lower bound
$\pi_\Lambda(\cdot;\bdeta)\ge c_0$ in a neighborhood of $\bdeta^*$, and the Donsker
conditions (C2) and (C5) (the maps
$\bdeta\mapsto g_i^\lambda(\bdeta)/\pi_i^\lambda(\bdeta)$ and their $\bdeta$-derivatives
are Lipschitz transformations of $\pi_\Lambda(\cdot;\bdeta)$, hence Donsker; see
Example~3.2.12 of \citet{van1996weak}), a first-order Taylor expansion in
$(\bdeta,\bnu)$ around $(\bdeta^*,\bnu^*)$ gives
\[
\sqrt{n}\,\text{\circled{2}}
=\dot\Psi_{C,n,\bdeta}\,\sqrt{n}(\widehat\bdeta_n-\bdeta^*)
+\dot\Psi_{C,n,\bnu}\,\sqrt{n}(\widehat\bnu_n-\bnu^*)
+o_p(1),
\]
where the empirical derivatives converge almost surely to the corresponding population
derivatives: $\dot\Psi_{C,n,\bdeta}\to\bB_{C,\bdeta}$ and $\dot\Psi_{C,n,\bnu}\to\bB_{C,\bnu}$.
Because $g_i^0(\bdeta)/\pi_i^0(\bdeta)\equiv-1$, only the shifted-arm row contributes
derivative terms in $\bdeta$ beyond the $\btheta$-part, and direct differentiation gives
\[
\bB_{C,\bdeta}
=
-\bbE\Bigl[\frac{1}{\{\pi_i^1\}^{2}}\,\bb_i
\bigl\{Y_i-m_i^1-(1-\pi_i^1)\,\bb_i^\top\btheta_C^*\bigr\}\,\dot\pi_i^{1\top}\Bigr]
+\bbE\Bigl[\frac{1-\pi_i^1}{\pi_i^1}\,\bb_i\bb_i^\top\btheta_C^*\,
\dot\pi_i^{1\top}\Bigr]
-\bbE\bigl[\bb_i\bb_i^\top\btheta_C^*\,\dot\pi_i^{0\top}\bigr],
\]
where
$\dot\pi_i^{\lambda}
=\partial\pi_\Lambda(\bZ_i^\lambda;\bdeta)/\partial\bdeta\vert_{\bdeta=\bdeta^*}$,
while
\[
\bB_{C,\bnu}
=-\bbE\Bigl[\sum_{\lambda\in\{0,1\}}
\frac{g_i^\lambda}{\pi_i^\lambda}\,\bb_i\,\dot m_i^{\lambda\top}\Bigr]
=-\bbE\bigl[\bb_i\bigl\{\dot m(A_i,\bX_i;\bnu^*)
-\dot m(A_i,\bX_i;\bnu^*)\bigr\}^\top\bigr]
=\bzero,
\]
by identity \eqref{eq:row1_transfer} of Lemma~\ref{lem:cmtp_projection} applied to
$\omega=\dot m(\cdot;\bnu^*)$, where
$\dot m_i^\lambda=\partial m(\bZ_i^\lambda;\bnu)/\partial\bnu\vert_{\bnu=\bnu^*}$: the
augmented direct CMTP score is orthogonal to the outcome-regression nuisance when the
propensity model is correctly specified. Using the asymptotic linearity in (C1) and
(C4) and Slutsky's theorem, $\sqrt{n}\,\text{\circled{2}}$ becomes
\[
-\bB_{C,\bdeta}\,J_{\bdeta}^{-1}\,\bbG_n U_{\bdeta}(\bO_i;\bdeta^*)
-\bB_{C,\bnu}\,J_{\bnu}^{-1}\,\bbG_n U_{\bnu}(\bO_i;\bnu^*)+o_p(1)
=-\bB_{C,\bdeta}\,J_{\bdeta}^{-1}\,\bbG_n U_{\bdeta}(\bO_i;\bdeta^*)+o_p(1).
\]

\medskip\noindent
{\bf Step 3: assembling the pieces.}
By the law of large numbers, uniform consistency of
$\pi_\Lambda(\cdot;\widehat\bdeta_n)$, and $\pi_\Lambda\ge c_0$,
$\widehat\bD_{C,n}\to\bD_C=\bbE\{b(\bX)b(\bX)^\top\}$ almost surely, so
\eqref{eqn:cmtp_aug_error_decomp} yields
\[
\sqrt{n}(\widehat{\btheta}_{C,\mathrm{aug}} - \btheta^*_C)
=\bD_C^{-1}\,\bbG_n\Bigl\{
\psi_{C,\mathrm{aug}}(\bO_i;\btheta_C^*,\bdeta^*,\bnu^*)
-\bB_{C,\bdeta}\,J_{\bdeta}^{-1}U_{\bdeta}(\bO_i;\bdeta^*)
\Bigr\}+o_p(1),
\]
and the multivariate central limit theorem gives
$\sqrt{n}(\widehat{\btheta}_{C,\mathrm{aug}} - \btheta^*_C)
\xrightarrow{d}
N(\bzero,\ \bD_C^{-1}\bSigma_{C,\mathrm{aug}}(\btheta_C^*;\bdeta^*,\bnu^*)\bD_C^{-1})$
with
$\bSigma_{C,\mathrm{aug}}
=\bbE[\{\psi_{C,\mathrm{aug}}(\bO_i;\btheta_C^*,\bdeta^*,\bnu^*)
-\bB_{C,\bdeta}J_{\bdeta}^{-1}U_{\bdeta}(\bO_i;\bdeta^*)\}^{\otimes2}]$,
as stated.

\medskip\noindent
{\bf Step 4: efficient propensity estimator.}
Let $\bthetatilde_C=\bD_{C,n}^{-1}\bd_{C,n,\mathrm{aug}}$ denote the estimator using the
true nuisance values. The decomposition above shows
$\sqrt{n}(\widehat\btheta_{C,\mathrm{aug}}-\btheta_C^*)
=\sqrt{n}(\bthetatilde_C-\btheta_C^*)
-\bD_C^{-1}\bB_{C,\bdeta}\sqrt{n}(\widehat\bdeta_n-\bdeta^*)+o_p(1)$.
If $\widehat\bdeta_n$ is efficient and jointly asymptotically normal with
$\sqrt{n}(\bthetatilde_C-\btheta_C^*)$, then by the variance-reduction results of
\citet{pierce1982asymptotic},
\[
\bSigma_{C,\mathrm{aug}}
=\widetilde\bSigma_{C,\mathrm{aug}}-\bB_{C,\bdeta}\bV_{\bdeta}\bB_{C,\bdeta}^\top,
\qquad
\widetilde\bSigma_{C,\mathrm{aug}}
=\bbE\bigl[\psi_{C,\mathrm{aug}}(\bO_i;\btheta_C^*,\bdeta^*,\bnu^*)^{\otimes2}\bigr].
\]
\end{proof}

\begin{theorem}[Asymptotic normality of logistic A-learning nudge]\label{thm:logistic_A_asymptotic}
Suppose that the outcome $Y\in\{0,1\}$ follows a canonical logistic
model conditional on $(A,\bX)$,
\[
\bbP(Y=1\mid A=a,\bX=\bx) \;=\; \mu\{\eta(a,\bx)\},
\qquad
\mu(u) = \frac{e^u}{1+e^u},
\]
and let $M(y,\eta) = -y\eta + \log(1+e^\eta)$ denote the logistic
negative log-likelihood.
Let $t_\delta(a,\bx)$ be represented in the linear basis
$t_\delta(a,\bx) = \btheta^\top b(a,\bx)$, and let $\btheta_A^*$ denote
the population minimizer of the logistic A-learning risk (for the nudge)
under the duplicated-arm construction, i.e.\ the unique solution of
\[
\Psi_A(\btheta,\bdeta^*) := 
\bbE\bigl[\psi_A(Z;\btheta,\bdeta^*)\bigr] = \bzero,
\]
where {$Z=(Y,A,\bX,\widetilde\Lambda)$ with $\widetilde\Lambda=\Lambda+\tfrac12\in\{0,1\}$ the recoded arm indicator of the main text,} and
\[
\psi_A(Z;\btheta,\bdeta)
=
\{\widetilde\Lambda - \pi_\Lambda(A,\bX;\bdeta)\}\,
\bigl\{
  Y - \mu\bigl( \{\widetilde\Lambda-\pi_\Lambda(A,\bX;\bdeta)\}\, \btheta^\top b(A,\bX) \bigr)
\bigr\}\,
b(A,\bX).
\]

Assume that conditions {\normalfont (C1)--(C3)} hold for the propensity model
$\pi_\Lambda(a,\bx;\bdeta)$, that the parameter $\btheta_A^*$ lies in the
interior of a compact parameter space, that
$\bbE\|b(A,\bX)\|^2<\infty$, and that the matrix
\begin{align*}
    &\bD_{\log}
:=
-\frac{\partial}{\partial\btheta^\top}
\Psi_A(\btheta,\bdeta^*)
\Big|_{\btheta=\btheta_A^*}
\\ &=
\bbE\Bigl[
  \{\Lambda-\pi_\Lambda(A,\bX)\}^2
  \,\mu'\bigl(\{\Lambda-\pi_\Lambda(A,\bX)\}\btheta_A^{*\top}b(A,\bX)\bigr)\,
  b(A,\bX)b(A,\bX)^\top
\Bigr]
\end{align*}
is nonsingular, where $\mu'(u)=\mu(u)\{1-\mu(u)\}$.
Let $\widehat{\bdeta}_n$ be an estimator of $\bdeta^*$ satisfying {\normalfont (C1)}, and let $\widehat{\btheta}_{A,\log}$ solve the empirical A-learning estimating equation
\[
\Psi_{A,n}(\btheta,\widehat{\bdeta}_n)
:=
\bbP_n\bigl[\psi_A(Z;\btheta,\widehat{\bdeta}_n)\bigr]
=
\bzero,
\]
where $\bbP_n$ denotes the empirical measure.
Then, as $n\to\infty$,
\[
\sqrt{n}\bigl(\widehat{\btheta}_{A,\log} - \btheta_A^*\bigr)
\;\xrightarrow{d}\;
N\bigl(\bzero,\,
  \bD_{\log}^{-1}\,\bSigma_{A,\log}(\btheta_A^*;\bdeta^*)\,\bD_{\log}^{-1}
\bigr),
\]
where
\[
\bSigma_{A,\log}(\btheta_A^*;\bdeta^*)
:=
\bbE\Bigl[
  \bigl\{
    \psi_A(Z;\btheta_A^*,\bdeta^*)
    - \bB_{\log}\,J_{\bdeta}^{-1} U_{\bdeta}(\bO;\bdeta^*)
  \bigr\}^{\otimes 2}
\Bigr],
\]
and
\[
\bB_{\log}
:=
\frac{\partial}{\partial\bdeta^\top}
\Psi_A(\btheta_A^*,\bdeta)
\Big|_{\bdeta=\bdeta^*}
\]
is the Jacobian of the population estimating equation with respect to
$\bdeta$. Here $U_{\bdeta}(\bO;\bdeta)$ and $J_{\bdeta}$ are as in
condition {\normalfont (C1)}, and for a vector or matrix $\bV$ we write
$\bV^{\otimes 2}=\bV\bV^\top$.

If, in addition, $\widehat{\bdeta}_n$ is an efficient estimator of
$\bdeta^*$ with asymptotic variance $\bV_{\bdeta}$, then
\[
\bSigma_{A,\log}(\btheta_A^*;\bdeta^*)
=
\widetilde{\bSigma}_{A,\log}(\btheta_A^*;\bdeta^*)
-
\bB_{\log}\,\bV_{\bdeta}\,\bB_{\log}^\top,
\]
where
\[
\widetilde{\bSigma}_{A,\log}(\btheta_A^*;\bdeta^*)
:=
\bbE\bigl[
  \psi_A(Z;\btheta_A^*,\bdeta^*)^{\otimes 2}
\bigr]
\]
is the covariance of the logistic A-learning influence function when the true
$\pi_\Lambda(A,\bX)$ is known.
\end{theorem}
\begin{proof}[Proof of Theorem~\ref{thm:logistic_A_asymptotic}]
Recall that the logistic A--learning estimator $\widehat{\btheta}_{A,\log}$ is defined
as the solution to the empirical estimating equation
\[
\Psi_{A,n}(\btheta,\widehat{\bdeta}_n)
:=
\bbP_n\bigl[
  \psi_A(Z;\btheta,\widehat{\bdeta}_n)
\bigr]
=
\bzero,
\]
where {$Z=(Y,A,\bX,\widetilde\Lambda)$}, $\bbP_n$ is the empirical measure, and
\[
\psi_A(Z;\btheta,\bdeta)
=
\{\widetilde\Lambda - \pi_\Lambda(A,\bX;\bdeta)\}\,
\bigl\{
  Y - \mu\bigl( \{\widetilde\Lambda-\pi_\Lambda(A,\bX;\bdeta)\}\,\btheta^\top b(A,\bX) \bigr)
\bigr\}\,
b(A,\bX),
\]
with $\mu(u) = e^u/(1+e^u)$ the logistic mean function. The population
counterpart is
\[
\Psi_A(\btheta,\bdeta)
:=
\bbE\bigl[\psi_A(Z;\btheta,\bdeta)\bigr],
\]
and, by definition, $\btheta_A^*$ satisfies
$\Psi_A(\btheta_A^*,\bdeta^*) = \bzero$.

We first establish the asymptotic linear representation. Consider the
map
\[
(\btheta,\bdeta)
\;\longmapsto\;
\Psi_{A,n}(\btheta,\bdeta)
=
\bbP_n\bigl[\psi_A(Z;\btheta,\bdeta)\bigr].
\]
By the conditions (C1)--(C3), the class
$\{\psi_A(\cdot;\btheta,\bdeta):(\btheta,\bdeta)\in\calT\times\calH^*\}$ is
Donsker, where $\calT$ is a compact neighborhood of $\btheta_A^*$ and
$\calH^*$ is a neighborhood of $\bdeta^*$. Moreover, $\psi_A$ is twice
continuously differentiable in $(\btheta,\bdeta)$ on
$\calT\times\calH^*$ with bounded derivatives, because $\mu(\cdot)$ is smooth and $\pi_\Lambda(a,\bx;\bdeta)$ is twice continuously differentiable with respect to $\bdeta$ by assumption. Thus, we can apply
a first–order Taylor expansion of $\Psi_{A,n}$ around $(\btheta_A^*,\bdeta^*)$.

Since $\widehat{\btheta}_{A,\log}$ solves
$\Psi_{A,n}(\widehat{\btheta}_{A,\log},\widehat{\bdeta}_n) = \bzero$, we
have
\begin{align*}
\bzero
&=
\Psi_{A,n}(\widehat{\btheta}_{A,\log},\widehat{\bdeta}_n)
\\
&=
\Psi_{A,n}(\btheta_A^*,\bdeta^*)
+ \dot{\Psi}_{A,n,\btheta}(\btheta_A^*,\bdeta^*)\,
  (\widehat{\btheta}_{A,\log}-\btheta_A^*)
+ \dot{\Psi}_{A,n,\bdeta}(\btheta_A^*,\bdeta^*)\,
  (\widehat{\bdeta}_n-\bdeta^*)
+ r_n,
\end{align*}
where
\[
\dot{\Psi}_{A,n,\btheta}(\btheta_A^*,\bdeta^*)
=
\frac{\partial}{\partial\btheta^\top}
\Psi_{A,n}(\btheta,\bdeta)
\Big|_{(\btheta,\bdeta)=(\btheta_A^*,\bdeta^*)},
\qquad
\dot{\Psi}_{A,n,\bdeta}(\btheta_A^*,\bdeta^*)
=
\frac{\partial}{\partial\bdeta^\top}
\Psi_{A,n}(\btheta,\bdeta)
\Big|_{(\btheta,\bdeta)=(\btheta_A^*,\bdeta^*)},
\]
and the remainder $r_n$ satisfies
$\|r_n\| = o_p(n^{-1/2})$ by the differentiability and Donsker
conditions (see, e.g., Theorem 5.23 of \cite{vandervaart1998asymptotic}). Rearranging terms and multiplying by $\sqrt{n}$ yields
\begin{align}
\sqrt{n}(\widehat{\btheta}_{A,\log}-\btheta_A^*)
&=
-\,\dot{\Psi}_{A,n,\btheta}(\btheta_A^*,\bdeta^*)^{-1}
\bigl\{
  \sqrt{n}\,\Psi_{A,n}(\btheta_A^*,\bdeta^*)
  +
  \dot{\Psi}_{A,n,\bdeta}(\btheta_A^*,\bdeta^*)\,
  \sqrt{n}(\widehat{\bdeta}_n-\bdeta^*)
\bigr\}
+ o_p(1).
\label{eq:logA-main-expansion}
\end{align}
We now identify the limits of each ingredient. First, note that
\[
\Psi_{A,n}(\btheta_A^*,\bdeta^*)
=
\bbP_n\bigl[\psi_A(Z;\btheta_A^*,\bdeta^*)\bigr]
=
P\bigl[\psi_A(Z;\btheta_A^*,\bdeta^*)\bigr]
+
\bbG_n\bigl[\psi_A(Z;\btheta_A^*,\bdeta^*)\bigr],
\]
where $P$ denotes the true distribution and $\bbG_n$ the empirical
process. Since $P[\psi_A(Z;\btheta_A^*,\bdeta^*)] = \bzero$ by the
definition of $\btheta_A^*$, we obtain
\[
\sqrt{n}\,\Psi_{A,n}(\btheta_A^*,\bdeta^*)
=
\bbG_n\bigl[\psi_A(Z;\btheta_A^*,\bdeta^*)\bigr]
=
O_p(1)
\]
and, by the central limit theorem for Donsker classes,
\[
\bbG_n\bigl[\psi_A(Z;\btheta_A^*,\bdeta^*)\bigr]
\xrightarrow{d}
N\bigl(\bzero,\,
  \widetilde{\bSigma}_{A,\log}(\btheta_A^*;\bdeta^*)
\bigr),
\]
where
\[
\widetilde{\bSigma}_{A,\log}(\btheta_A^*;\bdeta^*)
:=
\bbE\bigl[
  \psi_A(Z;\btheta_A^*,\bdeta^*)^{\otimes 2}
\bigr].
\]
Next, by (C1),
{\[
\sqrt{n}(\widehat{\bdeta}_n-\bdeta^*)
=
-\,J_{\bdeta}^{-1}
\bbG_n\bigl[U_{\bdeta}(\bO;\bdeta^*)\bigr]
+ o_p(1),
\]
with asymptotic variance $\bV_{\bdeta}$ and influence function
\(-J_{\bdeta}^{-1}U_{\bdeta}(\bO;\bdeta^*)\).} For the derivatives, observe that $\dot{\Psi}_{A,n,\btheta}(\btheta_A^*,\bdeta^*)$ is
\[
\bbP_n\Bigl[
  -\{\widetilde\Lambda-\pi_\Lambda(A,\bX;\bdeta^*)\}^2
   \mu'\bigl(\{\widetilde\Lambda-\pi_\Lambda(A,\bX;\bdeta^*)\}
              \btheta_A^{*\top} b(A,\bX)\bigr)\,
   b(A,\bX)b(A,\bX)^\top
\Bigr]
+ o_p(1),
\]
where $\mu'(u)=\mu(u)\{1-\mu(u)\}$. By the law of large numbers,
\[
\dot{\Psi}_{A,n,\btheta}(\btheta_A^*,\bdeta^*)
\;\xrightarrow{p}\;
-\bD_{\log},
\]
with
\[
\bD_{\log}
=
\bbE\Bigl[
  \{\widetilde\Lambda-\pi_\Lambda(A,\bX;\bdeta^*)\}^2
  \mu'\bigl(\{\widetilde\Lambda-\pi_\Lambda(A,\bX;\bdeta^*)\}
            \btheta_A^{*\top} b(A,\bX)\bigr)\,
  b(A,\bX)b(A,\bX)^\top
\Bigr],
\]
which is nonsingular by assumption. Similarly,
\[
\dot{\Psi}_{A,n,\bdeta}(\btheta_A^*,\bdeta^*)
\;\xrightarrow{p}\;
\bB_{\log},
\]
where
\[
\bB_{\log}
:=
\frac{\partial}{\partial\bdeta^\top}
\Psi_A(\btheta_A^*,\bdeta)
\Big|_{\bdeta=\bdeta^*}
=
\bbE\left[
  \frac{\partial}{\partial\bdeta^\top}\psi_A(Z;\btheta_A^*,\bdeta)
  \Big|_{\bdeta=\bdeta^*}
\right].
\]
Substituting these limits and the expansion for
$\sqrt{n}(\widehat{\bdeta}_n-\bdeta^*)$ into
\eqref{eq:logA-main-expansion}, and applying Slutsky's theorem, we obtain
\begin{align*}
\sqrt{n}(\widehat{\btheta}_{A,\log}-\btheta_A^*)
&=
\bD_{\log}^{-1}
\bbG_n\Bigl[
  \psi_A(Z;\btheta_A^*,\bdeta^*)
  - \bB_{\log}\,J_{\bdeta}^{-1}
    U_{\bdeta}(\bO;\bdeta^*)
\Bigr]
+ o_p(1).
\end{align*}
Therefore,
\[
\sqrt{n}(\widehat{\btheta}_{A,\log}-\btheta_A^*)
\xrightarrow{d}
N\bigl(\bzero,\,
  \bD_{\log}^{-1}\,\bSigma_{A,\log}(\btheta_A^*;\bdeta^*)\,\bD_{\log}^{-1}
\bigr),
\]
with
\[
\bSigma_{A,\log}(\btheta_A^*;\bdeta^*)
=
\bbE\Bigl[
  \bigl\{
    \psi_A(Z;\btheta_A^*,\bdeta^*)
    - \bB_{\log}\,J_{\bdeta}^{-1}
      U_{\bdeta}(\bO;\bdeta^*)
  \bigr\}^{\otimes 2}
\Bigr].
\]

Finally, when $\widehat{\bdeta}_n$ is efficient for $\bdeta^*$ with
asymptotic variance $\bV_{\bdeta}$, the usual variance reduction
argument for two-step estimators (see, e.g., \cite{pierce1982asymptotic})
implies that
\[
\bSigma_{A,\log}(\btheta_A^*;\bdeta^*)
=
\widetilde{\bSigma}_{A,\log}(\btheta_A^*;\bdeta^*)
-
\bB_{\log}\,\bV_{\bdeta}\,\bB_{\log}^\top,
\]
where
$\widetilde{\bSigma}_{A,\log}(\btheta_A^*;\bdeta^*)
:=\bbE[\psi_A(Z;\btheta_A^*,\bdeta^*)^{\otimes 2}]$
is the covariance when $\pi_\Lambda(A,\bX)$ is known. This completes the
proof.
\end{proof}

{The following theorem, relocated from the main paper to conserve space, gives the asymptotic normality result for the logistic direct CMTP A-learning estimator; it is referenced from Section~4 of the main paper.}

\begin{theorem}[Asymptotic normality of logistic direct CMTP A-learning]
\label{thm:logistic_direct_cmtp_asymptotic}
Let $Y\in\{0,1\}$ and write $\mu(v)=\operatorname{expit}(v)$ and
$\mu'(v)=\mu(v)\{1-\mu(v)\}$.
Under the duplicated-arm construction, subject $\bO_i$ contributes the two
rows $\bZ_i^\lambda=(A_i-\lambda\delta,\bX_i)$ with arm indicator
$\widetilde\Lambda_i^\lambda=\lambda\in\{0,1\}$; write
$\pi_i^\lambda(\bdeta):=\pi_\Lambda(\bZ_i^\lambda;\bdeta)$ and
$g_i^\lambda(\bdeta):=\lambda-\pi_i^\lambda(\bdeta)$.
Let $\bb_i:=\bb(\bX_i)$ denote the column vector of basis functions in $\bX$
only, used for the direct CMTP working model $f(\bx)=\btheta^\top\bb(\bx)$.

Let $\btheta_{C,\log}^*$ denote the unique population minimizer of the
$\pi_\Lambda$-stabilized logistic direct CMTP A-learning loss
\[
\ell_{A,\mathrm{CMTP}}^{\mathrm{NLL}}(\btheta)
:=
\bbE\!\left[
\sum_{\lambda\in\{0,1\}}
  \frac{1}{\pi_i^\lambda(\bdeta^*)}
  \left\{
    -Y_i\,g_i^\lambda(\bdeta^*)\,\bb_i^\top\btheta
    +b\!\left(g_i^\lambda(\bdeta^*)\,\bb_i^\top\btheta\right)
  \right\}
\right],
\qquad b(v)=\log(1+e^{v}),
\]
i.e.\ the solution to the population estimating equation
\begin{equation}\label{eq:logCMTP_pop_score}
\Psi_{C,\log}(\btheta,\bdeta^*)
:=
\bbE\!\left[
\sum_{\lambda\in\{0,1\}}
  \frac{g_i^\lambda(\bdeta^*)}{\pi_i^\lambda(\bdeta^*)}\,
  \bb_i
  \left\{
    Y_i-\mu\!\left(g_i^\lambda(\bdeta^*)\,\bb_i^\top\btheta\right)
  \right\}
\right]
=
\bzero.
\end{equation}
Define the bread matrix
\[
\bD_{C,\log}
:=
\bbE\!\left[
\sum_{\lambda\in\{0,1\}}
  \frac{\{g_i^\lambda(\bdeta^*)\}^2}{\pi_i^\lambda(\bdeta^*)}
  \,\mu'\!\left(g_i^\lambda(\bdeta^*)\,\bb_i^\top\btheta_{C,\log}^*\right)
  \bb_i\bb_i^\top
\right],
\]
and the influence function
\[
\varphi_{C,\log}(\bO_i;\btheta_{C,\log}^*,\bdeta^*)
:=
\sum_{\lambda\in\{0,1\}}
\frac{g_i^\lambda(\bdeta^*)}{\pi_i^\lambda(\bdeta^*)}\,
\bb_i
\left\{
  Y_i-\mu\!\left(g_i^\lambda(\bdeta^*)\,\bb_i^\top\btheta_{C,\log}^*\right)
\right\}.
\]
Assume conditions \emph{(C1)--(C3)} for the propensity model
$\pi_\Lambda(a,\bx;\bdeta)$, that $\btheta_{C,\log}^*$ lies in the interior
of a compact parameter space, that $\bbE\|\bb(\bX)\|^2<\infty$,
that $\min\{\pi_i^0(\bdeta^*),\pi_i^1(\bdeta^*)\}\ge c_0>0$ almost surely
for some constant $c_0$, and that $\bD_{C,\log}$ is nonsingular.

Let $\widehat\btheta_{C,\log}$ solve the empirical estimating equation
\[
\Psi_{C,\log,n}(\btheta,\widehat\bdeta_n)
:=
\bbP_n\bigl[\psi_{C,\log}(\bO_i;\btheta,\widehat\bdeta_n)\bigr]
=
\bzero.
\]
Then, as $n\to\infty$,
\[
\sqrt{n}\,\bigl(\widehat\btheta_{C,\log} - \btheta_{C,\log}^*\bigr)
\;\xrightarrow{d}\;
N\!\Bigl(\bzero,\;
\bD_{C,\log}^{-1}\,
\bSigma_{C,\log}(\btheta_{C,\log}^*;\bdeta^*)\,
\bD_{C,\log}^{-1}
\Bigr),
\]
where
\begin{align*}
    &\bSigma_{C,\log}(\btheta_{C,\log}^*;\bdeta^*)
:=
\bbE\!\left[
\Bigl\{
  \varphi_{C,\log}(\bO_i;\btheta_{C,\log}^*,\bdeta^*)
  -\bB_{C,\log}\,J_{\bdeta}^{-1}\,U_{\bdeta}(\bO_i;\bdeta^*)
\Bigr\}^{\otimes 2}
\right],
\\
&\bB_{C,\log}
:=
\frac{\partial}{\partial\bdeta^\top}
\Psi_{C,\log}(\btheta_{C,\log}^*,\bdeta)
\Bigg|_{\bdeta=\bdeta^*},
\end{align*}
with $J_{\bdeta}$ and $U_{\bdeta}(\bO_i;\bdeta)$ as in condition
\emph{(C1)}, and $\bV^{\otimes 2}=\bV\bV^\top$. If, in addition,
$\widehat\bdeta_n$ is an efficient estimator of $\bdeta^*$ with asymptotic
variance $\bV_{\bdeta}$, then
$\bSigma_{C,\log}(\btheta_{C,\log}^*;\bdeta^*)
=\widetilde\bSigma_{C,\log}(\btheta_{C,\log}^*;\bdeta^*)
-\bB_{C,\log}\,\bV_{\bdeta}\,\bB_{C,\log}^\top$,
where
$\widetilde\bSigma_{C,\log}(\btheta_{C,\log}^*;\bdeta^*)
:=\bbE\bigl[\varphi_{C,\log}(\bO_i;\btheta_{C,\log}^*,\bdeta^*)^{\otimes 2}\bigr]$
is the variance when $\pi_\Lambda$ is known.
\end{theorem}

\begin{proof}[Proof of Theorem~\ref{thm:logistic_direct_cmtp_asymptotic}]
The proof follows the M-estimation strategy used in
Theorem~\ref{thm:logistic_A_asymptotic}, with two structural
differences: (i) the basis functions are $\bb(\bX)$ rather than
$\bb(A,\bX)$, so the score does not depend on $A$ beyond what enters
through the centered arm scores $g_i^\lambda$; and (ii) the $\pi_\Lambda^{-1}$ stabilization
factor introduces an additional layer in the propensity-derivative
calculation.

\medskip\noindent
\emph{Step 1: Identification and regularity.}
Write $\pi^\lambda:=\pi_i^\lambda(\bdeta^*)$, $g^\lambda:=g_i^\lambda(\bdeta^*)$,
and $\eta_C^\lambda(\btheta):=g^\lambda\,\bb_i^\top\btheta$.
Consider a generic duplicated row with $(A_\Lambda,\bX)=(a,\bx)$ and write
$\pi^*:=\pi_\Lambda(\bx,a;\bdeta^*)$. Conditional on $(A_\Lambda,\bX)=(a,\bx)$,
the arm indicator satisfies $\widetilde\Lambda=1$ with probability $\pi^*$,
and, by consistency and Assumption~3(A$'$) as in
Proposition~2,
$\bbE[Y\mid \widetilde\Lambda=1,A_\Lambda=a,\bX=\bx]=m_0(a+\delta,\bx)$ and
$\bbE[Y\mid \widetilde\Lambda=0,A_\Lambda=a,\bX=\bx]=m_0(a,\bx)$.
Hence the row-level score
$s(\btheta):=\{(\widetilde\Lambda-\pi^*)/\pi^*\}\,\bb(\bx)
\{Y-\mu((\widetilde\Lambda-\pi^*)\,\bb(\bx)^\top\btheta)\}$ satisfies
\begin{align*}
&\bbE\bigl[s(\btheta)\mid A_\Lambda=a,\bX=\bx\bigr]\\
&=
\{1-\pi^*\}\,\bb(\bx)
\Bigl[
  \bigl\{m_0(a+\delta,\bx)-m_0(a,\bx)\bigr\}
  -\bigl\{\mu\bigl((1-\pi^*)\,\bb(\bx)^\top\btheta\bigr)
          -\mu\bigl(-\pi^*\,\bb(\bx)^\top\btheta\bigr)\bigr\}
\Bigr].
\end{align*}
Averaging over the duplicated-arm mixture law of $A_\Lambda$ given $\bX=\bx$,
with conditional density
$f_{A_\Lambda\mid\bX}(a\mid\bx)
=\tfrac12\{g_{A\mid\bX}(a\mid\bx)+g_{A\mid\bX}(a+\delta\mid\bx)\}$,
and using the identity
$\{1-\pi_\Lambda(\bx,a)\}\,f_{A_\Lambda\mid\bX}(a\mid\bx)
=\tfrac12\,g_{A\mid\bX}(a\mid\bx)$, the per-subject score (which sums the
two rows and thus equals twice the per-row mixture average in expectation)
satisfies
\begin{align*}
& \Psi_{C,\log}(\btheta,\bdeta^*)
=
\bbE\Bigl[
\bb(\bX)\,\bigl\{\tau_\delta(\bX)-\chi(\bX;\btheta)\bigr\}
\Bigr],
\\
& \chi(\bx;\btheta)
:=
\bbE\Bigl[
\mu\bigl(\{1-\pi_\Lambda(\bx,A)\}\,\bb(\bx)^\top\btheta\bigr)
-\mu\bigl(-\pi_\Lambda(\bx,A)\,\bb(\bx)^\top\btheta\bigr)
\,\Big|\,\bX=\bx
\Bigr],
\end{align*}
where the inner expectation is over the natural conditional law of $A$ given
$\bX=\bx$ and $\tau_\delta(\bx)=\bbE[m_0(A+\delta,\bx)-m_0(A,\bx)\mid\bX=\bx]$
by Proposition~2. Thus
\eqref{eq:logCMTP_pop_score} defines a link-scale projection of the CMTP:
if $\chi(\bx;\btheta_0)=\tau_\delta(\bx)$ for some $\btheta_0$ and
$P_{\bX}$-almost every $\bx$, then $\btheta_{C,\log}^*=\btheta_0$; this is
the NLL analogue of Proposition~5 and parallels
Proposition~6. This also establishes the derivations of the direct CMTP
losses for the logistic case promised in Section~\ref{direct_NLL_CMTP}. Uniqueness follows from strict convexity of
the stabilized loss, whose Hessian at any $\btheta$ is
$\bbE[\sum_{\lambda}\{g_i^\lambda\}^2\{\pi_i^\lambda\}^{-1}
\mu'(\eta_C^\lambda(\btheta))\,\bb_i\bb_i^\top]\succ 0$. In particular, the
bread of the M-estimator at $\btheta_{C,\log}^*$ is
\[
\bD_{C,\log}
=
-\frac{\partial}{\partial\btheta^\top}\Psi_{C,\log}(\btheta,\bdeta^*)
\bigg|_{\btheta=\btheta_{C,\log}^*}
=
\bbE\!\left[
\sum_{\lambda\in\{0,1\}}
  \frac{\{g_i^\lambda\}^2}{\pi_i^\lambda}\,
  \mu'\!\bigl(\eta_C^\lambda(\btheta_{C,\log}^*)\bigr)\,
  \bb_i\bb_i^\top
\right],
\]
which is nonsingular by assumption.

\medskip\noindent
\emph{Step 2: Asymptotic linear representation.}
Since $\widehat\btheta_{C,\log}$ solves
$\Psi_{C,\log,n}(\widehat\btheta_{C,\log},\widehat\bdeta_n)=\bzero$,
a first-order Taylor expansion around $(\btheta_{C,\log}^*,\bdeta^*)$ gives
\begin{align}
\bzero
&=
\Psi_{C,\log,n}(\btheta_{C,\log}^*,\bdeta^*)
+
\dot\Psi_{C,\log,n,\btheta}(\btheta_{C,\log}^*,\bdeta^*)
\bigl(\widehat\btheta_{C,\log}-\btheta_{C,\log}^*\bigr) \nonumber
\\& \qquad+
\dot\Psi_{C,\log,n,\bdeta}(\btheta_{C,\log}^*,\bdeta^*)
\bigl(\widehat\bdeta_n-\bdeta^*\bigr)
+ r_n,
\label{eq:logCMTP_Taylor}
\end{align}
where $\|r_n\|=o_p(n^{-1/2})$ by the twice-continuous
differentiability of $\mu$, the compactness of the parameter space,
and the Donsker property of the relevant function classes
(see below). Rearranging and multiplying by $\sqrt{n}$,
\begin{align}\label{eq:logCMTP_linear_rep}
&\sqrt{n}\,\bigl(\widehat\btheta_{C,\log}-\btheta_{C,\log}^*\bigr) \nonumber
\\ &=
-\dot\Psi_{C,\log,n,\btheta}(\btheta_{C,\log}^*,\bdeta^*)^{-1}
\Bigl\{
  \sqrt{n}\,\Psi_{C,\log,n}(\btheta_{C,\log}^*,\bdeta^*)
  +
  \dot\Psi_{C,\log,n,\bdeta}(\btheta_{C,\log}^*,\bdeta^*)
  \sqrt{n}(\widehat\bdeta_n-\bdeta^*)
\Bigr\}
+ o_p(1).
\end{align}

\medskip\noindent
\emph{Step 3: Limits of the derivative terms.}
For the $\btheta$-derivative,
\[
\dot\Psi_{C,\log,n,\btheta}(\btheta_{C,\log}^*,\bdeta^*)
=
-\bbP_n\!\left[
\sum_{\lambda\in\{0,1\}}
  \frac{\{g_i^\lambda\}^2}{\pi_i^\lambda}
  \mu'\!\bigl(\eta_C^\lambda(\btheta_{C,\log}^*)\bigr)
  \bb_i\bb_i^\top
\right]
\;\xrightarrow{p}\;
-\bD_{C,\log}.
\]
This convergence follows from the law of large numbers, using
$\bbE[\sum_\lambda\bb_i\bb_i^\top\{g_i^\lambda\}^2/\pi_i^\lambda\,\mu'(\eta_C^\lambda)] < \infty$
(implied by $\pi_i^\lambda\ge c_0>0$, compactness of the parameter space, and $\mu'\le 1/4$).
Since $\bD_{C,\log}$ is nonsingular, $\dot\Psi_{C,\log,n,\btheta}^{-1}$
is well-defined with probability approaching 1.

For the $\bdeta$-derivative, let $\dot\pi_i^\lambda := \partial\pi_\Lambda(\bZ_i^\lambda;\bdeta)/\partial\bdeta
\big|_{\bdeta=\bdeta^*}$. Using
$\partial g_i^\lambda/\partial\bdeta = -\dot\pi_i^\lambda$ and
$\partial\{g_i^\lambda/\pi_i^\lambda\}/\partial\bdeta
= -\lambda\,\dot\pi_i^\lambda/(\pi_i^\lambda)^2$
(note $g_i^\lambda/\pi_i^\lambda=\lambda/\pi_i^\lambda-1$), the chain rule
gives, for each component $\ell$ of $\bdeta$,
\begin{align*}
&\frac{\partial}{\partial\bdeta_\ell}
\left\{
  \frac{g_i^\lambda}{\pi_i^\lambda}\,\bb_i
  \bigl(Y_i-\mu(\eta_C^\lambda(\btheta_{C,\log}^*))\bigr)
\right\}\\
&=
-\frac{\lambda\,\dot\pi_{i,\ell}^\lambda}{(\pi_i^\lambda)^2}\,
\bb_i\bigl(Y_i-\mu(\eta_C^\lambda(\btheta_{C,\log}^*))\bigr)
+
\frac{g_i^\lambda}{\pi_i^\lambda}\,
\mu'\!\bigl(\eta_C^\lambda(\btheta_{C,\log}^*)\bigr)\,
\dot\pi_{i,\ell}^\lambda\;
\bb_i\,\bb_i^\top\btheta_{C,\log}^*,
\end{align*}
and the first term vanishes for $\lambda=0$ since $g_i^0/\pi_i^0\equiv-1$.
By the Donsker property of $\{\pi_\Lambda(\cdot;\bdeta):\bdeta\in\calH^*\}$
and its derivative (Lipschitz in $\bdeta$ by (C3), hence Donsker by
Example 19.7 of \citet{vandervaart1998asymptotic}), together with the smoothness of
$\mu$ and the moment conditions, the empirical derivative converges
in probability to its population counterpart:
\[
\dot\Psi_{C,\log,n,\bdeta}(\btheta_{C,\log}^*,\bdeta^*)
\;\xrightarrow{p}\;
\bB_{C,\log}
:=
\frac{\partial}{\partial\bdeta^\top}
\Psi_{C,\log}(\btheta_{C,\log}^*,\bdeta)\Big|_{\bdeta=\bdeta^*}.
\]

\medskip\noindent
\emph{Step 4: Empirical process term.}
Since $P[\varphi_{C,\log}(\bO_i;\btheta_{C,\log}^*,\bdeta^*)]=\bzero$ by definition
of $\btheta_{C,\log}^*$, we have
\[
\sqrt{n}\,\Psi_{C,\log,n}(\btheta_{C,\log}^*,\bdeta^*)
=
\bbG_n\bigl[\varphi_{C,\log}(\bO_i;\btheta_{C,\log}^*,\bdeta^*)\bigr]
=
O_p(1).
\]
The class $\{\varphi_{C,\log}(\cdot;\btheta,\bdeta):\btheta\in\calT,
\bdeta\in\calH^*\}$ is Donsker: it is a product of
$\pi^{-1}$-bounded, Lipschitz functions of $\bdeta$ (by (C3) and
$\pi_i^\lambda\ge c_0$), the bounded function $Y-\mu(g\,\bb^\top\btheta)$
(bounded because $\mu\in(0,1)$), and the fixed $\bb(\bX)$; note that
$\bbG_n$ acts on the per-subject summed score
$\varphi_{C,\log}(\bO_i;\cdot)$, which is i.i.d.\ across subjects and
correctly accounts for the fact that a subject's two duplicated rows share $Y_i$.
By the central limit theorem for Donsker classes,
\[
\bbG_n\bigl[\varphi_{C,\log}(\bO_i;\btheta_{C,\log}^*,\bdeta^*)\bigr]
\;\xrightarrow{d}\;
N\!\bigl(\bzero,\,\widetilde\bSigma_{C,\log}(\btheta_{C,\log}^*;\bdeta^*)\bigr),
\]
where
$\widetilde\bSigma_{C,\log}
:=\bbE[\varphi_{C,\log}^{\otimes 2}]$.

\medskip\noindent
\emph{Step 5: Propensity estimation term.}
By condition (C1), the propensity estimator satisfies
\[
\sqrt{n}(\widehat\bdeta_n-\bdeta^*)
=
-\,J_{\bdeta}^{-1}
\bbG_n\bigl[U_{\bdeta}(\bO_i;\bdeta^*)\bigr]
+ o_p(1).
\]

\medskip\noindent
\emph{Step 6: Assembling the influence function.}
Substituting Steps 3--5 into \eqref{eq:logCMTP_linear_rep} and applying
Slutsky's theorem,
\begin{align*}
\sqrt{n}\,\bigl(\widehat\btheta_{C,\log}-\btheta_{C,\log}^*\bigr)
&=
\bD_{C,\log}^{-1}
\bbG_n\!\left[
  \varphi_{C,\log}(\bO_i;\btheta_{C,\log}^*,\bdeta^*)
  - \bB_{C,\log}\,J_{\bdeta}^{-1}
    U_{\bdeta}(\bO_i;\bdeta^*)
\right]
+ o_p(1).
\end{align*}
By the multivariate central limit theorem applied to the joint
empirical process $(\bbG_n\varphi_{C,\log},\,\bbG_n U)$,
\[
\sqrt{n}\,\bigl(\widehat\btheta_{C,\log}-\btheta_{C,\log}^*\bigr)
\;\xrightarrow{d}\;
N\!\Bigl(\bzero,\;
\bD_{C,\log}^{-1}\,
\bSigma_{C,\log}(\btheta_{C,\log}^*;\bdeta^*)\,
\bD_{C,\log}^{-1}
\Bigr),
\]
with
\[
\bSigma_{C,\log}
=
\bbE\!\left[
\Bigl\{
  \varphi_{C,\log}(\bO_i;\btheta_{C,\log}^*,\bdeta^*)
  - \bB_{C,\log}\,J_{\bdeta}^{-1}U_{\bdeta}(\bO_i;\bdeta^*)
\Bigr\}^{\otimes 2}
\right].
\]

\medskip\noindent
\emph{Step 7: Efficient propensity estimator.}
If $\widehat\bdeta_n$ is efficient for $\bdeta^*$ with asymptotic
variance $\bV_{\bdeta}$, and is jointly asymptotically normal with
$\bbG_n\varphi_{C,\log}$, then by the variance-reduction argument of
\citet{pierce1982asymptotic}, the covariance matrix reduces to
\[
\bSigma_{C,\log}(\btheta_{C,\log}^*;\bdeta^*)
=
\widetilde\bSigma_{C,\log}(\btheta_{C,\log}^*;\bdeta^*)
-
\bB_{C,\log}\,\bV_{\bdeta}\,\bB_{C,\log}^\top,
\]
where $\widetilde\bSigma_{C,\log}
:=\bbE[\varphi_{C,\log}^{\otimes 2}]$
is the variance when $\pi_\Lambda(A,\bX)$ is known.
This completes the proof.
\end{proof}

\vspace*{-5 mm}
\setlength{\bibsep}{0pt plus 0.3ex}

\onehalfspacing
{\small
\bibliography{Bibliography}
}

\end{document}